\documentclass{CSML}

\def\dOi{12(2:12)2016}
\lmcsheading%
{\dOi}
{1--43}
{}
{}
{Nov.~19, 2015}
{Jun.~28, 2016}
{}

\ACMCCS{[{\bf Theory of computation}]: Models of
  computation---Computability---Lambda calculus; Semantics and
  Reasoning---Operational semantics; Logic---Proof theory}

\usepackage[utf8]{inputenc}
\usepackage[british]{babel}
\usepackage{amsmath}
\usepackage{amssymb}
\usepackage{mathpartir}
\usepackage{url}
\usepackage{breakurl}
\usepackage[breaklinks]{hyperref}

\usepackage{tikz}
\usetikzlibrary{matrix,arrows}
\tikzset{ % triple line for TikZ
  triple/.style={
    draw=black!75,
    color=black!75,
    thin,
    double distance=4pt,
    -,
    >=stealth},
  thirdline/.style={
    draw=black!75,
    color=black!75,
    thin,
    -,
    >=stealth}
}

\newif\ifrevision   % toggle fixmes
\revisiontrue
\renewcommand{\emph}[1]{\textsl{#1}}       % Avoid confusion with math.
\newcommand{\ie}{\emph{i.e.}}              % Mandatory by LMCS style
\newcommand{\eg}{\emph{e.g.}}
\newcommand{\cf}{\emph{cf.}}

\newcommand{\Term}[1]{\mathbf{#1}}
\newcommand{\I}{\Term{I}}
\newcommand{\K}{\Term{K}}
\newcommand{\DELTA}{\Term{\Delta}}
\newcommand{\OMEGA}{\Term{\Omega}}
\newcommand{\U}{\Term{U}}
\newcommand{\Ls}{\Term{L}}
\newcommand{\B}{\Term{B}}
\newcommand{\Y}{\Term{Y}}
\newcommand{\T}{\Term{T}}

\newcommand{\Model}[1]{\mathrm{#1}}

\newcommand{\stgy}[1]{\mathit{#1}}

\newcommand{\vnor}{\stgy{vn}}

\newcommand{\rel}[1]{\rightarrow_{#1}}

\newcommand{\mrel}[1]{\rightarrow_{#1}^*}

\newcommand{\lamK}{\lambda K}  % No modificar con \ifmmode, usar $\lamK$ en texto
\newcommand{\vv}{\textsl{V}}
\newcommand{\va}{_\vv\!}
\newcommand{\sv}{{\scriptsize\vv}}
\newcommand{\lamV}{\lambda_{\vv}}

\newcommand{\betaK}{\beta}
\newcommand{\betaV}{\beta_{\vv}} % Don't remove the space after the V
\newcommand{\betaC}{\beta_{\mathfrak{c}}}
\newcommand{\betaCV}{\beta_{\mathfrak{c}\vv}}
\newcommand{\betaKnf}{\mbox{$\betaK$-nf}}
\newcommand{\betaVnf}{\mbox{$\betaV$-nf}}
\newcommand{\vwnf}{vwnf}
\newcommand{\hnf}{hnf}
\newcommand{\chnf}{chnf}
\newcommand{\relK}{\rel{\beta}}
\newcommand{\relV}{\rel{\beta\va}}
\newcommand{\mrelK}{\mrel{\beta}}
\newcommand{\mrelV}{\mrel{\beta\va}}
\newcommand{\beqK}{=_{\beta}}
\newcommand{\beqV}{=_{\beta\va}}

\newcommand{\beqCV}{=_{\beta_{\mathfrak{c}\vv\!}}}
\newcommand{\ctx}[1]{\textup{\textbf{#1}}}

\newcommand{\hole}{[\enspace]}
\newcommand{\Omom}{{\OMEGA_\omega}}
\newcommand{\OmV}{{\OMEGA_{\vv}}}
\newcommand{\bVOmV}{{\betaV\OmV}}
\newcommand{\HV}{\mathcal{V}}

\newcommand{\lamGtz}{\lambda^{\mathsf{Gtz}}}
\newcommand{\LamGtz}{\Lambda^{\mathsf{Gtz}}}

\newcommand{\cas}[3]{[#1/#2]#3}
\newcommand{\FV}{\mathrm{FV}}
\newcommand{\he}{\textup{he}}
\newcommand{\hs}{\textup{hs}}
\newcommand{\bn}{\textup{bn}}
\newcommand{\bv}{\textup{bv}}
\newcommand{\ch}{\textup{ch}}
\newcommand{\rc}{\textup{rc}}

\newcommand{\cdv}[1]{^{\circ_{#1}}}
\newcommand{\Con}{{\mathrm{Con}}}

\newcommand{\NF}{\mathsf{NF}}
\newcommand{\Val}{\mathsf{Val}}
\newcommand{\HNF}{\mathsf{H}\NF} % \H prefix already defined by system
\newcommand{\V}{\mathsf{V}}      % V prefix
\newcommand{\W}{\mathsf{W}}      % W prefix
\newcommand{\Stuck}{\mathsf{Stuck}}
\newcommand{\Block}{\mathsf{Block}}
\newcommand{\Neu}{\mathsf{Neu}}
\newcommand{\NeuW}{\Neu\W}
\newcommand{\BlockW}{\Block\W}

\newcommand{\CH}{\mathsf{CH}}
\newcommand{\CHNF}{\CH\NF}

\newcommand{\VNF}{\V\NF}
\newcommand{\VWNF}{\V\W\NF}

\newtheorem{stmt}[thm]{Statement}

\newcounter{pi}
\newcounter{ai}

\begin{document}

\title[No solvable lambda-value term left behind]{No solvable lambda-value term left behind}

\author[Á.~ García-Pérez]{Álvaro García-Pérez\rsuper a}%
\address{{\lsuper a}Reykjavík University, Iceland}%
\email{alvarog@ru.is}%
\author[P.~Nogueira]{Pablo Nogueira\rsuper b}%
\address{{\lsuper b}IMDEA Software Institute, Madrid, Spain}%
\email{pablo.nogueira@imdea.org}%

\thanks{This research has been partially funded by the Spanish Ministry of Economy and
Competitiveness through project STRONGSOFT TIN2012-39391-C04-02, by the
Regional Government of Madrid through programme N-GREENS SOFTWARE
S2013/ICE-2731, and by the Icelandic Research Fund through project NOSOS
141558-053.
}

\begin{abstract}
  In the lambda calculus a term is solvable iff it is operationally relevant.
  Solvable terms are a superset of the terms that convert to a final result
  called normal form. Unsolvable terms are operationally irrelevant and can be
  equated without loss of consistency.  There is a definition of solvability
  for the lambda-value calculus, called v-solvability, but it is not
  synonymous with operational relevance, some lambda-value normal forms are
  unsolvable, and unsolvables cannot be consistently equated. We provide a
  definition of solvability for the lambda-value calculus that does capture
  operational relevance and such that a consistent proof-theory can be
  constructed where unsolvables are equated attending to the number of
  arguments they take (their `order' in the jargon). The intuition is that in
  lambda-value the different sequentialisations of a computation can be
  distinguished operationally. We prove a version of the Genericity Lemma
  stating that unsolvable terms are generic and can be replaced by arbitrary
  terms of equal or greater order.
\end{abstract}


\maketitle

%%%%%%%%%%%%%%%%%%%%%%%%%%%%%%%%%%%%%%%%%%%%%%%%%%%%%%%%%%%%%%%%%%%%%%%%%%%%%%
\section{Introduction}
\label{sec:intro}
Call-by-value is a common evaluation strategy of many functional programming
languages, whether full-fledged or fragments of proof assistants. Such
languages and their evaluation strategies can be formalised operationally in
terms of an underlying lambda calculus and its reduction strategies.
% The formal connection is the proof that a program evaluates to a definite
% final result iff in the calculus the term encoding the program reduces in
% standard fashion to the term encoding the result.
As shown in \cite{Plo75}, the classic lambda calculus $\lamK$ \cite{Bar84} is
inadequate to formalise call-by-value evaluation as defined by Landin's SECD
abstract machine. The adequate calculus is the lambda-value calculus $\lamV$.
The pure (and untyped) version \cite{RP04} is the core that remains after
stripping away built-in primitives whose main purpose is to facilitate the
encoding of programs as terms of the calculus. Hereafter we write $\lamV$ for
the pure version.

Unfortunately, the lambda-value calculus, and by extension its pure version,
are considered defective on several fronts for formalising call-by-value
evaluation at large, and many alternative calculi have been proposed with
various aims, \eg\ \cite{FF86,HZ09,Mog91,EHR91,AK10,AK12,AP12}.

We do not wish to propose yet another calculus. These proposals vary the
calculus to fit an intended call-by-value model, but this is one of the
choices for investigations on full abstraction. The other is to vary the model
to fit the intended calculus \cite[p.1]{Cur07}. The questions are: What does
$\lamV$ model? Is its import larger than call-by-value evaluation under SECD?
To answer these questions and avoid `the mismatch between theory [the
calculus] and practice [the model]' \cite[p.2]{Abr90} we have to first address
the open problem of whether $\lamV$ has a `standard theory'. A central piece
of a standard theory is the notion of solvability which is synonymous with
operational relevance. Let us elaborate these ideas first and discuss their
utility further below.

Recall that a lambda calculus consists of a set of terms and of proof-theories
for conversion and reduction of terms. Conversion formalises intensional
(computational) equality and reduction formalises directed computation. A term
converts/reduces to another term (both terms are in a conversion/reduction
relation) iff this fact can be derived in the conversion/reduction
proof-theory (Section~\ref{sec:prelim} illustrates). The relations must be
confluent for the proof-theory to be consistent. In the calculus the reduction
relation is full-reducing and `goes under lambda'. It is possible to reason
algebraically at any scope where free variables (which stand for unknown
operands in that scope) occur. Operational equivalence can be established for
`arbitrary terms, not necessarily closed nor of observable type'
\cite[p.3]{Cur07}.

Solvability is a basic concept in lambda calculus. It appears 18 pages after
the definition of terms in the standard reference \cite{Bar84} (terms are
defined on page 23 and solvability on page 41). Solvability was first studied
in \cite{Bar71,Bar72,Wad76} and stems from the realisation that not all
diverging terms (\ie\ terms whose reduction does not terminate) are
operationally irrelevant (\ie\ meaningless, useless, of no practical use,
etc.)  For a start, not all of them are equal. An inconsistent proof-theory
results from extending the conversion proof-theory with equations between all
diverging terms. Indeed, some diverging terms can be applied to suitable
operands such that the application converges to a definite final result of the
calculus (a `normal form' in the jargon). For other diverging terms the
application diverges no matter to how many or to which operands they are
applied. Solvable terms are therefore terms from which a normal form can be
obtained when used as functions. The name `solvable' stems from their
characterisation as solutions to a conversion. By definition, terms that
directly convert to a normal form are solvable.

In contrast, unsolvable terms are the terms that are operationally
irrelevant. A consistent proof-theory results from extending the conversion
proof-theory with equations between all unsolvables. This consistent extension
is satisfied by well-known models where unsolvables correspond to the
least-defined element of the model.  Any further extension that includes the
equations between unsolvables and is consistent is called \emph{sensible} in
the jargon. Finally, solvable terms can be characterised operationally: there
is a reduction strategy named `head reduction' that converges iff the input
term is solvable.

To summarise: $\lamK$ has a definition of solvability synonymous with
operational relevance, a sensible extended proof-theory, sensible models (\ie\
models of the sensible extension), and an operational characterisation of
solvables. All these ingredients are referred to in \cite[p.2]{Abr90} as a
`standard theory'.

However, in that work $\lamK$'s standard theory is criticised as a basis for
functional programming languages because program results are not normal forms,
there are no canonical initial models, etc. (Strictly speaking, however,
$\lamK$ is as unfit as Turing Machines as a basis for practical programming
languages.) A `lazy' lambda calculus is proposed which is closer to a
non-strict functional programming language, but that divorces solvability from
operational relevance. The latter is modified according to the notion of
`order of a term' \cite{Lon83}. Broadly, the order is the supremum (ordinal)
number of operands accepted by the term in the following inductive sense: if
the term converts to $\lambda x.M$ then it accepts $n+1$ operands where $n$ is
the number of operands accepted by $M$. Otherwise the term has order 0.
Operationally irrelevant terms are only the unsolvables of order 0. Other
unsolvables are operationally relevant and the extended proof-theory that
equates unsolvables of order $n>0$ is inconsistent.

Following similar steps, \cite{PR99,EHR91,EHR92,RP04} describe a call-by-value
calculus with a proof-theory induced by operational equivalence of terms under
SECD reduction. A definition of solvability, called $v$-solvability, is
proposed for $\lamV$. This definition is unsatisfactory because it does not
adapt $\lamK$'s original definition of solvable term, namely, `the application
of the term to suitable operands converts to a normal form'. It adapts a
derived definition, namely, `the application of the term to suitable operands
converts to the identity term'. This definition is equivalent to the former in
$\lamK$ but not in $\lamV$.  Consequently, $v$-solvability does not capture
operational relevance in $\lamV$, some normal forms of $\lamV$ (definite
results) are $v$-unsolvable, and the extended proof-theory is not
sensible. Moreover, the operational characterisation of $v$-solvables involves
a reduction strategy of $\lamK$, not of $\lamV$, and the notion of order used
is not defined in terms of $\lamV$'s conversion in a way analogous to
\cite{Lon83}. The blame is put on $\lamV$'s nature and continues to be put in
recent related work \cite{AP12,Gue13,CG14,GPR15}.

We show that $\lamV$ does indeed have a standard theory. First we revisit the
original definition of solvability in $\lamK$ and generalise it by connecting
it with the notion of effective use of an arbitrary (closed or open) term. We
then revisit $v$-solvability and show that it does not capture operational
relevance in $\lamV$ but rather `transformability', \ie\ the ability to send a
term to a chosen value. (Values are not definite results of $\lamV$ but a
requirement for confluence.) We introduce $\lamV$-solvability as the ability
to use the term effectively. Our \mbox{$\lamV$-solvability} captures
transformability and `freezability', \ie\ the ability to send a term to a
normal form, albeit not of our choice. The intuition is that terms can also be
solved by sending them to normal forms that differ operationally from
divergent terms at a point of potential divergence. The link between
solvability and effective use is a definition of order that uses $\lamV$'s
conversion, and a Partial Genericity Lemma which states that
$\lamV$-unsolvables of order $n$ are generic (can be replaced by any term) for
orders greater or equal than $n$. The $\lamV$-unsolvables of the same order
can be equated without loss of consistency, and so we construct a consistent
extension which we call $\HV$. Our proof of the Partial Genericity Lemma is
based on the proof of $\lamK$'s Genericity Lemma presented in \cite{BKV00}
that uses origin tracking. An ingredient of the proof is the definition of a
complete reduction strategy of $\lamV$ which we call `value normal order'
because we have defined it by adapting to $\lamV$ the results in \cite{BKKS87}
relative to the complete `normal order' strategy of $\lamK$. Value normal
order relies on what we call `chest reduction' and `ribcage reduction' in the
spirit of the anatomical analogy for terms in \cite{BKKS87}. The last two
strategies illustrate that standard reduction sequences fall short of
capturing all complete strategies of $\lamV$, and that a result analogous to
`quasi-needed reduction is normalising' \cite[p.208]{BKKS87} is missing for
$\lamV$. An operational characterisation of solvables in terms of a reduction
strategy of $\lamV$ is complicated but we believe possible
(Section~\ref{sec:operational-characterisation}).

To summarise, our contributions are: a definition of solvability in $\lamV$
that is synonymous with operational relevance, the Partial Genericity Lemma,
the reduction strategies value normal order, chest reduction and ribcage
reduction, and finally the sensible proof-theory where unsolvables of the same
order are equated.

The standard theory of $\lamV$ has practical consequences other than reducing
the mismatch between theory and practice, or the operational formalisation of
call-by-value.  Terms with the same functional result that may have different
sequentiality under different reduction strategies can be distinguished
operationally. Models for sequentiality exist \cite{BC82}. The full-reducing
and open-terms aspect of the calculus has applications in program optimisation
by partial evaluation and type checking in proof assistants \cite{Cre90}, in
the \textsc{PoplMark} challenge \cite{MMM05}, in reasoning within local open
scopes \cite{Cha12}, etc. The computational overload incurred by
proofs-by-reflection can be mitigated by reducing terms fully
\cite{GL02}. Finally, that some non-terminating terms (unsolvables) can be
equated without loss of consistency is of interest to proof assistants with a
non-terminating full-reducing programmatic fragment, \eg\ \cite{ACPPW08}.

This paper can be read by anyone able to follow the basic conventional lambda
calculus notions and notations that we recall in Section~\ref{sec:prelim}. The
first part of the paper provides the necessary exegesis and intuitions on
$\lamK$, $\lamV$, solvability, effective use, $v$-solvability, and introduces
our $\lamV$-solvability.  The more technical second part involves the proof of
the Partial Genericity Lemma and the consistent proof-theory.  Some background
material and routine proofs are collected in the appendix.  References to the
latter are labelled `App.' followed by a section number.

%%%%%%%%%%%%%%%%%%%%%%%%%%%%%%%%%%%%%%%%%%%%%%%%%%%%%%%%%%%%%%%%%%%%%%%%%%%%%%
\section{Overview of \texorpdfstring{$\lamK$}{lambda-K}  %
and \texorpdfstring{$\lamV$}{lambda-V}}
\label{sec:prelim}
This preliminaries section must be of necessity terse. Save for the extensive
use of EBNF grammars to define sets of terms, we follow definitions and
notational conventions of \cite{Bar84,HS08} for $\lamK$ and of \cite{Plo75}
for $\lamV$. The book \cite{RP04} collects and generalises both calculi. The
set of lambda terms is $\Lambda::=x~|~(\lambda x.\Lambda)~|\
(\Lambda\,\Lambda)$ with `$x$' one element of a countably infinite set of
variables that we overload in grammars as non-terminal for such
set. Uppercase, possibly primed letters $M$, $M'$, $N$, etc., will stand for
terms. In words, a term is a variable, or an abstraction $(\lambda x.M)$ with
bound variable $x$ and body $M$, or the application $(MN)$ of an operator $M$
to an operand $N$. We follow the common precedence and association convention
where applications associate to the left and application binds stronger than
abstraction. Hence, we can drop parenthesis and write $(\lambda
x.x\,y)\,p\,q\,(\lambda x.x)$ rather than $((((\lambda x.(x\,y))p)q)(\lambda
x.x))$, and we can write $\Lambda::=x~|~\lambda x.\Lambda~|\
\Lambda\,\Lambda$, and $\lambda x.M$, and $MN$. For brevity we write $\lambda
x_1\ldots x_n.M$ instead of $\lambda x_1.\lambda x_2.\ldots\lambda x_n.M$. We
write $\FV$ for the function that delivers the set of free variables of a
term. We assume the notions of bound and free variable and write $\equiv$ for
the identity relation on terms modulo renaming of bound variables.\footnote{We
  are following here the convention of Appendix~C in \cite{Bar84} not to be
  confused with the `Barendregt convention' or `hygiene rule' of \cite{Bar90}
  where bound variables and free variables must differ.}  For example,
$\lambda x.xz \equiv \lambda y.yz$. We also abuse $\equiv$ to define
abbreviations, \eg\ $\I\equiv \lambda x.x$. Like \cite{CF58,HS08}, we write
$\cas{N}{x}{M}$ for the capture-avoiding substitution of $N$ for the free
occurrences of $x$ in $M$. We write $\Lambda^0$ for the set of closed lambda
terms, \ie\ terms $M$ such that $\FV(M)=\emptyset$. We use the same postfix
superscript for the operation on a set of terms that delivers the subset of
closed terms. The set of values (non-applications) is $\Val::=x~|~\lambda
x.\Lambda$. The set of closed values is $\Val^0$ and consists of closed
abstractions. A context $\ctx{C}\hole$ is a term with one hole, \eg\
$\ctx{C}\hole \equiv \lambda x.\hole$. Plugging a term within the hole may
involve variable capture, \eg\ $\ctx{C}[\lambda y.x] \equiv \lambda x.\lambda
y.x$.

The conversion/reduction proof-theories of $\lamK$ and $\lamV$ can be
presented as instances of the Hilbert-style proof-theory shown in
Fig.~\ref{fig:lambda-calculi} that is parametric (\cf\ \cite{RP04}) on a set
$\mathbb{P}$ of permissible operands $N$ in the contraction rule ($\beta$)
which describes the conversion/reduction of the term $(\lambda x.B)N$, that
is, the application of an abstraction (a function) to an operand.  Operands
are arbitrary terms in $\lamK$ and restricted to values in $\lamV$ which means
that $\lamV$ has fewer conversions/reductions than $\lamK$.

\begin{figure}[ht]
\begin{mathpar}
  \inferrule*[Left=($\beta$)]%
  { N \in \mathbb{P}}%
  {(\lambda x.B)N = \cas{N}{x}B}%
  \and %
  \inferrule*[Left=($\mu$)]%
  {N = N'}%
  {M\,N = M\,N'}%
  \and %
  \inferrule*[Left=($\nu$)]%
  {M = M'}%
  {M\,N = M'\,N}%
  \and%
  \inferrule*[Left=($\xi$)]%
  {B = B'}%
  {\lambda x.B = \lambda x.B'}%
  \\
  \inferrule*[Left=($\rho$)]%
  { }%
  {M = M}%
  \and%
  \inferrule*[Left=($\tau$)]%
  {M = N \quad N = P}%
  {M = P}%
  \and%
  \inferrule*[Left=($\sigma$)]%
  {M = N}%
  {N = M}%
\end{mathpar}

\vspace{0.5cm}

\begin{tabular}{l|l|l|l}
  Theory & $=$ & $\mathbb{P}$ & discarded rules \\
  \hline\hline
  $\lamK$ conversion              & $\beqK$  & $\Lambda$ & none \\
  $\lamK$ multiple-step reduction & $\mrelK$ & $\Lambda$ & $\sigma$ \\
  $\lamK$ single-step reduction   & $\relK$  & $\Lambda$ & $\rho$, $\tau$,
  $\sigma$ \\
  $\lamV$ conversion              & $\beqV$  & $\Val$ & none \\
  $\lamV$ multiple-step reduction & $\mrelV$ & $\Val$ & $\sigma$ \\
  $\lamV$ single-step reduction   & $\relV$  & $\Val$ & $\rho$, $\tau$, $\sigma$
\end{tabular}
\caption{$\lamK$ and $\lamV$ proof-theories.}
\label{fig:lambda-calculi}
\end{figure}

\begin{figure}[ht]
\begin{tabular}{lllll}
  Set &&& Description & Abbreviation in the text\\
  \hline\hline
  $\Lambda$ & $::=$ & $x~|~\lambda x.\Lambda~|\ \Lambda\,\Lambda$ &
  lambda terms \\
  $\Val$ & $::=$ & $x~|~\lambda x.\Lambda$ & values \\
  $\Neu$ & $::=$ & $x\,\Lambda\,\{\Lambda\}^*$ & $\lamK$ neutrals \\
  $\NF$ & $::=$ & $\lambda x.\NF\ |\ x\,\{\NF\}^*$ & $\lamK$ normal forms
  & {\betaKnf}s (singular \betaKnf) \\
  $\HNF$ & $::=$ & $\lambda x.\HNF~|~x\,\{\Lambda\}^*$ & head normal forms
  & {\hnf}s (singular \hnf) \\
  $\Neu\V$ & $::=$ & $\Neu~|~\Block\,\{\Lambda\}^*$ & $\lamV$ neutrals  \\
  $\Block$ & $::=$ & $(\lambda x.\Lambda)\Neu\V$ & blocks \\
  $\VNF$  & $::=$ & $x\ |\ \lambda x.\V\NF\ |\ \Stuck$ & $\lamV$ normal forms
  & {\betaVnf}s (singular \betaVnf) \\
  $\Stuck$ & $::=$ & $x\,\VNF\,\{\VNF\}^*$ & stucks \\
  & $|$ & $\Block\NF\,\{\VNF\}^*$ & \\
  $\Block\NF$ & $::=$ & $(\lambda x.\VNF)\,\Stuck$ & blocks in \betaVnf
\end{tabular}
\caption{Sets of terms.}
\label{fig:lam-sets}
\end{figure}

\begin{figure}[ht]
\begin{tabular}{llll}
  Abbreviation & Term & has \betaKnf & has \betaVnf \\
  \hline\hline
  $\I$ & $\lambda x.x$ & yes & yes \\
  $\K$ & $\lambda x.\lambda y.x$ & yes & yes \\
  $\DELTA$ & $\lambda x.xx$ & yes & yes \\
  $\OMEGA$ & $\DELTA\DELTA$ & no & no \\
  $\U$ & $\lambda x.\B$ & no  & yes \\
  $\B$ & $(\lambda y.\DELTA)(x\,\I)\DELTA$ & no & yes
\end{tabular}
\caption{Glossary of particular terms.}
  \label{fig:lam-glossary}
\end{figure}

In $\lamV$ the rule ($\beta$) restricted to operand values is named
($\betaV$). The term $(\lambda x.B)N$ is called a $\betaK$-redex iff
$N\in\Lambda$, and a $\betaV$-redex iff $N\in\Val$. A term is a
$\betaK$-normal-form (hereafter abbrev. \betaKnf) iff it has no
$\betaK$-redexes. A term is a \betaVnf\ iff it has no $\betaV$-redexes.  The
inference rules are: compatibility ($\mu$) ($\nu$) ($\xi$), reflexivity
($\rho$), transitivity ($\tau$), and symmetry ($\sigma$). The table underneath
names the proof-theory obtained, and the relation symbol, for given
$\mathbb{P}$ and rules. The conversion relation includes the reduction
relation.  A term $M$ has a \betaKnf\ $N$ when $M \beqK N$ and $N$ is a
\betaKnf. A term $M$ has a \betaVnf\ $N$ when $M \beqV N$ and $N$ is a
\betaVnf.  A term $M$ has a value when $M \beqV N$ and $N\in\Val$.  All
proof-theories are consistent (not all judgements are derivable) due to
confluence (a term has at most one \betaKnf\ and at most one \betaVnf).

Fig.~\ref{fig:lam-sets} defines sets of terms and Fig.~\ref{fig:lam-glossary}
defines abbreviations of terms used in the following sections. A full table of
sets of terms and abbreviations of terms is provided in
App.~\ref{app:full-gloss}.  Observe that every term of $\Lambda$ has the form
$\lambda x_1\ldots x_n.\,H\,M_1\cdots M_m$ where $n\geq0$, $m\geq0$, and
$M_1\in\Lambda$, \ldots, $M_m\in\Lambda$. The head term $H$ is either a `head
variable' $x$ (which may or may not be one of $x_1 \ldots x_n$) or an
application $(\lambda x.B)N$ (which is a redex iff $N\in\mathbb{P}$).  The set
$\Neu$ of \emph{neutrals} of $\lamK$ contains applications $x\,M_1\cdots M_n$
with $n\geq 1$. The expression $\{\Lambda\}^*$ in the grammar stands for zero
or more occurrences of $\Lambda$. The applications associate as $(\ldots
((x\,M_1)M_2)\cdots M_n)$ according to the standard convention. The set $\NF$
of {\betaKnf}s consists of abstractions with bodies in \betaKnf, free
variables, and neutrals in \betaKnf. According to the grammar, every \betaKnf\
has the form $\lambda x_1 \ldots x_n.x\,N_1\cdots N_m$ where $n\geq0$, $m \geq
0$, $N_1\in\NF$, \ldots, $N_m\in\NF$, and $x$ may or may not be one of
$x_1\ldots x_n$. The set $\HNF$ of head normal forms (abbrev. {\hnf}s)
consists of terms that differ from {\betaKnf}s in that $N_1\in\Lambda$,
\ldots, $N_m\in\Lambda$. Clearly, $\NF\subset\HNF$.

Some examples: $\lambda x.\I$ is a \betaKnf\ and a \hnf, $\lambda x.\I\DELTA$
is not a \betaKnf\ (it contains the $\betaK$-redex $\I\,\DELTA$) nor a \hnf\
(it has no head variable), $\lambda x.\,x\,\I\DELTA$ is not a \betaKnf\ but it
is a \hnf, and both $x\,(\lambda x.\,\I)$ and $x\,\OMEGA$ are neutrals, with
only the first in \betaKnf.

The set $\Neu\V$ of neutrals of $\lamV$ contains the neutrals $\Neu$ of
$\lamK$ and blocks applied to zero or more terms. The set $\Block$ of
\emph{blocks} contains applications $(\lambda x.B)N$ where $N\in\Neu\V$. These
are applications that do not convert to a $\betaV$-redex and are therefore
blocked. (Our blocks differ from the `head blocks' of \cite[p.8]{RP04} and the
`pseudo redexes' of \cite[p.4]{HZ09} which require $N\not\in\Val$ and so
include terms like $(\lambda x.B)(\I\,\I)$ that convert to a $\betaV$-redex.)
The set $\V\NF$ of {\betaVnf}s contains variables, abstractions in \betaVnf,
and \emph{stuck terms} (`stucks' for short) which are neutrals of $\lamV$ in
\betaVnf. The set $\Stuck$ of stucks contains $\Neu$ neutrals of $\lamK$ in
\betaVnf\ and blocks in \betaVnf. According to the grammar, every \betaVnf\
has the form $\lambda x_1 \ldots x_n.H\,Z_1\cdots Z_m$ with $n\geq0$, $m\geq
0$, $Z_1\in\V\NF$, \ldots, $Z_m\in\V\NF$, and $H$ either a variable or a block
in \betaVnf.

Some examples: $x\,\OMEGA$ is a neutral not in \betaVnf, $x\,\DELTA$ is a
neutral in \betaVnf\ (a stuck), $(\lambda x.y)(x\,\OMEGA)$ is a block not in
\betaVnf, and $(\lambda x.y)(x\,\DELTA)$ is a block in \betaVnf\ (a stuck).

A \emph{reduction strategy} of $\lamK$ (resp. of $\lamV$) is a partial
function that is a subrelation of $\mrelK$ (resp. of $\mrelV$\,). A reduction
strategy is \emph{complete} with respect to a notion of irreducible term when
the strategy delivers the irreducible term iff the input term has one,
diverging otherwise. A reduction strategy is \emph{full-reducing} when the
notion of irreducible term is a \betaKnf\ (resp. \betaVnf). The Quasi-Leftmost
Reduction Theorem \cite[Thm.~3.22]{HS08} states, broadly, that any reduction
strategy of $\lamK$ that eventually contracts the leftmost redex is
full-reducing and complete.  One such well-known strategy is leftmost
reduction \cite{CF58}, also known as leftmost-outermost reduction (when
referring to the redex's position in the abstract syntax tree of the term) or,
more commonly, as normal order. The Standardisation Theorem
\cite[Thm.~3]{Plo75} guarantees that there are full-reducing and complete
strategies of $\lamV$. One such strategy is described in \cite{RP04} and
discussed in Section~\ref{sec:value-normal-order}.

%%%%%%%%%%%%%%%%%%%%%%%%%%%%%%%%%%%%%%%%%%%%%%%%%%%%%%%%%%%%%%%%%%%%%%%%%%%%%%
\section{Solvability reloaded}
\label{sec:lamK-solv}
As explained in the introduction, a term is solvable iff a normal form can be
obtained from it when used as a function. Solvability is usually defined first
for closed terms and then extended to open terms.
\begin{defi}[\textsc{SolN}]
  A term $M\in\Lambda^0$ is solvable in $\lamK$ iff there exists $N\in\NF$ and
  there exist operands $N_1\in\Lambda$, \ldots, $N_k\in\Lambda$ with $k\geq0$
  such that $M\,N_1\,\cdots\,N_k \beqK N$.
\end{defi}
This definition is the seminal one on page~87 of \cite{Bar71}.\footnote{The
  provisos $M\in\Lambda^0$ and $k\geq0$ are implicit in the original
  definition due to the context of the thesis (closed-term models) and its
  subscript convention. They are explicit in later definitions
  \cite{Bar72,Wad76,Bar84}. The order of existential quantifiers is
  immaterial. The original definition says `$M\,N_1\cdots N_k$ has a \betaKnf'
  which as explained in Section~\ref{sec:prelim} is the same as `converts to a
  \betaKnf'. In \cite{Bar84} the requirement on $N$ is immaterially changed
  from being a \betaKnf\ to having a \betaKnf.}  In words, a closed term is
solvable iff it converts to a \betaKnf\ when used in operator position at the
top level. If the term is or has a \betaKnf\ then it is trivially solvable by
choosing $k=0$. Let us illustrate with examples that also explain the focus on
closed terms. First, take the diverging closed term $\OMEGA$ (an abbreviation
of $\DELTA\DELTA$, \ie\ $\OMEGA\equiv\DELTA\DELTA\equiv(\lambda x.xx)(\lambda
x.xx)$). A \betaKnf\ cannot be obtained from it no matter to how many or to
which operands it is applied, \eg\ $(\DELTA\DELTA)N_1 \cdots N_k \beqK
((\lambda x.x\,x)\DELTA)N_1 \cdots N_k \beqK (\DELTA\DELTA)N_1 \cdots N_k
\beqK\ldots$ is an infinite loop. Terms like $\OMEGA$ are operationally
irrelevant. Now take the closed terms $\lambda x.x\,\I\,\OMEGA$ and $\lambda
x.x\,\K\,\OMEGA$. Both terms diverge and yet both deliver a \betaKnf\ when
applied to suitable operands.  For example, $(\lambda x.x\,\I\,\OMEGA)\K \beqK
\I$, and $(\lambda x.x\,\K\, \OMEGA)\K \beqK \K$. The {\betaKnf}s obtained
from such diverging function terms are different, therefore they have
different operational behaviour and cannot be equated. More precisely, a
proof-theory with judgements $M = N$ can be obtained by taking the conversion
proof-theory (if $M \beqK N$ then $M = N$) and adding the equation $\lambda
x.x\,\I\,\OMEGA = \lambda x.x\,\K\,\OMEGA$. However, this extended
proof-theory is inconsistent because the false equation $\I = \K$ is then
provable.

The focus on closed terms is because some open terms contain neutral terms
(Section~\ref{sec:prelim}) that block applications \cite{Wad76}. For example,
take the neutral $x\,\OMEGA$ and apply it to operands: $(x\,\OMEGA)N_1\cdots
N_k$. The conversion to \betaKnf\ is impossible because the diverging subterm
$\OMEGA$ is eventually converted due to the presence of the free variable $x$
that blocks the application to the operands. (Similarly, in $x\,y\,\OMEGA$ the
neutral subterm $x\,y$ blocks the application.) However, a free variable
stands for some operator, so substituting a closed operator for the variable
may yield a solvable term. For example, substitute $\K\,\I$ for $x$ and choose
$k=0$, then $\K\,\I\,\OMEGA \beqK \I$. Traditionally, open terms are defined
as solvable iff the closed term resulting from such substitutions is solvable.
We postpone the discussion to Section~\ref{sec:open-open} where we show that
fully closing is excessive in $\lamK$. In Section~\ref{sec:v-solv} we show
that it is counterproductive for defining solvability in $\lamV$. We conclude
this section with the role of solvables in the development of a standard
theory.

Solvable terms are approximations of totally defined terms. They are `at least
partially defined'~\cite{Wad76}. In contrast, unsolvable terms are
`hereditarily'~\cite{Bar71} or `totally'~\cite{Wad76} undefined, and can be
equated without loss of consistency. More precisely, given the set of
equations $\mathcal{H}_0=\{M = N\ |\ M,N\in\Lambda^0\ \text{unsolvable} \}$, a
consistent extended proof-theory $\mathcal{H}$ results from adding
$\mathcal{H}_0$'s equations as axioms to $\lamK$ (\ie\ $\mathcal{H} =
\mathcal{H}_0+\lamK$) \cite{Bar84}. A consistent extension where unsolvables
are equated (\ie\ contains $\mathcal{H}$) is called \emph{sensible}. A
consistent extension that does not equate solvables and unsolvables is called
\emph{semi-sensible}. There are standard models that satisfy $\mathcal{H}$,
with unsolvables corresponding to the least elements of the model
\cite{Bar72,Bar84}. By extension, such models are called \emph{sensible
  models} \cite[p.505]{Bar84}. Solvable terms can be characterised
operationally: there is a reduction strategy of $\lamK$ called `head
reduction' that converges iff the input term is solvable. (Solvability, like
having \betaKnf, is semi-decidable.) More precisely, solvable terms exactly
correspond to terms with \hnf, and head reduction delivers a \hnf\ iff the
input term has one, diverging otherwise \cite{Wad76,Bar84}. (In the technical
jargon, head reduction is said to be \emph{complete} with respect to \hnf.)

\subsection{Other equivalent definitions of solvability}
\label{sec:eq-defs}
There are two other equivalent defi\-nitions of solvability that use different
equations \cite{Bar72,Wad76,Bar84}.
\begin{defi}[\textsc{SolI}]
  A term $M\in\Lambda^0$ is solvable in $\lamK$ iff there exist operands
  $N_1\in\Lambda$, \ldots, $N_k\in\Lambda$ with $k\geq0$ such that
  $M\,N_1\,\cdots\,N_k \beqK \I$.
\end{defi}
\begin{defi}[\textsc{SolX}]
  A term $M\in\Lambda^0$ is solvable in $\lamK$ iff for all $X\in\Lambda$
  there exist operands $N_1\in\Lambda$, \ldots, $N_k\in\Lambda$ with $k\geq0$
  such that $M\,N_1\,\cdots\,N_k \beqK X$.
\end{defi}
In words, a closed term is solvable iff it is convertible by application to
the identity term or to any given term. Definition \textsc{SolI} is \emph{de
  facto} in most presentations. These definitions are equivalent to
\textsc{SolN} (capture the same set of solvables) because of two properties
that hold in $\lamK$. The first is stated in the following lemma.
\begin{lem}[Lemma~4.1 in \cite{Wad76}]
\label{lem:4.1Wad76}
If $M\in\Lambda^0$ has a \betaKnf\ then for all $X\in\Lambda$ there exist
operands $X_1\in\Lambda$, \ldots, $X_k\in\Lambda$ with $k\geq 0$ such that
$M\,X_1\,\cdots\,X_k \beqK X$.
\end{lem}
In words, a closed term with \betaKnf\ can be converted by application to any
given term. This lemma is the link between \textsc{SolN}'s \emph{existential}
property of having a \betaKnf\ and \textsc{SolX}'s \emph{universal} property
of converting to any term. The shape of a \betaKnf\ is the key to this link,
as the proof of the lemma illustrates.
\begin{proof}[Proof of Lemma~\ref{lem:4.1Wad76}]
  As explained in Section~\ref{sec:prelim}, a \betaKnf\ has the form $\lambda
  x_1\ldots x_n.\,x\,N_1\cdots N_m$ with $n\geq 0$, $m\geq 0$, and
  $N_1\in\NF$, \ldots, $N_m\in\NF$. Since $M$ is closed, its \betaKnf\ $M'$
  has $n>0$ with $x$ is one of $x_i$. Lemma~\ref{lem:4.1Wad76} holds by
  choosing $k=n$, $X_j$ arbitrary for $j\not=i$, and $X_i\equiv \K^m X$, with
  $\K^m$ the term that takes $m+1$ operands and returns the first one. Thus,
  \mbox{$M\,X_1 \cdots {(\K^mX)}_i\cdots X_n \beqK X$} holds because $M \beqK
  M'$ and $M'\,X_1 \cdots {(\K^mX)}_i\cdots X_n \beqK (\K^mX)N'_1 \cdots N'_m
  \beqK X$, with $N'_i$ the result of substitutions on $N_i$.
\end{proof}
The link between \textsc{SolI} and \textsc{SolX} is provided by the property
that for all $X\in\Lambda$ the conversion $\I\,X \beqK X$ holds
\cite[p.171ff]{Bar84}. We provide here an explicit proof.
\begin{lem}
  \label{lem:equiv-solvs}
  The solvability definitions \textsc{SolN}, \textsc{SolI}, and \textsc{SolX}
  are equivalent in $\lamK$.
\end{lem}
\begin{proof}
  We use different operand symbols and subscripts to distinguish the
  equations:
  \begin{align}
    M\,N_1\,\cdots\,N_k & \beqK  N \tag*{\textsc{SolN}} \\
    M\,Y_1\,\cdots\,Y_l & \beqK \I \tag*{\textsc{SolI}} \\
    M\,Z_1\,\cdots\,Z_j & \beqK X  \tag*{\textsc{SolX}}
  \end{align}
  We first prove \textsc{SolX} iff \textsc{SolN}: From \textsc{SolX} we prove
  \textsc{SolN} by choosing $k=j$, $N_i\equiv Z_i$, and $X$ the \betaKnf\
  $N$. Conversely, given \textsc{SolN} then $MN_1\cdots N_k$ has a \betaKnf,
  so by Lemma~\ref{lem:4.1Wad76} we have that forall $X\in\Lambda$ the
  conversion $M\,N_1\cdots N_k\,X_1\cdots X_{k'} \beqK X$ holds. Then
  \textsc{SolX} follows by choosing $j=k+k'$, $Z_1 \equiv N_1$, \ldots, $Z_k
  \equiv N_k$, $Z_{k+1} \equiv X_1$, \ldots, $Z_j \equiv X_{k'}$.

  We now prove \textsc{SolX} iff \textsc{SolI}: From \textsc{SolX} we prove
  \textsc{SolI} by choosing $l=j$, $Y_i\equiv Z_i$, and $X\equiv\I$.
  Conversely:

  \begin{tabular}{lll}
    (a) & $M\,Y_1\,\cdots\,Y_l \beqK \I$ & \textsc{SolI} \\
    (b) & $M\,Y_1\cdots Y_l\,X \beqK \I\,X$ & by ($\nu$) on (a) with any $X$ \\
    (c) & $\I X \beqK X$ & by ($\beta$) \\
    (d) & $M\,Y_1\cdots Y_l\,X \beqK X$ & by ($\tau$) on (b),(c)
  \end{tabular}\\
  \noindent
  Then, \textsc{SolX} holds by choosing $j=l+1$, $Z_1 \equiv Y_1$, \ldots,
  $Z_{j-1}=Y_l$, $Z_j \equiv X$.
\end{proof}
Bear in mind that although all definitions are equivalent, \textsc{SolI} and
\textsc{SolX} are possible because of properties that hold in $\lamK$, and
therefore \textsc{SolI} and \textsc{SolX} are secondary. As we shall see in
Section~\ref{sec:v-solv}, the anaologous in $\lamV$ of
Lemma~\ref{lem:4.1Wad76} is not the case, nor are the analogous of
\textsc{SolI}, \textsc{SolX}, and Lemma~\ref{lem:equiv-solvs}. Adapting
\textsc{SolI} or \textsc{SolX} to that calculus will leave solvable terms
behind.

\subsection{Open terms, and open and non-closing contexts}
\label{sec:open-open}
Solvability has been typically extended to open terms by requiring at least
one closed substitution instance or all closures of the open term\footnote{A
  closed substitution instance of $M$ is a closed term resulting from
  substituting closed terms for all the free variables of $M$. A closure of
  $M$ is a term $\lambda x_1 \ldots x_n.M$ such that
  $\FV(M)=\{x_1,\ldots,x_n\}$. Since different closures differ only on the
  order of prefix lambdas, if one closure is solvable then all other closures
  are too by passing the operands to the closure in the appropriate
  order. Substitutions and closures are connected by the $\betaK$-rule.} to be
solvable~\cite{Wad76,Bar72,Bar84}. As we discussed in
Section~\ref{sec:lamK-solv}, neutral terms are the reason for closing.
Substituting closed operators for the blocking free variables of neutrals may
yield solvable terms. For example,
$\cas{\K\,\I}{x}{(x\,\OMEGA)}\equiv\K\,\I\,\OMEGA$ is trivially solvable
according to \textsc{SolN} by choosing $k=0$. Similarly, the closure $\lambda
x.x\,\OMEGA$ is solvable by choosing $k=1$ and $N_1\equiv\K\,\I$.

A traditional definition of solvability for open and closed terms uses a `head
context' to close the term before passing the operands \cite{Wad76} (head
contexts are defined on page 491 and solvability with head contexts on page
503).
\begin{defi}[\textsc{SolH}]
  A term $M\in\Lambda$ is solvable in $\lamK$ iff there exists $N\in\NF^0$ and
  there exists a head context $\ctx{H}\hole \equiv ((\lambda x_1\ldots
  x_n.\hole)C_1\cdots C_n)N_1 \cdots N_k$ with $n\geq0$, $k\geq0$,
  $\FV(M)=\{x_1,\ldots,x_n\}$, $C_1\in\Lambda^0$, \ldots, $C_n\in\Lambda^0$,
  and $N_1\in\Lambda^0$, \ldots, $N_k\in\Lambda^0$ such that $\ctx{H}[M] \beqK
  N$.
\end{defi}
In words, the head context forces the closed $C_i$ to be substituted for all
the free variables (if there are any) of the term placed within the hole. The
resulting closed substitution instance is then at the top-level operator
position where it is applied to the closed $N_i$ operands. The top-level
operator position is a `head' position (Section~\ref{sec:prelim}), hence the
name of the context. Since $\ctx{H}\hole$ is a \emph{closed and closing}
context, the \betaKnf\ $N$ has to be closed too. In \cite{PR99}, \textsc{SolH}
and \textsc{SolI} are combined and the conversion is $\ctx{H}[M] \beqK \I$.

However, using a closed and closing context is excessive. The nature of
solvability and the previous definitions do not require it. To begin with, an
open term that is or has a \betaKnf\ is, by its very nature, solvable. For
other open terms not every free variable has to be substituted, only the
blocking ones that prevent solving the term. In all the previous definitions
the $N_i$ operands are arbitrary, and so the requirement that $N_i$ are closed
in $\ctx{H}\hole$ can be dropped. Since in \textsc{SolI} both $M$ and $\I$ are
closed then the open $N_i$ or their open subterms must be eventually discarded
in the conversion to $\I$. But in \textsc{SolN} the \betaKnf\ $N$ is arbitrary
too, so not every open $N_i$ operand or open subterm therein has to be
discarded.

A less restrictive definition is perfectly possible:
\begin{defi}[\textsc{SolF}]
  A term $M\in\Lambda$ is solvable in $\lamK$ iff there exists $N\in\NF$ and
  there exists a \emph{function context} $\ctx{F}\hole \equiv (\lambda
  x_1\ldots x_n.\hole)N_1\cdots N_k$ with $n\geq 0$, $k\geq0$, and
  $N_1\in\Lambda$, \ldots, $N_k\in\Lambda$ such that $\ctx{F}[M] \beqK N$
\end{defi}
This definition is closer to \textsc{SolN}. The function context can be
\emph{open} and \emph{non-closing}: $N$ and $N_i$ may be open, and not every
free variable of $M$ need be substituted. For example, $x\,\OMEGA$ is solved
by the open function context $(\lambda x.\hole)(\K\,N)$ where $N$ is an open
\betaKnf. And $x\,y\,\OMEGA$ is solved by the non-closing function context
$(\lambda x.\hole)\K$ which does not close $y$.
\begin{lem}[Generalisation of Lemma~\ref{lem:4.1Wad76}]
  \label{lem:GenOf4.1Wad76}
  If $M\in\Lambda$ has a \betaKnf\ then for all $X\in\Lambda$ there exists a
  function context $\ctx{F}\hole$ such that $\ctx{F}[M] \beqK X$.
\end{lem}
\begin{proof}
  The \betaKnf\ of $M$ has the form $\lambda x_1\ldots x_n.x\,N_1\cdots
  N_m$ with $n\geq0$, $m\geq 0$ and $N_1\in\NF$, \ldots, $N_m\in\NF$. If
  $x\in\FV(M)$ the lemma holds by choosing $\ctx{F}\hole \equiv (\lambda
  x.\hole)(\K^m X)X_1\cdots X_n$ with $X_i$ arbitrary and $\K^m$ the term
  that takes $m+1$ operands and returns the first one. If $x\not\in\FV(M)$
  then $x$ is $x_i$ for some $i$. The lemma holds by choosing $\ctx{F}\hole
  \equiv \hole X_1\cdots X_{i-1}(\K^m X)X_{i+1}\cdots X_n$.
\end{proof}
Let us note that the lemma also holds with the proviso relaxed to `$M$ has a
\hnf'.
\begin{thm}
  \label{thm:equiv-solvs-open}
  In $\lamK$ the solvability definitions \textsc{SolH} and \textsc{SolF} are
  equivalent.
\end{thm}
Intuitively, if we have a solving head context then we have a solving function
context because function contexts subsume head contexts. And if we have a
solving function context then we can construct a solving head context by
carefully closing the former and the \betaKnf. The proof of
Thm.~\ref{thm:equiv-solvs-open} is not so short and we have put it in
App.~\ref{app:open-open} with an accompanying example illustrating the
construction of a solving head context from a solving function context.
% The free variables in the function context that do not end in $N$ can be
% replaced by arbitrary closed terms. The remaining variables are head
% variables of some subterm in \betaKnf, and can be substituted for closed
% terms in \betaKnf.

As we shall see in Section~\ref{sec:lamV-solv}, the analogous in $\lamV$ of
Thm.~\ref{thm:equiv-solvs-open} is not the case.  Adapting \textsc{SolH} to
that calculus will leave solvable terms behind.

\subsection{Solvability and effective use}
\label{sec:effective-use-lamK}
As noted in \cite{PR99} there is a more general definition of solvability that
connects the notions of `operational relevance' and `effective use' of a
term. A term is effectively used when it is eventually used as an operator.
The term is operationally relevant iff it then delivers a final result, which
in $\lamK$ is a \betaKnf.  In all previous solvability definitions, the term
to solve is placed at the top-level operator position and thus it is
effectively used. If it were placed at other positions then it may be
eventually used as operator or it may be trivially used (discarded). If placed
at an operand position that is never discarded, never gets to an operator
position where it is applied to operands, and is returned as the final result,
then the term is effectively used. It is as if the term were placed within an
empty function context. Thus, a final result is in operator position, is
effectively used, and is operationally relevant.

An unsolvable term cannot be effectively used to deliver a \betaKnf:
`unsolvable terms can never have a nontrivial effect on the outcome of a
reduction' \cite[p.506]{Wad76}. More precisely, if $M$ is unsolvable then for
all $X$, $M\,X$ is unsolvable \cite[Cor.~8.34]{Bar84}. Unsolvable terms that
are not effectively used are generic: they can be substituted by arbitrary
terms. This is formalised by the so-called Genericity Lemma. The following
statement of the Lemma is a combination of the versions in
\cite[Prop.~14.3.24]{Bar84} and \cite[Cor.~5.5]{Wad76} (both collected in
App.~\ref{app:effective-use-lamK} for ease of reference). These versions use
arbitrary contexts $\ctx{C}\hole$ because $\ctx{C}[M]$ is more general than
$M\,X$. The latter is a particular case of the former for $\ctx{C}\hole \equiv
\hole\,X$. With the context, the term plugged into the hole may eventually
appear in operator position.
\begin{lem}[Genericity Lemma]
  \label{lem:lamK-genericity-lemma}
  Let $M\in \Lambda$ and $N\in\NF$. $M$ is unsolvable in $\lamK$ implies that
  for all contexts $\ctx{C}\hole$, if $\ctx{C}[M] \beqK N$ then for all
  $X\in\Lambda$ it is the case that $\ctx{C}[X]\beqK N$. In formal logic:
  \begin{displaymath}
    M\ \textrm{unsolvable} \Rightarrow
      (\forall\ctx{C}\hole.\, \ctx{C}[M] \beqK N \Rightarrow
      (\forall X\in\Lambda.\,\ctx{C}[X] \beqK N))
  \end{displaymath}
\end{lem}
In words, if plugging an unsolvable term in a given arbitrary context converts
to a \betaKnf\ then plugging any other term also converts to that \betaKnf.
The unsolvable is not used effectively in the context. Although the lemma is
stated as an implication, it is actually an equivalence because the negation
of the consequent is a necessary condition for `$M$ solvable' by the
\textsc{SolF} definition of solvability.  Clearly, if $M$ is solvable then
there exists $\ctx{C}\hole\equiv\ctx{F}\hole$ such that $\ctx{F}[M]\beqK N$,
and by the shape of $\ctx{F}\hole$ it is not the case that for all
$X\in\Lambda$, $\ctx{F}[X]\beqK N$. Take for instance $\ctx{F}[\OMEGA]$ which
diverges. (Note that if $M$ is solvable and $\ctx{C}[M] \beqK N$ holds then
$\ctx{C}[X]\beqK N$ should not hold for terms $X$ that are not convertible to
$M$ unless $M$ is not effectively used in $\ctx{C}\hole$.)

The lemma is a definition of solvability when read as the inverse
equivalence:
\begin{displaymath}
  M\ \textit{solvable} \Leftrightarrow
    (\exists \ctx{C}\hole.\,\ctx{C}[M] \beqK N \wedge
    \neg(\forall X\in\Lambda.\,(\ctx{C}[X] \beqK N)))
\end{displaymath}
The following definition simply moves $N$ to the formula from the proviso.
\begin{defi}[\textsc{SolC}]
  A term $M\in\Lambda$ is solvable in $\lamK$ iff there exists a context
  $\ctx{C}\hole$ such that $\ctx{C}[M]\beqK N$ for some $N\in\NF$ and not for
  all $X\in\Lambda$ it is the case that $\ctx{C}[X]\beqK N$.
\end{defi}
In words, $M$ solvable means there exists a context that uses $M$ effectively
to deliver a \betaKnf. Function contexts are just one possible type of context
applicable in \textsc{SolC}.

%%%%%%%%%%%%%%%%%%%%%%%%%%%%%%%%%%%%%%%%%%%%%%%%%%%%%%%%%%%%%%%%%%%%%%%%%%%%%%
\section{Call-by-value and %
  pure \texorpdfstring{$\lamV$}{lambda-V}}
\label{sec:pure-lamV}
In call-by-value functional programming languages, the evaluation of
application expressions $e_1\,e_2$ can be broadly described in `big-step'
fashion as follows. The operator expression $e_1$ is first evaluated to a
`value' $v_1$ where `value' means here a first-class final result of the
language. Functions are first-class values in such languages and their bodies
are compiled, not evaluated. (In the SECD machine, the corresponding
abstraction is not reduced, SECD reduction is `weak', meaning it does not `go
under lambda'.) The operand expression $e_2$ is next evaluated to a value
$v_2$. Finally, the result of passing $v_2$ to $v_1$ is evaluated. Evaluation
diverges at the point where the first sub-evaluation diverges. Evaluation may
halt due to a run-time error. The order of evaluation matters w.r.t. the point
of divergence or halting.\footnote{Some languages prefer to evaluate $e_2$
  before $e_1$, or instead of binary applications consider applications with
  multiple operands, evaluating the latter in left-to-right or right-to-left
  fashion. Some languages eschew divergence and run-time errors by means of a
  strong but yet expressive type discipline.}

In pure $\lamV$, an application $M\,N$ can be reduced to \betaVnf\ in several
ways with the restriction that if $M$ is an abstraction or reduces to an
abstraction, say $\lambda x.B$, and $N$ is a value or reduces to a value, say
$V$, then the redex application $(\lambda x.B)V$ can be reduced in one step to
$\cas{V}{x}{B}$, with reduction continuing on the result of the substitution.
Either the abstraction $\lambda x.B$, or the value $V$, or both may be fully
reduced in \betaVnf\ depending on the reduction strategy. If $N$ is not a
value or does not reduce to a value then $(\lambda x.B)N$ is a neutral which
may only be reduced to a stuck. Abstractions are values, and so are free
variables because they range over values as discussed in more detail
below. Terms can be open, reduction may `go under lambda' with free variables
possibly occurring within that scope, and final results are not values but
{\betaVnf}s.

The rationale behind the restricted reduction/conversion and the definition of
values is not merely to model call-by-value but to uphold confluence which is
a \emph{sine qua non} property of the calculus because it upholds the
consistency of the proof-theories.  Intuitively, the rationale is \emph{to
  preserve confluence by preserving potential divergence}. To preserve
confluence, applications cannot be passed as operands unless given the
opportunity to diverge first. This point is fundamental to understanding our
approach to solvability for $\lamV$ and so the rest of this section elaborates
it.

In $\lamV$ the reduction relation $\mrelV$ is confluent \cite[App.~A2]{HS08}.
Confluence applies even for terms without \betaVnf. The implication is that
terms have at most one \betaVnf, and so terms with different \betaVnf\ are not
$\betaV$-reducible/convertible. Not every $\betaV$-reduction/conversion is
provable and the reduction/conversion proof-theory is consistent.  The proof
of confluence requires substitutivity which is the property that
reduction/conversion is preserved under substitution, \eg\ if $M \beqV N$ then
$\cas{L}{x}{M} \beqV \cas{L}{x}{N}$. In $\lamV$, permissible operands and
subjects of substitutions cannot be applications, whether arbitrary or in
\betaVnf.  Otherwise, substitutivity and confluence would not hold.  (This is
explained in \cite[p.135-136]{Plo75}, see App.~\ref{app:pure-lamV} for a
detailed discussion.) Substitutivity requires the proviso $L\in\Val$ which
explains why free variables are members of $\Val$, namely, because they range
over members of $\Val$.

For illustration, the neutral $x\,\DELTA$ cannot be passed in applications
such as $(\lambda x.y)(x\,\DELTA)$ because whether it diverges depends on what
value $x$ is. For example, substituting the value $\I$ for $x$ yields
$(\lambda x.y)(\I\,\DELTA)$ which converges to $y$. But substituting the value
$\DELTA$ for $x$ yields $(\lambda x.y)(\DELTA\DELTA)$ which diverges.
Applications must be given the opportunity to diverge before being passed, not
only to model call-by-value but because whether a neutral converges depends on
which values are substituted for its free variables. The same goes for stucks:
in the above examples $x\,\DELTA$ is actually a stuck.

\subsection{Neutrals, stucks,  and sequentiality}
\label{sec:neutrals-seq}
Before moving on we must recall that the nesting and order of neutrals confer
the sequentiality character to $\lamV$. Take the following neutrals adapted
from \cite[p.25]{Mil90} and assume $V$ and $W$ are closed values:
\begin{displaymath}
  \begin{array}{lll}
    \Ls_1 & \equiv & (x\,V)(y\,W) \\
    \Ls_2 & \equiv & (\lambda z.z(y\,W))(x\,V) \\
    \Ls_3 & \equiv & (\lambda z.(x\,V)z)(y\,W)
  \end{array}
\end{displaymath}
Respectively substituting values $X$ and $Y$ for $x$ and $y$ we get:
\begin{displaymath}
  \begin{array}{lll}
    \Ls_1' & \equiv & (X\,V)(Y\,W) \\
    \Ls_2' & \equiv & (\lambda z.z(Y\,W))(X\,V) \\
    \Ls_3' & \equiv & (\lambda z.(X\,V)z)(Y\,W)
  \end{array}
\end{displaymath}
If all $\Ls_i'$ have \betaVnf\ then it is the same and the instances are
convertible. But different reduction sequences differ on the order in which
$(X\,V)$ and $(Y\,W)$ are reduced in $\Ls_2'$ and $\Ls_3'$ and thus on which
order is the same as in $\Ls_1'$.  Under SECD reduction the $(X\,V)$ is
reduced before $(Y\,W)$ in $\Ls_1'$ and $\Ls_2'$ whereas in $\Ls_3'$ the order
is reversed. However, in a reduction sequence where abstraction bodies are
reduced before operands then $(X\,V)$ is reduced before $(Y\,W$) in $\Ls_1'$
and $\Ls_3'$ whereas in $\Ls_2'$ the order is reversed.\footnote{In this
  example we have in mind a complete reduction sequence. There is a complete
  reduction strategy of $\lamV$ that goes under lambda in such `spine' fashion
  (Section~\ref{sec:value-normal-order}).}

Suppose operators and operands were reduced in separate processors. If $x$ is
instead substituted by a value $X$ such that $X\,V$ converts to a stuck, then
we can tell on which processor reduction got stuck first. If we substitute $y$
for a value $Y$ such that $Y\,W$ diverges then one processor would diverge
whereas the other would get stuck.

As another example consider the following terms where now $V$ and $W$ are
closed values in \betaVnf:
\begin{displaymath}
  \begin{array}{lll}
    \Ls_4 & \equiv & (\lambda z.V\,W)(xx) \\
    \Ls_5 & \equiv & (\lambda z.(\lambda y.y\,W))(xx)V
  \end{array}
\end{displaymath}
Observe that $\Ls_5$ is a \betaVnf\ whereas $\Ls_4$ is not. If $V\,W$
converges to a \betaVnf\ $N$ then $(\lambda z.N)(xx)$ is a \betaVnf\ different
from $\Ls_5$. If $V\,W$ diverges then $\Ls_4$ diverges but $\Ls_5$ does not
(it is a \betaVnf). Let us now play with substitutions for the blocking
variable $x$. Substitute in $\Ls_4$ and $\Ls_5$ a closed value $X$ for $x$
such that $XX$ converges to a value:
\begin{displaymath}
  \begin{array}{lll}
    \Ls_4' & \equiv & (\lambda z.V\,W)(XX) \\
    \Ls_5' & \equiv & (\lambda z.(\lambda y.y\,W))(XX)V
  \end{array}
\end{displaymath}
In the case where $V\,W$ converges to a \betaVnf\ $N$ then $\Ls_4'$ and
$\Ls_5'$ converge to $N$, but in $\Ls_4'$ whether $(V\,W)$ is reduced before
$(XX)$ depends on whether the reduction strategy goes first under lambda,
whereas in $\Ls_5'$ the term $(XX)$ is reduced first with that same
strategy. In the case where $V\,W$ diverges, whether $\Ls_4'$ diverges before
reducing $(XX)$ also depends on whether the reduction strategy goes first
under lambda, whereas in $\Ls_5'$ the term $(XX)$ is reduced first with that
same strategy. Thus, $\Ls_4$ and $\Ls_5$ are operationally
distinguishable. For example, the concrete instantiations $(\lambda
z.\I\I)(xx)$ and $(\lambda z.(\lambda y.y\,\I))(xx)\I$ are operationally
distinguishable (here $V\equiv\I$, $W\equiv\I$, and $\I\I$ converges to a
\betaVnf).

Neutral terms differ on the point at which a free variable pops up, that is,
on the point of potential divergence. Stucks are only fully reduced neutrals
that keep that point of divergence. Terms with neutrals that may convert to
the same \betaVnf\ when placed in the same closed context are nonetheless
operationally distinguishable when placed in an open context. And the choice
of substitutions for the blocking free variables is important. Keep this in
mind when reading the following sections.

%%%%%%%%%%%%%%%%%%%%%%%%%%%%%%%%%%%%%%%%%%%%%%%%%%%%%%%%%%%%%%%%%%%%%%%%%%%%%%
\section{An overview of \texorpdfstring{$v$}{v}-solvability}
\label{sec:v-solv}
Solvability for $\lamV$ is first studied in \cite{PR99} where a definition of
\mbox{$v$-solvable} term is introduced which adapts to $\lamV$ the
\textsc{SolI} definition of solvability for $\lamK$.
\begin{defi}[$v$-solvability]
\label{def:v-solv}
A term $M$ is $v$-solvable in $\lamV$ iff there exist closed values
$N_1\in\Val^0$, \ldots, $N_k\in\Val^0$ with $k \geq 0$ such that $(\lambda
x_1\ldots x_n.M)N_1\cdots N_k \beqV \I$ where $\FV(M)=\{x_1,\ldots,x_n\}$.
\end{defi}
The definition can be stated alternatively in terms of the head contexts of
Section~\ref{sec:open-open} by requiring the $C_i$'s and $N_i$'s in the head
contexts to be closed values instead of closed terms. The provisos
$N_i\in\Val^0$ could have been omitted because they are required by the
$\betaV$-conversion to the closed value~$\I$.  In line with the discussion in
Section~\ref{sec:open-open}, an open head context whose free variables are
discarded in the conversion can also be used, and so it is in
\cite[p.9]{AP12}.

Adapting \textsc{SolI} to $\lamV$ instead of \textsc{SolN} is surprising
because, as anticipated in Section~\ref{sec:eq-defs}, the two properties that
justify the equivalence between \textsc{SolI} and \textsc{SolN} in $\lamK$ do
not hold in $\lamV$. (And as discussed in Section~\ref{sec:open-open}, the use
of a closed and closing head context is excessive, but more on this below.)

First, $\I\,X \beqV X$ holds iff $X$ has a value. Assuming such proviso, the
\textsc{SolX} equivalent of Def.~\ref{def:v-solv} is that a term is
$v$-solvable iff it is convertible by application not to any term but to any
\emph{value}. Indeed, if $M$ is $v$-solvable then \mbox{$(\lambda x_1\ldots
  x_n.M)N_1\cdots N_k \beqV \I$} and, by compatibility, $(\lambda x_1\ldots
x_n.M)N_1\cdots N_k\,X \beqV \I\,X$ for any $X\in\Lambda$.  The conversion
$(\lambda x_1\ldots x_n.M)N_1\cdots N_k\,X \beqV X$ is obtained by
transitivity with $\I\,X \beqV X$ iff $X$ has a value.

Second, the adaptation of Lemma~\ref{lem:4.1Wad76} to $\lamV$ does not hold.
\begin{stmt}[Adapts Lemma~\ref{lem:4.1Wad76} to $\lamV$]
\label{prop:4.1Wad76-lamV}
If $M\in\Lambda^0$ has a \betaVnf\ then for all $X\in\Lambda$ there exist
operands $X_1\in\Lambda$, \ldots, $X_k\in\Lambda$ with $k\geq 0$ such that
$M\,X_1 \cdots X_k \beqV X$.
\end{stmt}
This statement does not hold even with $X_i$ and $X$ values, whether open or
closed.  The controversial term $\U\equiv\lambda x.(\lambda
y.\DELTA)(x\,\I)\DELTA$ mentioned in \cite{PR99} is one possible
counter-example. (Notice the close resemblance to the term $\Ls_5$ in
Section~\ref{sec:neutrals-seq}.)  This term is a closed value and a
\betaVnf. It is an abstraction with a stuck body.  There is no operand $X_1$,
let alone further operands, that lets us convert $\U$ to any given $X$ whether
arbitrary, a value, or a closed value.

Suppose $X_1\in\Val^0$. Then $\U\,X_1$ converts to $(\lambda y.\DELTA)
(X_1\,\I)\DELTA$. If $(X_1\,\I)$ diverges then the latter diverges.  If
$(X_1\,\I)$ converts to a closed value $V$ then $(\lambda y.\DELTA)\,V\DELTA$
converts to $\DELTA\DELTA\equiv\OMEGA$ which diverges. However, $\U\,X_1$
converts to a \betaVnf\ if $(X_1\,\I)$ converts to a stuck. But the shape of
the \betaVnf, namely $(\lambda y.\DELTA)(\ldots)\DELTA$, is determined by the
shape of $\U$. Ony the concrete \betaVnf\ obtained depends on the choice of
\emph{open} value $X_1$ that generates the stuck. For example:
$X_1\equiv\lambda x.z\,\I$ leads to $(\lambda y.\DELTA)(z\,\I)\DELTA$ whereas
$X_1\equiv\lambda x.(\lambda x.x)(z\,\K)$ leads to $(\lambda
y.\DELTA)((\lambda x.x)(z\,\K))\DELTA$, etc. We cannot send $\U$ to any
arbitrary \betaVnf. The only degree of freedom is $X_1$.

The term $\U$ is controversial because, although a \betaVnf, it is considered
operationally equivalent to $\lambda x.\OMEGA$ in \cite{PR99}. Certainly,
$\U\,X_1$ and $(\lambda x.\OMEGA)X_1$ diverge for all $X_1\in\Val^0$. But as
illustrated in the last paragraph, $\U$ and $\lambda x.\OMEGA$ are
operationally distinguishable in an open context: there exists $X_1\in\Val$
such that $\U\,X_1$ converts to a \betaVnf, but there is no $X_1\in\Val$ such
that $(\lambda x.\OMEGA)X_1$ converts to a \betaVnf. The difference between
$\U$ and $\lambda x.\OMEGA$ is illustrated by the old chestnut `toss a coin,
heads: you lose, tails: toss again'. We can pass a value to $\U$ to either
diverge immediately or to postpone divergence, but this choice is not possible
for $\lambda x.\OMEGA$ which diverges whatever value passed. And since $\U$ is
a \betaVnf, it should be by definition solvable in $\lamV$.

The restriction of operands to elements of $\Val^0$ is natural in the setting
of SECD's weak reduction of closed terms where final results are closed
values. This is the setting considered in \cite{PR99} where the proof-theory
is not $\lamV$'s but consists of equations `$M = N$ iff $M$ and $N$ are
operationally equivalent under SECD reduction'. However, $v$-solvability
(Def.~\ref{def:v-solv}) is defined for $\lamV$ and its proof-theory, not the
alternative pure-SECD-theory. Several problems arise. First, closed values
such as $\U$ and $\lambda x.\OMEGA$ which are definite results of SECD are
$v$-unsolvable, so $v$-solvability is not synonymous with operational
relevance. Second, there is a $v$-unsolvable $\U$ that is nevertheless a
\betaVnf\ of $\lamV$.  As discussed in the introduction, the blame is
mistakenly put on $\lamV$, not on $v$-solvability.

The operational relevance of final results is partly recovered in
\cite[p.21]{PR99} by adapting to $v$-unsolvables the notion of order of a term
\cite{Lon83,Abr90} in the following fashion: a \mbox{$v$-unsolvable} $M$ is of
order $n$ iff it reduces under the so-called `inner machine' to $\lambda
x_1\ldots x_n.B$ where $n$ is maximal. That is, $M$ reduces to a value with
$n$ lambdas. If $M$ has order $0$ then it does not reduce to a value. If $M$
has order $n>0$ then $M$ accepts $n-1$ operands and reduces to a value.  For
example, $\OMEGA$ has order 0, and $\lambda x.\OMEGA$ and $\U$ have order
1. With this notion of order, definite results include $v$-solvables and
\mbox{$v$-unsolvables} of order $n>0$. This corresponds with the behaviour of
SECD. The \mbox{$v$-unsolvables} of order $0$ denote the least element of the
model $\Model{H}$ of \cite{EHR92} and can be equated without loss of
consistency.

However, the `inner machine' is a call-by-value reduction strategy of $\lamK$.
It performs $\betaK$-reduction, reducing redexes when the operand is not a
value. Furthermore, $v$-unsolvables of order $n>0$, which according to
\cite{PR99} are operationally irrelevant because no arbitrary result can be
obtained from them, are definite results. These \mbox{$v$-unsolvables} cannot
be consistently equated \cite{PR99} and thus the model $\Model{H}$ is not
sensible. Moreover, it is not semi-sensible since some $v$-solvables can be
equated to $v$-unsolvables (Thm.~5.12 in \cite[p.22]{PR99}). Finally, the
operational characterisation of \mbox{$v$-solvability}, namely having a
$v$-\hnf, is given by the so-called `ahead machine' which is also a reduction
strategy of $\lamK$, not of $\lamV$.

The reason why $v$-solvability does not capture operational relevance in
$\lamV$ is because it is based on \textsc{SolI} which requires the
\emph{universally} (any $X$) quantified Lemma~\ref{def:v-solv} to hold. The
solution lies in adapting to $\lamV$ the \emph{existentially} (has some
\betaVnf) quantified \textsc{SolN} definition with open and non-closing
contexts. As we shall see, there are two ways to solve a term in $\lamV$. One
is to apply it to suitable values to obtain any given value (or closed value
as in \mbox{$v$-solvability}). We call this to \emph{transform} the
application. Another is to pass suitable values to obtain some \betaVnf.  We
call this to \emph{freeze} the application. Terms like $\U$ cannot be
transformed but frozen.

In \cite[p.36]{RP04} it is the open body of $\U$, \ie\
$\mathbf{B}\equiv(\lambda y.\DELTA)(x\,\I)\DELTA$, what is considered
operationally equivalent to $\OMEGA$. Now, $\B$ is not a value, but it is a
\betaVnf, a definite result of $\lamV$. The difference between $\B$ and
$\OMEGA$ lies in the value substituted for $x$. The intuition is best
expressed using the following  experiment paraphrased from
\cite[p.4]{Abr90}:
\begin{quote}
  Given [an arbitrary] term, the only experiment of depth $1$ we can do is to
  evaluate [weakly] and see if it converges to some abstraction [or to some
  neutral subsequently closed to some abstraction] $\lambda x.M_1$. If it does
  so, we can continue the experiment to depth 2 by supplying [an arbitrary
  value $N_1$ that may be open] as input to $M_1$, and so on. Note that what
  the experimenter can observe at each stage is only the \emph{fact} of
  convergence, not which term lies under the abstraction. [Note that the term
  \emph{reports} the need to provide a value for the blocking free variable by
  closing the neutral to an abstraction.]
\end{quote}

%%%%%%%%%%%%%%%%%%%%%%%%%%%%%%%%%%%%%%%%%%%%%%%%%%%%%%%%%%%%%%%%%%%%%%%%%%%%%%
\section{Introducing %
  \texorpdfstring{$\lamV$}{lambda-V}-solvability}
\label{sec:lamV-solv}
We have seen that terms like the $\Ls'_i$ of Section~\ref{sec:neutrals-seq},
or $\U$ and $\lambda x.\OMEGA$ in the previous section, are operationally
distinguishable in open contexts. We thus define solvability in $\lamV$ by
adapting \textsc{SolF} to that calculus.
\begin{defi}[\textsc{SolF$_\vv$}] A term $M\in\Lambda$ is solvable in
  $\lamV$ iff there exists $N\in\VNF$ and there exists a function context
  $\ctx{F}\hole$ such that $\ctx{F}[M]\beqV N$.
\end{defi}
Notice that operands in function contexts may be values if so wished.
Hereafter we abbreviate `$M$ is solvable in $\lamV$' as `$M$ is
$\lamV$-solvable'.

The set of $\lamV$-solvables is a proper superset of the union of the set of
terms with \betaVnf\ and the set of \mbox{$v$-solvables}. A witness example is
$\T_1 \equiv (\lambda y.\DELTA)(x\,\I)\DELTA(x(\lambda x.\OMEGA))$. This term
has no \betaVnf. This term is not $v$-solvable: there is no closed and closing
head context sending $\T_1$ to $\I$, or to a closed value, or to a closed
\betaVnf. However, the function context $\ctx{F}\hole\equiv(\lambda
x.\hole)(\lambda x.z\,\I)$ sends $\T_1$ to the \betaVnf\ $(\lambda
x.\DELTA)(z\,\I)\DELTA(z\,\I)$. Therefore $\T_1$ is \mbox{$\lamV$-solvable.}

Notice that $\T_1$ has $\B$ as subterm, with the blocking variable $x$ of $\B$
the same blocking variable of the neutral $x(\lambda x.\OMEGA)$. The use of
the same blocking variable illustrates that the function context in
\textsc{SolF}$_\vv$ has to be open. There is no closed function context (nor
head context) sending $\T_1$ to a \betaVnf\ since substituting a closed value
for $x$ would make $\B$ diverge. In contrast, the free variable $z$ in
$\ctx{F}\hole$ above is key to produce a stuck. We anticipated in
Section~\ref{sec:open-open} that adapting \textsc{SolH} to $\lamV$ leaves
solvable terms behind. The terms $\U$ and $\T_1$ are two witness examples.

We now connect $\lamV$-solvability and operational relevance with effective
use in $\lamV$, as we did for $\lamK$ in Section~\ref{sec:effective-use-lamK}.
To this end we adapt to $\lamV$ the notion of `order of a term' \cite{Lon83}.
\begin{defi}[Order of a term in $\lamV$]
  \label{def:oder-term-lamV}
  A term $M\in\Lambda$ is of order $0$ iff there is no $N$ such that
  $M\beqV\lambda x.N$. A term $M\in\Lambda$ is of order $n+1$ iff
  $M\beqV\lambda x.N$ and $N$ is of order $n$. In the limit, \ie\ when a
  maximum natural $k$ does not exist such that $M\beqV\lambda x_1\ldots
  x_k.N$, we say $M$ is of order $\omega$.
\end{defi}
This definition differs from the one in \cite[p.21]{PR99}. The latter is for
\mbox{$v$-unsolvables} and uses the `inner machine' which is a reduction
strategy of $\lamK$ (Section~\ref{sec:v-solv}). Ours is for arbitrary terms
(not just $\lamV$-unsolvables) and uses $\betaV$-conversion.

The order of a term is an ordinal number that comprises the finite ordinals
(\ie\ the naturals) and the first limit ordinal $\omega$. An example of a term
of order $\omega$ is $\Y\,\K$ where $\Y$ is Curry's fixed-point combinator
(see Prop.~2.7.(iv) in \cite[p.6]{Abr90} and Ex.~2 in
\cite[p.502]{Wad76}). The term $\Y\,\K$ $\betaV$-converts to $\lambda
x_1\ldots x_k.\Y\,\K$ with $k$ arbitrarily large. Notice that a term of order
$\omega$ has no \betaVnf\ and is $\lamV$-unsolvable.

With this notion of order at hand we can now state our version of
\textsc{SolC} for $\lamV$.
\begin{defi}[\textsc{SolC$_\vv$}]
  \label{def:solvV}
  A term $M\in\Lambda$ of order $n$ is solvable in $\lamV$ iff there exists a
  context $\ctx{C}\hole$ such that $\ctx{C}[M]\beqV N$ for some $N\in\V\NF$,
  and not for all $X\in\Lambda$ of order $m\geq n$ it is the case that
  $\ctx{C}[X]\beqV N$.
\end{defi}
Note that $X\in\Val$ is allowed by the definition.

As was the case in $\lamK$ (Section~\ref{sec:effective-use-lamK}), the piece
that lets us obtain \textsc{SolC}$_\vv$ from \textsc{SolF}$_\vv$ is a
genericity lemma which in $\lamV$ has to take into account the order of
$\lamV$-unsolvables.
\begin{lem}[Partial Genericity Lemma]
  \label{lem:partial-genericity}
  Let $M\in\Lambda$ be of order $n$ and $N\in\VNF$. $M$ is $\lamV$-unsolvable
  implies that for all contexts $\ctx{C}\hole$, if $\ctx{C}[M]\beqV N$ then
  for all $X\in\Lambda$ of order $m\geq n$ it is the case that
  $\ctx{C}[X]\beqV N$.
\end{lem}
We postpone the proof to Section~\ref{sec:partial-genericity-lemma} and focus
here on the intuitions. The lemma tells us that $\lamV$-unsolvables of order
$n$ are \emph{partially} generic, \ie\ they are generic for terms of order $m
\geq n$. A $\lamV$-solvable can be used effectively to produce a \betaVnf\
therefore \mbox{$\lamV$-solvability} is synonymous with operational relevance.
However, not all $\lamV$-unsolvables are totally undefined. Only
$\lamV$-unsolvables of order $0$ are totally undefined. A $\lamV$-unsolvable
of order $n$ cannot be used effectively to produce a \betaVnf, but it can be
used trivially (discarded) after receiving at most $n-1$ operands. Hence, it
is partially defined.

For example, take $M\equiv\lambda x.\lambda y.\OMEGA$.  This term is
$\lamV$-unsolvable of order $2$. The context $\ctx{C}\hole\equiv(\lambda
x.(\lambda y.\I)(x\,\DELTA))\hole$ uses $M$ first `administratively' (\ie\
passes $\DELTA$ to it) and then `trivially' (\ie\ discards the result) such
that $\ctx{C}[M]\beqV \I$. Replacing $M$ with a totally undefined term like
$\OMEGA$ would make $\ctx{C}[\OMEGA]$ diverge. But since $\ctx{C}\hole$ uses
$M$ only up to passing one argument, $M$ could be replaced by any term $X$ of
order $2$ and still $\ctx{C}[X]\beqV\I$.

The Partial Genericity Lemma is stated as an implication but, as was the case
with Lemma~\ref{lem:lamK-genericity-lemma}, it is an equivalence. Clearly, if
$M$ is $\lamV$-solvable then there exists $\ctx{C}\hole\equiv\ctx{F}\hole$
such that $\ctx{F}[M]\beqV N$, and by the shape of $\ctx{F}\hole$ it is not
the case that for all $X\in\Lambda$ of order $m\geq n$, $\ctx{F}[X]\beqV
N$. Take for instance $\ctx{F}[\lambda x_1\ldots x_m.\OMEGA]$ which diverges.
% (Note that if $M$ is $\lamV$-solvable and $\ctx{C}[M] \beqV N$ holds then
% $\ctx{C}[X]\beqV N$ should not hold for terms $X$ that are not convertible
% to $M$ unless $M$ is not effectively used in $\ctx{C}\hole$.)
Stated as an equivalence, the Partial Genericity Lemma coincides with
\textsc{SolC}$_\vv$ when read in the inverse.

Pure $\lamV$ still has `functional character' \cite{CDV81,EHR92} but its
notion of operational relevance takes into account trivial uses of terms that
occur inside operands of other terms up to administratively passing them a
number of operands. More precisely, if a term occurs inside the operand of
another term then it has `negative polarity'. Otherwise it has `positive
polarity'. The import of polarity for operational relevance is inherent to the
duality between call-by-name and call-by-value \cite{CH00}. Subterms with
positive polarity are used effectively. Subterms with negative polarity may or
may not occur eventually with positive polarity, in which case they would
respectively be used effectively or trivially (perhaps after receiving some
operands). The partially generic terms may only be used trivially (up to order
$n$) to produce a \betaVnf\ if they occur with negative polarity.

Partially generic terms can be equated attending to their order without loss
of consistency. More precisely, given the set
\begin{displaymath}
  \HV_0 = \{M = N~|~M,N\in\Lambda^0\ \text{are $\lamV$-unsolvables of the same
    order}\}
\end{displaymath}
a consistent extended proof-theory $\HV$ results from adding $\HV_0$'s
equations as axioms to $\lamV$ (\ie\ $\HV = \HV_0+\lamV$). The consistency of
$\HV$ is proved in Section~\ref{sec:consistent-lamV-theory}. We say that a
consistent extension where $\lamV$-unsolvables of the same order are equated
(\ie\ contains $\HV$) is \emph{$\omega$-sensible}.

Since the operational experiments that we have in mind
(Sections~\ref{sec:neutrals-seq} and~\ref{sec:v-solv}) distinguish
sequentiality features, no $\omega$-sensible \emph{functional models} (\eg,
models that are solution to the domain equation ${D \cong [D \to_\bot D]}$ for
strict functions) seem to exist. However, we conjecture the existence of
$\omega$-sensible models that may resemble the `sequential algorithms' of
\cite{BC82}. The notion of operational relevance in $\lamV$ that we advocate
calls for increased `separating capabilities' (in the spirit of \cite{Cur07})
that $\omega$-sensible models would exhibit. Such capabilities are not present
in existing models for `lazy' call-by-value (\eg, the model H in \cite{EHR92}
based on the solution to the domain equation ${D \cong [D \to_\bot D]_\bot}$
for lifted strict functions). We also conjecture that existing functional
models could be constructed from $\omega$-sensible models via some quotient
that blurs the differences in sequentiality.

As for the operational characterisation of $\lamV$-solvables, that is, a
reduction strategy of $\lamV$ that terminates iff the input term is
\mbox{$\lamV$-solvable}, we postpone the discussion to
Section~\ref{sec:operational-characterisation}.

%%%%%%%%%%%%%%%%%%%%%%%%%%%%%%%%%%%%%%%%%%%%%%%%%%%%%%%%%%%%%%%%%%%%%%%%%%%%%%
\section{Towards the Partial Genericity Lemma}
\label{sec:partial-genericity-lemma}
Our proof of the Partial Genericity Lemma is based on the proof of $\lamK$'s
Genericity Lemma presented in \cite{BKV00} that uses origin tracking. Given a
reduction sequence $M\relK\ldots\relK N$ with $N\in\NF$, origin tracking
traces the symbols in $N$ back to a prefix of $M$ (\ie\ a `useful' part) which
is followed by a lower part (\ie\ the `garbage') that does not affect the
result $N$. The tracking mechanism employs a refinement of Lévy-labels
\cite{Lev75}.

In our case the reduction sequence is $M\relV\ldots\relV N$ with
$N\in\VNF$. Instead of tracking the symbols in $N$ back to the the useful part
in $M$, we mark as garbage a predefined subterm in $M$, namely, the
$\lamV$-unsolvable of order $n$ that we want to test for partial genericity.
We track this subterm forwards and check that it is discarded in the reduction
sequence before passing $n$ operands to it. To this end we need two main
ingredients: (i) a reduction strategy that is complete with respect to
\betaVnf\ (Section~\ref{sec:value-normal-order}) and (ii) a tracking mechanism
that keeps count of the number of operands that are passed to a predefined
subterm (Section~\ref{sec:labelling}). We prove that the predefined term is
discarded by the complete reduction strategy after receiving at most $n-1$
operands (Section~\ref{sec:pgl-stated}). Confluence allows us to generalise
from the reduction strategy to any reduction sequence ending in \betaVnf.

\subsection{Value normal order}
\label{sec:value-normal-order}
The first ingredient we need is a reduction strategy of reference that is
complete with respect to \betaVnf.  We define one such strategy and call it
\emph{value normal order} because we have defined it by adapting to $\lamV$
the results in \cite{BKKS87} relative to the complete \emph{normal order}
strategy of $\lamK$ mentioned in Section~\ref{sec:prelim}.  Those results are
collected in App.~\ref{app:head-and-head-spine} for ease of reference. In this
section we introduce their analogues for $\lamV$. The unacquainted reader may
find it useful to read App.~\ref{app:head-and-head-spine} and this section in
parallel.

We advance that value normal order is not quite the same strategy as the
complete reduction strategy of $\lamV$ named $\rel{\Gamma}^p$ that is obtained
as an instantiation of the `principal reduction machine' of \cite{RP04}. The
latter reduces the body and operator of a block in right-to-left fashion
whereas value normal order uses the more natural left-to-right order (see
Section~\ref{sec:related-work} for details). This difference does not affect
completeness because both strategies entail standard reduction sequences (a
notion defined in \cite[p.137]{Plo75} for the applied $\lamV$ and adapted to
pure $\lamV$ in Def.~\ref{def:v-standard} below). For every $\lamV$ reduction
sequence from $M$ to $N$, there exists a standard reduction sequence that
starts at $M$ and ends at $N$. A reduction strategy that entails standard
reduction sequences and that arrives at a \betaVnf\ is complete. And standard
reduction sequences are not unique (Section~\ref{sec:complete-standard}).
% The standard reduction sequences of a term end in the term's \betaVnf\ if it
% has any. A reduction strategy that entails standard reduction sequences is
% complete. And standard reduction sequences are not unique
% (Section~\ref{sec:complete-standard}).

Normal order can be defined as follows. The \emph{active components} of a term
\cite[Def.~2.3]{BKKS87} (\ie\ the maximal subterms that are not in \hnf) are
considered in left-to-right fashion and reduced by \emph{head reduction}
\cite[Def.~8.3.10]{Bar84}. At the start, the input term is the only active
component if it is not a \hnf. Once a \hnf\ is reached its active components
occur as subterms inside a `frozen' \betaKnf\ context. Every time the \hnf\ of
an active component is reached, the subsequent active components in it (if
any) are recursively considered in left-to-right-fashion. We define value
normal order by adapting this pattern to $\lamV$. In particular, we adapt the
definition of needed redex, of active component, and of head reduction, whose
analogue we have called `chest reduction' following the convention of
\cite[Sec.~4]{BKKS87} of considering the abstract syntax tree of a term and an
anatomical analogy for terms.

First we adapt the notion of needed redex \cite[p.212]{BKKS87} to $\lamV$:
\begin{defi}[Needed redex in $\lamV$]
  \label{def:betaV-needed}
  Let $M\in\Lambda$ and $R$ a \mbox{$\betaV$-redex} in $M$. $R$ is needed iff
  every reduction sequence of $M$ to \betaVnf\ contracts (some residual of)
  $R$.
\end{defi}
The \emph{chest} and \emph{ribcage} of a term provide progressively better
approximations to the set of needed $\betaV$-redexes of a term. The chest of
the term contains the head of the term and the outermost \emph{ribs}, that is,
all the nodes connected by application nodes to the head of the term save for
the \emph{rib ends}. The rib ends are the nodes descending through lambda
nodes from the ribs. The ribcage of a term consists of the head spine and the
ribs connected to the head spine, that is, all the nodes connected by
application nodes to the head spine of the term save for the rib ends.
Fig.~\ref{fig:chest-ribcage} illustrates with an example that is further
developed after the following formal definition of chest and ribcage.
% \footnote{The chest and ribcage of a term provide
%   progressively better approximations to the set of needed $\betaV$-redexes of
%   a term and can be used to define needed reduction for $\lamV$. This result
%   is immaterial here and the interested reader is referred to
%   App.~\ref{app:conclusion-future-work}.}

In Def.~\ref{def:chest-ribcage} below we define the functions $\bv$, $\ch$, and
$\rc$. The last two underline respectively the chest and the ribcage of a
term.  Both rely on auxiliary $\bv$ related to call-by-value as explained
further below.
\begin{defi}[Chest and ribcage]
  \label{def:chest-ribcage}
  Functions $\ch$ and $\rc$ underline the chest and the ribcage of a term
  respectively.
  \begin{displaymath}
    \begin{array}{rcl}
      \bv(x)           &=& \underline{x}\\
      \bv(\lambda x.B) &=& \underline{\lambda x}.B\\
      \bv(M\,N)        &=& \bv(M)\bv(N)\\[4pt]
      \ch(x)           &=& \underline{x}\\
      \ch(\lambda x.B) &=& \underline{\lambda x}.\ch(B)\\
      \ch(M\,N)        &=& \bv(M)\bv(N)\\[4pt]
      \rc(x)           &=& \underline{x}\\
      \rc(\lambda x.B) &=& \underline{\lambda x}.\rc(B)\\
      \rc(M\,N)        &=& \rc(M)\bv(N)
    \end{array}
  \end{displaymath}
  A $\betaV$-redex is chest (resp. ribcage) if the outermost lambda of it is
  underlined by function $\ch$ (resp. $\rc$).
\end{defi}
Function $\bv$ underlines the outermost lambda of the $\betaV$-redexes that
are reduced by the call-by-value strategy of pure $\lamV$
(Def.~\ref{def:call-by-value}). This strategy differs from its homonym in
\cite[p.136]{Plo75} which is for an applied version of the calculus. See
\cite{Fel87,Ses02,RP04} for details on the difference. The chest and ribcage
$\betaV$-redexes realise the idea that operands in applications must be
reduced to a value.

\begin{figure}
  \begin{center}
    \begin{tikzpicture} [
      level distance=1cm,
      level 2/.style={sibling distance=2.5cm},
      level 3/.style={sibling distance=1.5cm},
      level 4/.style={sibling distance=1.5cm},
      level 5/.style={sibling distance=1.5cm},
      chest/.style={very thick},
      ribcage/.style={dotted},
      norm/.style={thin,solid}
      ]
      \begin{scope}
        \node (lx) {$\lambda x$}
        child[chest]{ node (a1) {$@$}
          child{ node (a2) {$@$}
            child{ node (ly) {$\lambda y$}
              child[ribcage]{ node (a3) {$@$}
                child{ node (y) {$y$} }
                child{ node (a4) {$@$}
                  child{ node (lz) {$\lambda z$}
                    child[norm]{ node (m1) {$M_1$} }}
                  child{ node (x3) {$x$} }}}}
            child{ node (x1) {$x$} }}
          child{ node (a3) {$@$}
            child{ node (lt) {$\lambda t$}
              child[norm]{ node (m2) {$M_2$} }}
            child{ node (x2) {$x$}}}};
      \end{scope}
    \end{tikzpicture}
  \end{center}
  \caption{Chest (thick edges) and ribcage (thick edges and dotted edges) of the
    term $\lambda x.(\lambda y.y((\lambda z.M_1)x))x((\lambda t.M_2)x)$.}
  \label{fig:chest-ribcage}
\end{figure}

As an example, consider the term whose abstract syntax tree is depicted in
Fig~\ref{fig:chest-ribcage}. The chest (thick edges in the figure) is
underlined in $\underline{\lambda x.(\lambda y}.y((\lambda z.M_1)x))
\underline{x((\lambda t}.M_2)
{\color{black}\underline{{\color{white}\underline{{\color{black}x}}}}})$.
% \begin{displaymath}
%   \underline{\lambda
%     x.(\lambda y}.y((\lambda z.M_1)x)) \underline{x((\lambda t}.M_2)
%   {\color{black}\underline{{\color{white}\underline{{\color{black}x}}}}})
% \end{displaymath}
The ribcage (thick edges and dotted edges) is underlined in
$\underline{\lambda x.(\lambda y.y((\lambda z}.M_1)\underline{x))x ((\lambda
  t}.M_2) {\color{black}\underline{{\color{white}
      \underline{{\color{black}x}}}}})$.
% \begin{displaymath}
%   \underline{\lambda x.(\lambda
%     y.y((\lambda z}.M_1)\underline{x))x ((\lambda t}.M_2)
%   {\color{black}\underline{{\color{white}\underline{{\color{black}x}}}}})
% \end{displaymath}
The subterms $M_1$ and $M_2$ are the rib ends.  The subterms $(\lambda
y.y((\lambda z.M_1)x))x$ and $(\lambda t.M_2)x$ are both chest and ribcage
$\betaV$-redexes. (The former is also a head and head-spine $\beta$-redex, and
the latter is neither head nor head-spine.)  The subterm $(\lambda z.M_1)x$ is
a ribcage $\betaV$-redex but it is neither a chest \mbox{$\betaV$-redex}, nor
a head or head spine $\beta$-redex.

We now define call-by-value and chest reduction using a (context-based)
reduction semantics~\cite{Fel87} which is a handy device for defining
small-step reduction strategies. It consists of EBNF-grammars for terms,
irreducible forms, and reduction contexts, together with a contraction rule
for redexes within context holes. The reduction strategy is defined by the
iteration of single-step reductions which consist of (i) uniquely decomposing
the term into a reduction context plus a redex within the hole, (ii)
contracting the redex within the hole and, (iii) recomposing the resulting
term. The iteration terminates iff the term is irreducible.

Call-by-value is the strategy that contracts the leftmost $\betaV$-redex that
is not inside an abstraction \cite[p.42]{Fel87}. Chest reduction is the
strategy that contracts the leftmost chest $\betaV$-redex. Observe that the
reduction contexts of chest reduction contain the reduction contexts of
call-by-value.
\begin{defi}[Call-by-value strategy]
\label{def:call-by-value}
The call-by-value strategy $\rel{\vv}$ is defined by the the following
reduction semantics:
  \begin{displaymath}
    \begin{array}{rcl}
      \ctx{BV}\hole &::=&\hole~|~\ctx{BV}\hole\,\Lambda~|~
      \VWNF\,\ctx{BV}\hole\\
      \VWNF   &::=& \Val~|~ \NeuW \\
      \NeuW   &::=& x\,\VWNF\,\{\VWNF\}^* ~|~  \BlockW\,\{\VWNF\}^* \\
      \BlockW &::=& (\lambda x.\Lambda)\,\NeuW \\
      \multicolumn{3}{l}{\ctx{BV}[(\lambda x.B)N] \rel{\vv}
        \ctx{BV}[\cas{N}{x}{B}] \quad\quad\textup{with}\ N\in\Val}
  \end{array}
\end{displaymath}
\end{defi}\medskip

\noindent The set $\VWNF$ of $\betaV$-weak-normal-forms ({\vwnf}s for short) consists of
the terms that do not have $\betaV$-redexes except under abstraction. It
contains values and neutrals in \vwnf.
\begin{defi}[Chest reduction]
  \label{def:chest-reduction}
  The chest-reduction strategy $\rel{\stgy{ch}}$ is defined by the following
  reduction semantics:
  \begin{displaymath}
    \begin{array}{l}
      \ctx{CH}\hole\   ::=\ \hole\ |\ \ctx{BV}\hole\,\Lambda\ |\
      \VWNF\,\ctx{BV}\hole\ |\ \lambda x.\ctx{CH}\hole \\
      \\
      \ctx{CH}[(\lambda x.B)N] \rel{\stgy{ch}}
        \ctx{CH}[\cas{N}{x}{B}] \quad\quad\textup{with}\ N\in\Val
    \end{array}
  \end{displaymath}
\end{defi}\medskip

\noindent The set $\CHNF ::= x~|~\lambda x.\CHNF~|~\NeuW$ of chest normal forms
({\chnf}s for short) consists of variables, abstractions with body in \chnf,
and neutrals in \vwnf. A chest normal form has the following shape:
\begin{displaymath}
  % \begin{array}{l}
  %   \lambda x_1\ldots x_n.\\
  %   \quad(\lambda y_p.B_p)\\
  %   \quad\quad(~\cdots((\lambda y_1.B_1)(z\,W_1^0\cdots W_{m_0}^0)
  %   W_1^1\cdots W_{m_1}^1)\\
  %   \quad\quad\;\;\cdots~)
  %   W_1^p\cdots W_{m_p}^p
  % \end{array}
  \lambda x_1\ldots x_n.(\lambda y_p.B_p)
  (~\cdots((\lambda y_1.B_1)(z\,W_1^0\cdots W_{m_0}^0)W_1^1\cdots W_{m_1}^1)
  \cdots~)W_1^p\cdots W_{m_p}^p
\end{displaymath}
where $n\geq 0$, $p\geq 0$, $m_1\geq 0$, \ldots, $m_p\geq 0$, and $W_i^j$ are
in \vwnf. We say that $M\,W_1^j\cdots W_{m_j}^j$ is an \emph{accumulator},
where $M$ is its leftmost operator which is either a variable or a block.  The
operand of the block in an accumulator could be, in turn, an accumulator, and
accumulators are nested in this way, where the innermost one has a variable as
its leftmost operator. We call this variable the \emph{blocking variable},
which is variable $z$ in the term above.

The term $\T_1\equiv(\lambda y.\DELTA)(x\,\I)\DELTA(x(\lambda x.\OMEGA))$
introduced in Section~\ref{sec:lamV-solv} is an example of a \chnf\ that has
no \betaVnf.
\begin{defi}[Ribcage reduction]
  \label{def:ribcage-reduction}
  The ribcage-reduction strategy $\rel{\stgy{rc}}$ is defined by the following
  reduction semantics:
  \begin{displaymath}
    \begin{array}{l}
      \ctx{RC}\hole\   ::=\ \hole\ |\ \ctx{RC}\hole\,\Lambda\ |\
      \VWNF\,\ctx{BV}\hole\ |\ \lambda x.\ctx{RC}\hole \\
      \\
      \ctx{RC}[(\lambda x.B)N] \rel{\stgy{rc}}
      \ctx{RC}[\cas{N}{x}{B}] \quad\quad\textup{with}\ B\in\CHNF\
      \textup{and}\ N\in\Val
    \end{array}
  \end{displaymath}
\end{defi}\medskip

\noindent Ribcage reduction delivers a \chnf\ if the term has some. (A term can convert
to several $\betaV$-convertible {\chnf}s that differ in the rib ends.)  The
only difference with respect to chest reduction is that ribcage reduction
contracts the body of a $\betaV$-redex to \chnf\ before contracting the
\mbox{$\betaV$-redex}.
\begin{defi}[Active components in $\lamV$]
  \label{def:active-components-lamV}
  The $\lamV$-active components of $M\in\Lambda$ are the maximal subterms of
  $M$ that are not in \chnf.
\end{defi}
Paraphrasing \cite[p.195]{BKKS87} to the $\lamV$ case:
\begin{quote}
  The word ``active'' refers to the fact that the [$\lamV$-active] components
  are embedded in a context which is ``frozen'', \ie\ a [\betaVnf] when the
  holes are viewed as variables. (This frozen context of $M$ is the trivial
  empty context if $M$ is not a [\chnf].)
\end{quote}
A \betaVnf\ has no $\lamV$-active components. The \mbox{$\lamV$-active}
components of a term are disjoint. For example, the $\lamV$-active components
of $\lambda x.x(\lambda y.\I\,\I)(\lambda z.(\lambda t.z)(x\,y)(\I\,\I))$ are
the subterms $\lambda y.\I\,\I$ and $\lambda z.(\lambda t.z)(x\,y)(\I\,\I)$.

Value normal order is defined in terms of chest reduction as follows. The
$\lamV$-active components of the term are considered in left-to-right fashion
and reduced by chest reduction. (The following lines paraphrase the ones for
normal order written at the beginning of the section.) At the start, the input
term is the only $\lamV$-active component if it is not a \chnf. Once a \chnf\
is reached, the \mbox{$\lamV$-active} components in it (if any) are subterms
inside a `frozen' \betaVnf\ context. Every time the \chnf\ of a $\lamV$-active
component is reached, the subsequent $\lamV$-active components in it (if any)
are recursively considered in left-to-right fashion.

\begin{defi}[Value normal order]
  \label{def:value-normal-order}
  The value normal order strategy $\rel{\vnor}$ is defined by the following
  reduction semantics:
  \begin{displaymath}
    \ctx{A}[\ctx{CH}[(\lambda x.B)N]\rel{\vnor}
    \ctx{A}[\ctx{CH}[\cas{N}{x}{B}]
    \quad\quad\textup{with}\ N\in\Val
  \end{displaymath}
  where $\ctx{CH}\hole$ is a chest reduction context and $\ctx{CH}[(\lambda
  x.B)N]$ is the leftmost $\lamV$-active component of
  $\ctx{A}[\ctx{CH}[(\lambda x.B)N]]$, \ie\ either $\ctx{A}\hole\equiv\hole$
  and $\ctx{CH}[(\lambda x.B)N]$ is not in \chnf, or
  $\ctx{A}\hole\not\equiv\hole$ and $\ctx{A}[\ctx{CH}[(\lambda x.B)N]$ is a
  \chnf\ such that every subterm at the left of $\ctx{CH}[(\lambda x.B)N]$ is
  in \betaVnf.
\end{defi}
We now adapt to pure $\lamV$ the notion of standard reduction sequence in
\cite[p.137]{Plo75}.
\begin{defi}[Standard reduction sequence in $\lamV$]
  \label{def:v-standard}
  A standard reduction sequence (abbrev. SRS) is a sequence of terms defined
  inductively as follows:
  \begin{enumerate}
  \item \label{it:variable} Any variable $x$ is a SRS.
  \item \label{it:prepend-value} If $N_2,\ldots,N_k$ is a SRS and $N_1
    \rel{\vv} N_2$, then $N_1, \ldots,N_k$ is a SRS.
  \item \label{it:lambda} If $N_1,\ldots,N_k$ is a SRS then $\lambda x.N_1,
    \ldots,\lambda x.N_k$ is a SRS.
  \item \label{it:applicative} If $M_1,\ldots,M_j$ and $N_1,\ldots,N_k$ are
    SRS then $M_1\,N_1,\ldots,M_j\,N_1,\ldots,M_j\,N_k$ is a SRS.
  \end{enumerate}
\end{defi}

\begin{thm}
  \label{thm:vnor-standard}
  Value normal order entails a SRS.
\end{thm}
\begin{proof}
  The reduction contexts of value normal order
  (Def.~\ref{def:value-normal-order}) are of the shape
  $\ctx{A}[\ctx{CH}\hole]$ where $\ctx{CH}\hole$ is a chest-reduction context
  (Def.~\ref{def:chest-reduction}) and, if $R$ is the next $\betaV$-redex to
  be contracted, then $\ctx{CH}[R]$ is the leftmost $\lamV$-active component
  of $\ctx{A}[\ctx{CH}[R]]$. The reduction contexts for value normal order can
  be broken down further into $\ctx{A}[\lambda x_1\ldots x_n.\ctx{BV}\hole]$,
  where $n\geq 0$ and $\ctx{BV}\hole$ is a call-by-value reduction context
  (Def.~\ref{def:call-by-value}). Def.~\ref{def:v-standard}(\ref{it:prepend-value})
  says that any reduction sequence entailed by the reduction contexts
  $\ctx{BV}\hole$ of $\rel{\vv}$ is standard.
  Def.~\ref{def:v-standard}(\ref{it:lambda}) says that these reduction
  sequences can be lifted to any number of surrounding lambdas, and so it
  ensures that chest-reduction contexts $\lambda x_1\ldots x_n.\ctx{BV}\hole$
  are standard. The step of locating the leftmost $\lamV$-active component of
  $\ctx{A}\hole$ is standard by Def.~\ref{def:v-standard}(\ref{it:lambda}) and
  Def.~\ref{def:v-standard}(\ref{it:applicative}).
\end{proof}
\subsection{Labelling for counting operands}
\label{sec:labelling}
The second ingredient for the proof of the Partial Genericity Lemma is a
tracking mechanism that counts the number of operands that have been passed to
a particular term. Following \cite{Klo80,BKV00} we define this tracking by
introducing a \emph{lambda calculus labelling} \cite[Def.~8.4.26]{Ter03} that
specifies a generalised notion of descendant. Def.~\ref{def:counting} defines
the labelling $\mathfrak{C}$ for counting. The labels range over
$\{\varepsilon\}\cup\mathbb{N}$, \ie\ either an empty count $\varepsilon$ or a
count $c\geq 0$. When non-empty, the count of the operator in a redex is
increased, assigned to the body of the redex, and then the redex is contracted
(\ie\ the operand is substituted by the free occurrences of the formal
parameter in the body of the redex).
\begin{defi}[Counting labelling]
  \label{def:counting}
  Let the labels $\mathbb{L}=\{\varepsilon\}\cup\mathbb{N}$ be the union of
  the the empty count and the natural numbers. The counting labelling
  $\mathfrak{C}$ and the bisimulation $\mathbb{C}$ are defined by mutual
  induction as follow:
  \begin{itemize}
  \item The labelled terms $\mathfrak{C}(\Lambda)$ are labelled variables
    $x^\ell$ (with $\ell \in \mathbb{L}$), labelled abstractions $(\lambda
    x.B)^\ell$ (with $B$ a labelled term), and labelled applications
    $(M\,N)^\ell$ (with $M$ and $N$ labelled terms). The following statements
    about bisimulation $\mathbb{C}$ hold:
    \begin{itemize}
    \item $x\mathbb{C}x^\ell$.
    \item If $B\mathbb{C}B'$ then $(\lambda x.B)\mathbb{C}(\lambda
      x.B')^\ell$.
    \item If $M\mathbb{C}M'$ and $N\mathbb{C}N'$ then
      $(M\,N)\mathbb{C}(M'\,N')^\ell$.
    \end{itemize}
  \item Suppose $B\mathbb{C}B'$ and $N\mathbb{C}N'$ with the $\beta$-rule of
    the form $(\lambda x.B)N \rel{\beta} \cas{N}{x}{B}$. Let $B'\equiv
    C^{\ell_1}$. Consider the $\betaC$-rule
    \begin{displaymath}
      ((\lambda x.C^{\ell_1})^{\ell_2}N')^{\ell_3} \rel{\betaC}
      \left\{
          \begin{array}{ll}
            \cas{N'}{x}{(C^c)}&\text{if}\ \ell_1=c\\
            \cas{N'}{x}{(C^{c+1})}&\text{if}\ \ell_2=c\\
            \cas{N'}{x}{(C^c)}&\text{if}\ \ell_3=c\\
            \cas{N'}{x}{(C^\varepsilon)}
            &\text{if}\ \ell_1,\ell_2,\ell_3=\varepsilon \\
          \end{array}
        \right.
    \end{displaymath}
    where the capture avoiding substitution function for labelled terms
    (defined below) preserves the label of the subject of the substitution:
    \begin{displaymath}
      \begin{array}{rcl}
        \cas{T^{\ell_1}}{x}{(x^{\ell_2})}&=&T^{\ell_1}\\
        \cas{T^{\ell_1}}{x}{((\lambda x.B^{\ell_2})^{\ell_3})}&=&
        (\lambda x.\cas{T^{\ell_1}}{x}{(B^{\ell_2})})^{\ell_3}\\
        \cas{T^{\ell_1}}{x}{((M^{\ell_2}\,N^{\ell_3})^{\ell_4})}&=&
        ((\cas{T^{\ell_1}}{x}{(M^{\ell_2})})(\cas{T^{\ell_1}}{x}{(N^{\ell_3})}))^{\ell_4}
      \end{array}
    \end{displaymath}
    If $\ell_2=c$, rule $\betaC$ increments the count of the abstraction and
    assigns it to the body $C$ before performing the substitution. (Below we
    show that if some of the $\ell_1$, $\ell_2$, and $\ell_3$ are non-empty,
    the first three alternatives of the $\betaC$-rule coincide.) We set
    $\beta\mathbb{C}\betaC$.
  \end{itemize}
\end{defi}\smallskip

\noindent The definition of labelled terms is extended to contexts
$\mathfrak{C}(\ctx{C}\hole)$ in the trivial way, observing that the hole
$\hole$ in a labelled context does not carry any label. When no confusion
arises, we will omit the epithet `labelled' for terms and contexts.

Initially, all subterms have empty count $\varepsilon$ except for a particular
subterm.
\begin{defi}
  Function $\mathfrak{c}$ takes a term $M$ and delivers $M'$ such that
  $M\mathbb{C}M'$ and where all the subterms of $M'$ have empty count
  $\varepsilon$. For example,
  \begin{displaymath}
    \mathfrak{c}((\lambda x.\lambda y.x)z(\lambda x.x))=
    (((\lambda x.(\lambda y.x^\varepsilon)^\varepsilon)^\varepsilon
    z^\varepsilon)^\varepsilon(\lambda x.x^\varepsilon)^\varepsilon)^\varepsilon
  \end{displaymath}
  The labelling function $\mathfrak{c}$ is extended to contexts in the trivial
  way.
\end{defi}
Typically, we would assign the non-empty count `0' to the unsolvable subterm
that we wish to trace.
\begin{defi}
  Function $\mathfrak{s}$ selects a subterm $M$ in $\ctx{C}[M]$, assigning
  count $0$ to $M$ and empty count everywhere else in $\ctx{C}[M]$, including
  the proper subterms of $M$.
  \begin{displaymath}
    \mathfrak{s}(\ctx{C}\hole,M)=\ctx{C}'[M'^0] \quad\quad\text{where}\
    \mathfrak{c}(\ctx{C}\hole)\equiv\ctx{C}'\hole\
    \text{and}\ \mathfrak{c}(M)\equiv M'^\varepsilon
  \end{displaymath}
  Notice that $M\mathbb{C}(\mathfrak{s}(\ctx{C}\hole,M))$. When no confusion
  arises, we write $\mathfrak{s}(\ctx{C}[M])$ instead of
  $\mathfrak{s}(\ctx{C}\hole,M)$.
\end{defi}

Labelling $\mathfrak{C}$ serves two different purposes. It tracks some
unsolvable with non-empty count, and it counts the operands that have been
passed to it. Consider the $\betaC$-reduction step $((\lambda
x.B^{\varepsilon})^cN')^{\varepsilon} \rel{\betaC} \cas{N'}{x}{(B^{c+1})}$
% \begin{displaymath}
%   ((\lambda x.B^{\varepsilon})^cN')^{\varepsilon} \rel{\betaC} \cas{N'}{x}{(B^{c+1})}
% \end{displaymath}
with $B\not\equiv x$ and $c$ a non-empty count. We are interested in counting
the number of operands passed to operator $\lambda x.B$, and thus the second
line of $\betaC$ assigns the non-empty count $c+1$ to the body $B$ in the
substitution instance $\cas{N'}{x}{(B^{c+1})}$.

Notice that the tracking implemented by $\mathfrak{C}$ differs from the
conventional notion of descendant \cite{Klo80,Bar84,KOV99}. In the example
above, the term $\cas{N'}{x}{(B^{c+1})}$ would be a trace of $(\lambda
x.B^\varepsilon)^c$ if $B\not\equiv x$. And similarly, if the label to be
preserved was that of the application, as the \mbox{$\betaC$-reduction} step
$((\lambda x.B^\varepsilon)^\varepsilon N')^c \rel{\betaC} \cas{N'}{x}{(B^c)}$
illustrates,
% \begin{displaymath}
%   ((\lambda x.B^\varepsilon)^\varepsilon N')^c \rel{\betaC}
%   \cas{N'}{x}{(B^c)}
% \end{displaymath}
then the term $\cas{N'}{x}{(B^c)}$ would be a trace of $((\lambda
x.B^\varepsilon)^\varepsilon N')^c$ if $B\not\equiv x$. But the \emph{traces}
$\cas{N'}{x}{(B^{c+1})}$ and $\cas{N'}{x}{(B^c)}$ could never be
\emph{descendants} in the conventional sense of $(\lambda x.B^\varepsilon)^c$
and $((\lambda x.B^\varepsilon)^\varepsilon N')^c$ (respectively), because
according to the labelled $\beta$-reduction of \cite[p.19]{Klo80} the labels
of the operator and of the redex (\ie\ the $c$ in each of the examples above)
vanish. We distinguish our more refined tracking from the conventional notion
of descendant by using `trace' and `origin' for the former and `descendant'
and `ancestor' for the latter. Notice that all the descendants of $M$ in
$\mathfrak{s}(\ctx{C}[M])$ are traces (\ie\ have non-empty count), but not all
the traces of $M$ in $\mathfrak{s}(\ctx{C}[M])$ are descendants.

The counting labelling $\mathfrak{C}$ can be applied to $\lamV$ by restricting
rule $\betaC$ above with $N'\in\mathfrak{C}(\Val)$. We call the restricted
rule $\betaCV$ and set $\betaV\mathbb{C}\betaCV$. The definition of
$\lamV$-solvable, of order of a term, and of value normal order is extended to
labelled terms in the trivial way.

Our counting labelling captures accurately the number of operands that are
passed to the tracked unsolvable. That is, when tracking $M$ (an unsolvable of
order $n$) in $\mathfrak{s}(\ctx{C}[M])$, all the traces $M_t^c$ in the
$\betaC$-reduction sequence are unsolvables of order $n-c$. In order to prove
this invariant we first need to show that unsolvability and `order of an
unsolvable' are preserved by substitution. This result holds respectively for
$\lamK$ and $\lamV$, by taking the definitions of solvability and of `order of
a term' in the corresponding calculus: for $\lamK$, the usual definition of
solvability and the `order of a term' in \cite{Lon83}; for $\lamV$,
Def.~\ref{def:solvV} and order of a term in Section~\ref{sec:lamV-solv}. We
present the result for $\lamV$ first, since this one is the novel result. The
result for $\lamK$ follows straightforwardly by adapting the proof of the
former.

% (In $\lamK$, notice that by \textsc{Axiom}~3 of \cite{KOV99} we know that
% $\cas{N}{x}{M}$ is $\lamK$-unsolvable.)

% We consider the notion of order in \cite{Lon83}, where the orders comprise
% finite ordinal numbers and the first limit ordinal, (\ie\ the natural numbers
% and $\omega$, see Section~\ref{sec:lamV-solv}).

\begin{lem}[Order of a $\lamV$-unsolvable is preserved by substitution]
  \label{lem:subst-pres-order-lamV}
  Let $M\in\Lambda$ be a \mbox{$\lamV$-unsolvable} of order $n$. For every
  $N\in\Val$, the substitution instance $\cas{N}{x}{M}$ is a
  \mbox{$\lamV$-unsolvable} of order $n$.
\end{lem}
\begin{proof}
  We distinguish two cases:
  \begin{enumerate}
  \item $M$ is of order $\omega$. Then $M\beqV\lambda y_1\ldots y_k.B$ with
    $k$ arbitrarily large. If $x=y_i$ for some $i\leq k$, then by
    substitutivity and by definition of the substitution function
    $\cas{N}{x}{M}\beqV\lambda y_1\ldots y_k.B$. If $x\not=y_i$ with $i\leq
    k$, then by substitutivity of $\beqV$ and by definition of the
    substitution function, then $\cas{N}{x}{M}\beqV\lambda y_1\ldots
    y_k.\cas{N}{x}{B}$ and we are done.
  \item $M$ is of order $n<\omega$. Then $M\beqV\lambda y_1\ldots y_n.B$ and
    by substitutivity of $\beqV$. If $x=y_i$ for some $i\leq n$ then
    $\cas{N}{x}{M}\equiv\lambda y_1\ldots y_n.B$ and the lemma holds. If
    $x\not=y_i$ for every $i\leq n$, since $B$ is $\lamV$-unsolvable of order
    $0$ and by the definitions of \mbox{$\lamV$-solvability} and of the
    substitution function, it suffices to show that $\cas{N}{x}{B}$ is of
    order $0$. We proceed by \emph{reductio ad absurdum}. Assume that
    $\cas{N}{x}{B}$ is of order $m>0$. Then $\cas{N}{x}{B}\beqV\lambda
    z_1\ldots z_m.C$. If $x=z_j$ for some $j\leq m$, then by substitutivity
    and by definition of the substitution function $M\beqV\lambda x_1\ldots
    x_nz_1\ldots z_m.C$, which contradicts the assumptions. If $x\not=z_j$ for
    every $j\leq m$, then $C\equiv \cas{N}{x}{B'}$ for some $B'$, and by
    substitutivity and by definition of the substitution function then
    $M\beqV\lambda x_1\ldots x_nz_1\ldots z_m.B'$, which also contradicts the
    assumptions and we are done.\qedhere
  \end{enumerate}
\end{proof}

\begin{lem}[Order of a $\lamK$-unsolvable is preserved by substitution]
  \label{lem:subst-pres-order-lamK}
  Let $M\in\Lambda$ be a \mbox{$\lamK$-unsolvable} of order $n$. For every
  $N\in\Lambda$, the substitution instance $\cas{N}{x}{M}$ is a
  \mbox{$\lamK$-unsolvable} of order $n$.
\end{lem}
\begin{proof}
  By adapting the proof of Lemma~\ref{lem:subst-pres-order-lamV} to $\lamK$ in
  a straightforward way.
\end{proof}

The invariant stated in Lemma~\ref{lem:lamV-labelling} and
\ref{lem:lamK-labelling} below ensure that, even if several of the $\ell_1$,
$\ell_2$, and $\ell_3$ in the $\betaC$-rule above are non-empty, all the
alternatives coincide and thus $\betaC$-reduction is confluent.

% (In the statement of the lemmata, recall that we mind the left-cancellative
% addition for ordinal numbers.)

\begin{lem}\label{lem:lamV-labelling}
  Let $M_0\in\Lambda$ $\lamV$-unsolvable of order $n_0$. Every trace of $M_0$
  with non-empty count $c$ in any $\betaCV$-reduction sequence starting at
  $\mathfrak{s}(\ctx{C}[M_0])$ is $\lamV$-unsolvable of order $n$ such that
  $n_0=c+n$.\footnote{We assume the standard conventions on ordinal number
    arithmetic \cite{Sie65}. The successor of an ordinal $\alpha$ is
    $\alpha+1$. Addition is non-commutative and left-cancellative, that is,
    let $n$ be a finite ordinal, then $n+\omega=0+\omega=\omega$. Only left
    subtraction is definable, \ie\ $\alpha-\beta=\gamma$ iff $\beta\leq\alpha$
    and $\gamma$ is the unique ordinal such that $\alpha=\beta+\gamma$.}
\end{lem}
\begin{proof}
  By definition, only the traces of $M_0$ (we refer to them as $M_t$) have
  non-empty count $c$. We prove that the contractum of a $\betaCV$-redex
  preserves the invariant $n_0=c+n$ (recall that we mind the left-cancellative
  addition for ordinals) for each labelled trace $M_t$ with non-empty count
  $c$ and order $n$. We consider any $\betaCV$-reduction sequence and proceed
  by induction on the sequence order of the term in which the $\betaCV$-redex
  occurs. (The general case coincides with the base case, except for the small
  differences pinpointed in Cases~2, 3, and 4 below.) Consider the
  $\betaCV$-redex $R\equiv(\lambda x.B)N$ with $N\in\mathfrak{C}(\Val)$
  occurring at step $s$ that is contracted in order to produce the reduct at
  step $s+1$. We focus on each occurrence (if any) of $M_t$ with non-empty
  count $c$ in $R$ and distinguish the following cases:
  \begin{enumerate}
  \item $R\equiv (\lambda x.\ctx{C}[M_t])N$. The contractum is
    $C\equiv\ctx{C}'[\cas{N}{x}{M_t}]$ where
    $\ctx{C}'\hole\equiv\cas{N}{x}{(\ctx{C}\hole)}$. By
    Lemma~\ref{lem:subst-pres-order-lamV}, if $\ctx{C}\hole\equiv\hole$ then
    the occurrence of $\cas{N}{x}{M_t}$ in the contractum is
    $\lamV$-unsolvable of order $n$ and the lemma holds. (Notice that the
    first line of rule $\betaC$ of Def.~\ref{def:counting} takes care of
    preserving the count $c$ of the redex's body $M_t$ if
    $\ctx{C}\hole\equiv\hole$.) Otherwise the order and count of the
    occurrences of $M_t$ in $\cas{N}{x}{(\ctx{C}[M_t])}$ are trivially
    preserved and the lemma follows.
  \item $R\equiv M_t\,N$. Then $M_t\equiv(\lambda x.B)^c$ with $B$
    $\lamV$-unsolvable of order $n-1$. By
    Lemma~\ref{lem:subst-pres-order-lamV} $\cas{N}{x}{B^{c+1}}$ is
    $\lamV$-unsolvable of order $n'=n-1$ and the lemma holds. Notice that
    left-subtraction allows for the limit case when both $n$ and $n'$ are
    infinite ordinals (\ie\ $\omega=n=1+n'=1+\omega=\omega$). This is enough
    for the base case (\ie\ $s=1$), but for the general case there can be an
    overlap with Case~1 if some trace $M'_t$ of $M_0$ occurs in $B^c$. The
    lemma follows as in Case~1 except if $\ctx{C}\hole\equiv\hole$, because
    the first and the second lines of rule $\betaC$ of Def.~\ref{def:counting}
    produce a critical pair. But we show that both alternatives coincide. Let
    $M'_t$ with non-empty count $c'$ be $\lamV$-unsolvable of order $n'$. By
    the induction hypothesis, the invariant holds for $M'_t$ (\ie\
    $n_0=c'+n'$). In the limit case (\ie\ $n_0=\omega$) both $M_t\equiv\lambda
    x.M'_t$ and $M'_t$ have infinite order (\ie\ $n=n'=\omega$) and then
    $n_0=c'+n'=c+n$ and the lemma follows. In the finite case, $n-n'=1$ and
    then $n_0=c'+n'=c+n'+1$ and $c'=c+1$. Therefore both alternatives for rule
    $\betaCV$ coincide and the lemma follows.
  \item $R\equiv M_t$. Then $M_t\equiv((\lambda x.B)N)^c$ with
    $N\in\mathfrak{C}(\Val)$. For the base case the lemma follows because the
    third line of $\betaC$ of Def.~\ref{def:counting} preserves the count of
    the $\betaCV$-redex. For the general case there can be overlap with
    Cases~1 and 2, and the lemma follows because the different alternatives
    for $\betaC$ coincide by the induction hypothesis, as was illustrated in
    Case~2 above.
  \item $R\equiv(\lambda x.B)(\ctx{C}[M_t])$. For each occurrence of $x$ in
    $B$, the order and count of $M_t'$ is trivially preserved by the
    definition of the substitution function (Def.~\ref{def:counting}) and the
    lemma holds. This is enough for the base case.  For the general case there
    can be an overlap with Cases~1, 2, and 3, and the lemma follows because
    the different alternatives for $\betaC$ coincide by the induction
    hypothesis, as was illustrated in Case~2 above.\qedhere
  \end{enumerate}
\end{proof}

\begin{lem}\label{lem:lamK-labelling}
  Let $M_0\in\Lambda$ $\lamK$-unsolvable of order $n_0$. Every trace of $M_0$
  with non-empty count $c$ in any $\betaC$-reduction sequence starting at
  $\mathfrak{s}(\ctx{C}[M_0])$ is $\lamK$-unsolvable of order $n$ such that
  $n_0=c+n$.
\end{lem}
\begin{proof}
  By adapting the proof of Lemma~\ref{lem:lamV-labelling} to $\lamK$ in a
  straightforward way.
\end{proof}

\subsection{Generalised statement and illustration of the proof}
\label{sec:pgl-stated}
We generalise the statement of the Partial Genericity Lemma we gave in
Lemma~\ref{lem:partial-genericity} to provide a proof by induction on the
length of the reduction sequence of value normal order.

% The full proof of the generalised statement (Thm.~\ref{thm:generalised-thm}
% below) is rather low-level and we have put it in App.~\ref{sec:pgl-stated}. We
% prefer to illustrate the proof through an example.

% We assume that $M$ is $\lamV$-unsolvable and that $\ctx{C}[M]
% \beqV N$ with $\N\in\V\NF$ and prove that $\ctx{C}[X]\beqV N$ for all terms
% $X\in\Lambda$ of order $m \geq n$.  We show that all the traces of $M$ are
% discarded in the value-normal-order reduction sequence of $\ctx{C}[M]$ before
% their count reaches $n$ (\ie\ $M$ is applied to at most $n-1$ operands). If
% some trace was never discarded, or if it reached count $c\geq n$, then
% $\ctx{C}[\lambda x_1\ldots x_n.\OMEGA]\not\beqV N$, which is incompatible with
% the assumptions.

First, we take Lemma~\ref{lem:partial-genericity} and pull out the universal
quantifier `for all contexts $\ctx{C}\hole$' from the consequent of the
implication. We take value normal order (Def.~\ref{def:value-normal-order})
and the counting labelling $\mathfrak{C}(\lamV)$ (Def.~\ref{def:counting}).
We take $M$, $N$, and $\ctx{C}\hole$ in Lemma~\ref{lem:partial-genericity} and
subscript them with a 0 to indicate that $M_0\in\mathfrak{C}(\Lambda)$ is the
initial labelled $\lamV$-unsolvable, $N_0\in\mathfrak{C}(\VNF)$ is the
labelled normal form, and $\ctx{C}_0\hole$ is the initial labelled context
such that $\mathfrak{s}(\ctx{C}[M])=\ctx{C}_0[M_0]$, $\ctx{C}_0[M_0] \beqCV
N_0$ and $N\mathbb{C}N_0$. (We also rename $n$ to $n_0$ for uniformity.) The
generalised theorem reads as follows.
\begin{thm}
  \label{thm:generalised-thm}
  Let $M'\in\mathfrak{C}(\Lambda)$ of order $n'\leq n_0$ and $\ctx{C}'\hole$ a
  labelled context. That $\ctx{C}'[M']$ is a labelled reduct in the
  value-normal-order reduction sequence of $\ctx{C}_0[M_0]$ and $M'$ has
  non-empty count implies that if $\ctx{C}'[M']\beqCV N_0$ then for all terms
  $X$ of order $m\geq n_0$ it is the case that
  $\ctx{C}_0[\mathfrak{c}(X)]\beqCV N_0$.
\end{thm}
This theorem coincides modulo $\mathbb{C}$ bisimilarity with
Lemma~\ref{lem:partial-genericity} by taking
$\ctx{C}'\hole\equiv\ctx{C}_0\hole$, $M'\equiv M_0$, and $n' = n_0$. In that
case $\ctx{C}'[M']\equiv\ctx{C}_0[M_0]$ is the first term in the reduction
sequence and $M'$ has non-empty count $0$ in $\ctx{C}'[M']$.
\begin{proof}[Proof of Thm.~\ref{thm:generalised-thm}]
  For brevity, we drop the $\mathfrak{C}$ and $\mathfrak{c}$ from the sets of
  terms and from the reduction rule respectively. Recall from
  Def.~\ref{def:value-normal-order} that the terms in a value-normal-order
  reduction sequence have the shape $\ctx{A}[\ctx{CH}[R]]$, where $R$ is the
  next $\betaV$-redex to be contracted, $\ctx{CH}\hole$ is a chest-reduction
  context, and $\ctx{A}\hole$ is the context in which the leftmost
  \mbox{$\lamV$-active} component $A\equiv\ctx{CH}[R]$ occurs. We focus on the
  traces of $M_0$ that pop up in the value-normal-order reduction sequence of
  $\ctx{C}'[M']$. We proceed by induction on the length of the
  value-normal-order reduction sequence of $\ctx{C}'[M']$.

  The base case is when all the traces of $M_0$ (\ie\ all the subterms with
  non-empty count) that occur in $\ctx{C}'[M']$ (the $M'$ itself is one of
  these traces since it has non-empty count) are discarded in the next
  value-normal-order reduction step. That is, the hole in $\ctx{C}'\hole$, and
  every other trace of $M_0$, lie inside the operand of the next
  \mbox{$\betaV$-redex} $R$, \ie\ $\ctx{C}'[M']\equiv\ctx{A}[\ctx{CH}[R]]$,
  and $R\equiv(\lambda x.B)(\ctx{C}_1[M'])$ such that $x$ does not occur free
  in $B$ and $\ctx{C}_1[M']$ is a value that contains all the traces of
  $M_0$. The next reduct is $\ctx{A}[\ctx{CH}[B]]$. There is no
  $\ctx{C}''\hole$ such that $\ctx{C}''[M']\mrel{\betaV}\ctx{A}[\ctx{CH}[B]]$
  and such that the value-normal-order reduction sequence of $\ctx{C}''[M']$
  is of length less than the value-normal-order reduction sequence of
  $\ctx{C}'[M']$, which explains why this is the base case.

  Since $M'$ is a trace of $M_0$ and $M_0$ has count $0$ in $\ctx{C}_0[M_0]$,
  if the count of $M'$ is greater than $0$ then by Def.~\ref{def:counting}
  this can only be the result of a reduction step $\ctx{A}[\ctx{CH}[(\lambda
  x.B^c)N]]\rel{\vnor}\ctx{A}[\ctx{CH}[\cas{N}{x}{(B^{c+1})}]$ with $B^c$ a
  trace of $M_0$ with non-empty count $c$. Had the count of the contractum
  $\cas{N}{x}{(B^{c+1})}$ reached $n_0$, then by
  Lemmata~\ref{lem:subst-pres-order-lamV} and \ref{lem:lamV-labelling} the
  contractum would be $\lamV$-unsolvable of order $0$ and the
  $\ctx{A}[\ctx{CH}[\cas{N}{x}{(B^{c+1})}]$ would have diverged under value
  normal order. But this contradicts the assumption $\ctx{C}'[M']\beqV
  N_0$. Therefore $M'$ has count at most $n_0-1$ and order greater than $0$,
  \ie\ $M'\in\Val$. If the $M_0$ was replaced by a term $X$, the trace of $X$
  would have reached at most count $n_0-1$. Therefore, for any $X\in \Lambda$
  of order greater or equal than $n_0$ it is the case that $\ctx{C}_0[X]\beqV
  N_0$, and the theorem follows.

  Now we proceed with the general case. We analyse the cases:
  \begin{enumerate}
  \item No trace of $M_0$ pops up, \ie\ $\ctx{C}'[M']$ is not of the form
    $\ctx{A}[\ctx{CH}[M_t]]$ where $M_t$ is a $\betaV$-redex with non-empty
    count $c$. Let $R$ be the next $\betaV$-redex to be contracted, \ie\
    $\ctx{A}[\ctx{CH}[R]]$. Either $R\equiv(\lambda x.B)(\ctx{C}_1[M'])$ with
    $x$ not free in $B$ and this case matches the conditions of the base case
    and we are done, or contracting $R$ does not discard all the traces of
    $M_0$ that occur in $\ctx{C}'[M']$. Let $R'$ be the contractum of $R$ and
    $M_t$ be one of the traces of $M_0$ that occurs in
    $\ctx{A}[\ctx{CH}[R']]$, \ie\ $\ctx{A}[\ctx{CH}[R']]\equiv\ctx{C}_2[M_t]$
    and $M_t$ has non-empty count $c$ (it is immaterial for the proof which of
    the the existing traces of $M_0$ you pick).  By an argument similar to the
    one in the base case, the count of $M_t$ is at most $n_0-1$. The theorem
    holds for $\ctx{C}_2\hole$ and $M_t$ by the induction hypothesis.
  \item A trace of $M_0$ pops up, \ie\ $\ctx{C}'[M']$ is of the form
    $\ctx{A}[\ctx{CH}[M_t]]$ where $M_t$ is a $\betaV$-redex with non-empty
    count $c$. $M_t$ is the next redex to be contracted, and thus it is of the
    shape $((\lambda x.B)N)^c$ with $N\in\Val$. By Def.~\ref{def:counting} the
    contractum of $M_t$ is $\cas{N}{x}{(B^c)}$, which has count $c$ by
    Lemma~\ref{lem:lamV-labelling}. The theorem holds for
    $\ctx{A}[\ctx{CH}\hole]$ and $\cas{N}{x}{(B^c)}$ by the induction
    hypothesis.\qedhere
  \end{enumerate}
\end{proof}\medskip

\noindent The following example illustrates the proof. (Remember we are dropping the
$\mathfrak{C}$ and $\mathfrak{c}$ from the sets of terms and from the
reduction rule respectively.) Consider the context
$\ctx{C}_0\hole\equiv(\lambda x.(\lambda y.\I)(x\,x))\hole$ and the
$\lamV$-unsolvable $M_0\equiv(\I(\lambda x.\lambda y.x\,\OMEGA))^0$ of order 2
and with non-empty count $0$. The conversion $\ctx{C}_0[M_0]\beqV \I$ holds,
where $\I\in\VNF$. The proof proceeds by induction on the length of the
value-normal-order reduction sequence of $\ctx{C}_0[M_0]$. We analyse this
reduction sequence and check that the $\I$ is reached when replacing $M_0$ by
a generic term $X$ of order $m\geq 2$. The first $\betaV$-redex to be
contracted is $\ctx{C}_0[M_0]$. Not all traces of $M_0$ are discarded in the
next reduct and we are at the sub-case of the general case where no trace of
$M_0$ pops up in the reduction sequence.
\begin{displaymath}
  \begin{array}{rcl}
    &&(\lambda x.(\lambda y.\I)(x\,x))
    (\I(\lambda x.\lambda y.x\,\OMEGA))^0\\[4pt]
    &\rel{\vnor}&
    \textup{\scriptsize$\left\{
        \begin{array}{l}
          \ctx{A}[\ctx{CH}\hole]\equiv[\hole]{\color{white}^0}\\[2pt]
          R\equiv(\lambda x.(\lambda y.\I)(x\,x))
          (\I(\lambda x.\lambda y.x\,\OMEGA))^0
        \end{array}\right\}
      $}\\[10pt]
    &&(\lambda y.\I)
    ((\I(\lambda x.\lambda y.x\,\OMEGA))^0
    (\I(\lambda x.\lambda y.x\,\OMEGA))^0)
  \end{array}
\end{displaymath}
The remaining reduction sequence has length less than the reduction sequence
of $\ctx{C}_0[M_0]$ and the property holds for $M_1\equiv M_0$ and
$\ctx{C}_1\hole\equiv(\lambda y.\I)(\hole(\I(\lambda x.\lambda
y.x\,\OMEGA))^0)$ by the induction hypothesis. (Alternatively, we could have
picked the rightmost trace of $M_0$ and the property would also hold for
$M_1\equiv M_0$ and $\ctx{C}_1\hole\equiv(\lambda y.\I)((\I(\lambda x.\lambda
y.x\,\OMEGA))^0\hole)$.) The next $\betaV$-redex to be contracted is the
leftmost occurrence of $M_0$ in $\ctx{C}_1[M_1]$. Not all the traces of $M_0$
are discarded in the next reduct and we are at the sub-case of the general
case where a trace of $M_0$ pops up in the reduction sequence.
\begin{displaymath}
  \begin{array}{rcl}
    &&(\lambda y.\I)
    ((\I(\lambda x.\lambda y.x\,\OMEGA))^0
    (\I(\lambda x.\lambda y.x\,\OMEGA))^0)\\[4pt]
    &\rel{\vnor}&
    \textup{\scriptsize$\left\{
        \begin{array}{l}
          \ctx{A}[\ctx{CH}\hole]\equiv
          [(\lambda y.\I)(\hole(\I(\lambda x.\lambda y.x\,\OMEGA))^0)]
          {\color{white}^0}\\[2pt]
          R\equiv(\I(\lambda x.\lambda y.x\,\OMEGA))^0
        \end{array}\right\}
      $}\\[10pt]
    &&(\lambda y.\I)
    ((\lambda x.\lambda y.x\,\OMEGA)^0
    (\I(\lambda x.\lambda y.x\,\OMEGA))^0)
  \end{array}
\end{displaymath}
The trace converts to $(\lambda x.\lambda y.x\,\OMEGA)^0$, which is
$\lamV$-unsolvable of order $2$. The remaining reduction sequence has length
less than the reduction sequence of $\ctx{C}_1[M_1]$ and the property holds
for $M_2\equiv(\lambda y.\lambda z.x\,\OMEGA)^0$ and
$\ctx{C}_2\hole\equiv(\lambda y.\I)(\hole(\I(\lambda x.\lambda
y.x\,\OMEGA))^0)$ by the induction hypothesis. The next $\betaV$-redex to be
contracted is the rightmost occurrence of $M_0$ in $\ctx{C}_2[M_2]$. Again, we
are at the sub-case of the general case where a trace of $M_0$ pops up in the
reduction sequence.
\begin{displaymath}
  \begin{array}{rcl}
    &&(\lambda y.\I)
    ((\lambda x.\lambda y.x\,\OMEGA)^0
    (\I(\lambda x.\lambda y.x\,\OMEGA))^0)\\[4pt]
    &\rel{\vnor}&
    \textup{\scriptsize$\left\{
        \begin{array}{l}
          \ctx{A}[\ctx{CH}\hole]\equiv
          [(\lambda y.\I)((\lambda x.\lambda y.x\,\OMEGA)^0\hole)]
          {\color{white}^0}\\[2pt]
          R\equiv(\I(\lambda x.\lambda y.x\,\OMEGA))^0
        \end{array}\right\}
      $}\\[10pt]
    &&(\lambda y.\I)
    ((\lambda x.\lambda y.x\,\OMEGA)^0
    (\lambda x.\lambda y.x\,\OMEGA)^0)
  \end{array}
\end{displaymath}
The trace converts to $(\lambda x.\lambda y.x\,\OMEGA)^0$, which is
$\lamV$-unsolvable of order $2$. The remaining reduction sequence has length
less than the reduction sequence of $\ctx{C}_2[M_2]$. The property holds for
$M_3\equiv(\lambda y.\lambda z.x\,\OMEGA)^0$ and $\ctx{C}_3\hole\equiv(\lambda
y.\I)((\lambda y.\lambda z.x\,\OMEGA)^0\hole)$ by the induction
hypothesis. (Alternatively, we could have picked the leftmost occurrence of
$(\lambda y.\lambda z.x\,\OMEGA)^0$ as the trace of $M_0$ and the property
would also hold for $M_3$ as before and $\ctx{C}_3\hole\equiv(\lambda
y.\I)(\hole(\lambda y.\lambda z.x\,\OMEGA)^0)$.) The next $\betaV$-redex to be
contracted is $(\lambda x.\lambda y.x\,\OMEGA)^0 (\lambda x.\lambda
y.x\,\OMEGA)^0$ (\ie\ it is not a trace of $M_0$ itself, but it has the traces
of $M_0$ both as the operator and as the operand). Not all the traces of $M_0$
are discarded in the next reduct and we are at the sub-case of the general
case where no trace of $M_0$ pops up in the reduction sequence.
\begin{displaymath}
  \begin{array}{rcl}
    &&(\lambda y.\I)
    ((\lambda x.\lambda y.x\,\OMEGA)^0
    (\lambda x.\lambda y.x\,\OMEGA)^0)\\[4pt]
    &\rel{\vnor}&
    \textup{\scriptsize$\left\{
        \begin{array}{l}
          \ctx{A}[\ctx{CH}\hole]\equiv
          [(\lambda y.\I)\hole]
          {\color{white}^0}\\[2pt]
          R\equiv(\lambda x.\lambda y.x\,\OMEGA)^0
          (\lambda x.\lambda y.x\,\OMEGA)^0
        \end{array}\right\}
      $}\\[10pt]
    &&(\lambda y.\I)
    (\lambda y.(\lambda x.\lambda y.x\,\OMEGA)^0\OMEGA)^1
  \end{array}
\end{displaymath}
This step increases the count of the operator to $1$, which is now
\mbox{$\lamV$-unsolvable} of order $1$. The next redex discards all the traces
of $M_0$, neither of which has reached count $2$. We are at the base case.
\begin{displaymath}
  \begin{array}{rcl}
    &&(\lambda y.\I)
    (\lambda y.(\lambda x.\lambda y.x\,\OMEGA)^0\OMEGA)^1\\[4pt]
    &\rel{\vnor}&
    \textup{\scriptsize$\left\{
        \begin{array}{l}
          \ctx{A}[\ctx{CH}\hole]\equiv[\hole]
          {\color{white}^0}\\[2pt]
          R\equiv(\lambda y.\I)
          (\lambda y.(\lambda x.\lambda y.x\,\OMEGA)^0\OMEGA)^1
        \end{array}\right\}
      $}\\[10pt]
    &&\I
  \end{array}
\end{displaymath}
Indeed, the property holds by replacing $M_0$ for any $X$ of order $m\geq
2$. Consider $X\equiv(\lambda x.\lambda y.M)$ with $M\in\Lambda$. The
reduction sequence becomes:
\begin{displaymath}
  % \begin{array}{rcl}
  %   &&(\lambda x.(\lambda y.\I)(x\,x))(\lambda x.\lambda y.M)\\[4pt]
  %   &\rel{\vnor}&
  %   (\lambda y.\I)((\lambda x.\lambda y.M)(\lambda x.\lambda y.M))\\[4pt]
  %   &\rel{\vnor}&
  %   (\lambda y.\I)(\lambda y.\cas{(\lambda x.\lambda y.M)}{x}M))\\[4pt]
  %   &\rel{\vnor}& \I
  % \end{array}
  \begin{array}{l}
    (\lambda x.(\lambda y.\I)(x\,x))(\lambda x.\lambda y.M)\rel{\vnor}
    (\lambda y.\I)((\lambda x.\lambda y.M)(\lambda x.\lambda y.M))\\
    \rel{\vnor}
    (\lambda y.\I)(\lambda y.\cas{(\lambda x.\lambda y.M)}{x}M))\rel{\vnor}\I
  \end{array}
\end{displaymath}

\subsection{Complete strategies of %
  \texorpdfstring{$\lamV$}{lambda-V} that are not standard}
\label{sec:complete-standard}
Standard reduction sequences are not unique \cite[Sec.1.5]{HZ09}. To this we
add that not every complete reduction sequence that only contracts needed
redexes is standard! There are reduction strategies of $\lamV$ which only
contract needed redexes but do not entail standard reduction sequences.  This
fact is the analogous in $\lamV$ to the result in \cite{BKKS87} about spine
strategies of $\lamK$. We shall see an example in
Def.~\ref{def:ribcage-reduction} below.

To illustrate the non-uniqueness of standard reduction sequences, consider the
term $M\equiv (\lambda x.(\lambda y.z\,y)\I)((\lambda y.z\,y)\K)$ that
converts to the stuck $(\lambda x.z\,\I)(z\,\K)$. The reduction sequence is
standard and ends in $M$'s \betaVnf:
\begin{displaymath}
  (\lambda x.(\lambda y.z\,y)\I)((\lambda y.z\,y)\K)
  \rel{\vv}
  (\lambda x.(\lambda y.z\,y)\I)(z\,\K)
  \rel{\betaV}
  (\lambda x.z\,\I)(z\,\K)
\end{displaymath}
The first $\rel{\vv}$ step is a call-by-value step, which is standard by
Def.~\ref{def:v-standard}(\ref{it:prepend-value}). The second $\rel{\betaV}$
step is standard by Def.~\ref{def:v-standard}(\ref{it:applicative}),
Def.~\ref{def:v-standard}(\ref{it:lambda}), and
Def.~\ref{def:v-standard}(\ref{it:prepend-value}).

However, the following alternative reduction sequence is also standard and
also ends in $M$'s \betaVnf:
\begin{displaymath}
  (\lambda x.(\lambda y.z\,y)\I)((\lambda y.z\,y)\K)
  \rel{\betaV}
  (\lambda x.z\,\I)((\lambda y.z\,y)\K)
  \rel{\betaV}
  (\lambda x.z\,\I)(z\,\K)
\end{displaymath}
The first $\rel{\betaV}$ step is standard by
Def.~\ref{def:v-standard}(\ref{it:applicative}),
Def.~\ref{def:v-standard}(\ref{it:lambda}), and
Def.~\ref{def:v-standard}(\ref{it:prepend-value}). The second $\rel{\betaV}$
step is standard by Def.~\ref{def:v-standard}(\ref{it:applicative}) and
Def.~\ref{def:v-standard}(\ref{it:prepend-value}).

Ribcage reduction (Def.\ref{def:ribcage-reduction}) is complete with respect
to \chnf\ and only contracts needed redexes. The definition of value normal
order (Def.~\ref{def:value-normal-order}) can be modified to use ribcage
reduction instead of chest reduction for $\lamV$-active components. The
resulting strategy is full-reducing and complete with respect to \betaVnf, but
it does not entail a standard reduction sequence. For example, consider the
term $N\equiv(\lambda x.(\lambda y.x)z)\I$ which converts to the \betaVnf\
$\I$. Ribcage reduction entails the reduction sequence
\begin{displaymath}
    (\lambda x.(\lambda y.x)z)\I
    \rel{\stgy{rc}}
    (\lambda x.x)\I
    \rel{\stgy{rc}}
    \I
\end{displaymath}
This reduction sequence is not standard, although the steps, in isolation, are
standard. The first is standard by
Def.~\ref{def:v-standard}(\ref{it:applicative}),
Def.~\ref{def:v-standard}(\ref{it:lambda}), and
Def.~\ref{def:v-standard}(\ref{it:prepend-value}). The second is standard by
Def.~\ref{def:v-standard}(\ref{it:prepend-value}). However, none of the rules
of Def.~\ref{def:v-standard} allow us to prepend the first step to the
standard reduction sequence consisting of the second step.

Standard reduction sequences to \betaVnf\ fall short of capturing all complete
strategies of $\lamV$. In \cite[p.208]{BKKS87} they generalise the
Quasi-Leftmost Reduction Theorem \cite[Thm.~3.22]{HS08} and show that
`quasi-needed reduction is normalising'. An analogous result is missing for
$\lamV$ (Section~\ref{sec:conclusion-future-work}).

\subsection{An operational characterisation of %
  \texorpdfstring{$\lamV$}{lambda-V}-solvability?}
\label{sec:operational-characterisation}
Although analogous to head reduction and similar in spirit, chest reduction
does not provide an operational characterisation of $\lamV$-solvability. The
term $\T_1\equiv(\lambda y.\DELTA)(x\,\I)\DELTA(x(\lambda x.\OMEGA))$
introduced in Section~\ref{sec:lamV-solv} and the term $\T_2\equiv(\lambda
y.\DELTA)(x\,\I)\DELTA(\lambda x.\OMEGA)$ are {\chnf}s that are not
\mbox{$\lamV$-solvable}. The diverging subterm $\lambda x.\OMEGA$ cannot be
discarded because $(\lambda y.\DELTA)(x\,\I)\DELTA$ is not transformable.
Although $(\lambda y.\DELTA)(x\,\I)\DELTA$ is trivially freezable into a
\betaVnf, there is no context $\ctx{C}\hole$ that transforms that term to some
term that could discard the trailing $\lambda x.\OMEGA$ and obtain a \betaVnf.

The $\lamV$-solvables are `more reduced' than {\chnf}s. This brings us to the
question of the existence of an operational characterisation of
$\lamV$-solvables, that is, a reduction strategy of $\lamV$ that terminates
iff the input term is $\lamV$-solvable. We believe such strategy exists but
cannot be compositional because it requires non-local information about the
shape of the term to decide which is the next $\betaV$-redex
(Section~\ref{sec:conclusion-future-work}).

%%%%%%%%%%%%%%%%%%%%%%%%%%%%%%%%%%%%%%%%%%%%%%%%%%%%%%%%%%%%%%%%%%%%%%%%%%%%%%
\section{The consistent %
  \texorpdfstring{$\lamV$}{lambda-V}-theory %
  \texorpdfstring{$\HV$}{V}}
\label{sec:consistent-lamV-theory}
We adapt \cite[Def.~4.1.1]{Bar84} and say that a $\lamV$\emph{-theory} is a
consistent extension of a conversion proof-theory of $\lamV$. In this section
we prove the consistency of the $\lamV$-theory $\HV$ introduced in
Section~\ref{sec:lamV-solv}. The proof proceeds in similar fashion to the
proof of consistency of the $\lamK$-theory $\mathcal{H}$ introduced in
Section~\ref{sec:lamK-solv}.  The latter proof is detailed in
\cite[Sec.~16.1]{Bar84} and employs some technical machinery introduced in
\cite[Sec.~15.2]{Bar84}. We prove the consistency of $\HV$ in similar fashion,
save for the use of a shorter proof technique in a particular lemma.  We ask
the reader to read this section in parallel with \cite[Sec.~16.1]{Bar84} and
\cite[Sec.~15.2]{Bar84}. The reader also needs to recall the definition of
`notion of reduction' \cite[p.50ff]{Bar84} and `substitutive' binary relation
\cite[p.55ff]{Bar84}. Rule $\betaV$ is a notion of reduction from which
relations $\relV$, $\mrelV$, and $\beqV$ are generated
(Section~\ref{sec:prelim}).

The structure of this section is as follows: We first define $\OmV$-reduction
that sends $\lamV$-unsolvables of order $n$ to a special symbol $\OMEGA_n$.
We then consider the notion of reduction $\betaV\cup\OmV$ which, paraphrasing
\cite[p.388]{Bar84}, is interesting because it analyses provability in
$\lamV$.  We define $\bVOmV$-reduction as the compatible, reflexive, and
transitive closure of $\betaV\cup\OmV$, and prove that it is a
\sv-substitutive relation.  At this point the storyline differs from
\cite{Bar84} in that we introduce the notion of complete $\OmV$-development of
a term, and use the Z property \cite{Oos08} to prove that $\betaV\cup\OmV$ is
Church-Rosser ($\bVOmV$-reduction is confluent). Finally, we define $\HV$ and
the notion of $\omega$-sensibility, and prove that $\HV$ is generated by
$\betaV\cup\OmV$. The consistency of $\HV$ (Thm.~\ref{thm:Hv-lamV-theory})
follows from the confluence of $\bVOmV$-reduction.
\begin{defi}
  \label{def:omv-red}
  The $\OmV$-reduction, $\mrel{\OmV}$, is the compatible, reflexive, and
  transitive closure of the notion of reduction
  \begin{displaymath}
    \OmV = \{(M,\OMEGA_n)\ |\ M\
    \text{$\lamV$-unsolvable\ of order $n$ and}\ M \not\equiv \OMEGA_n\}
  \end{displaymath}
  where $\OMEGA_n$ stands for the term $\lambda x_1\ldots x_n.\OMEGA$ (if
  $n\not=\omega$), or the term $\Y\,\K$ (if $n=\omega$). Notice that $\Y\,\K$
  does not have a \betaVnf\ and that it reduces to $\lambda x_1\ldots
  x_k.\Y\,\K$ with $k$ arbitrarily large.

  The $\OmV$-conversion, $=_{\OmV}$, is the symmetric closure of
  $\mrel{\OmV}$.
\end{defi}

\begin{defi}
  The $\bVOmV$-reduction, $\mrel{\bVOmV}$, is the compatible, reflexive, and
  transitive closure of the notion of reduction $\betaV\cup\OmV$.

  The $\bVOmV$-conversion, $=_\bVOmV$, is the symmetric closure of
  $\mrel{\bVOmV}$.
\end{defi}

\begin{defi}
  \label{lem:v-substitutive}
  Let $M,N\in\Lambda$ and $V\in\Val$. A binary relation $R$ is
  \sv-substitutive iff $R(M,N)$ implies $R(\cas{V}{x}{M},\cas{V}{x}{N})$.
\end{defi}

\begin{lem}
  \label{lem:notion-red-v-substitutive}
  If $R$ is \sv-substitutive, then $\rel{R}$, $\mrel{R}$, and $=_{R}$ are
  \sv-substitutive.
\end{lem}
\begin{proof}
  Straightforward by structural induction on the derivations of $\rel{R}$,
  $\mrel{R}$, and $=_{R}$, respectively (\ie\ by considering the sets
  $\{\mu,\nu,\xi\}$, $\{\mu,\nu,\xi,\rho,\sigma\}$, or
  $\{\mu,\nu,\xi,\rho,\sigma,\tau\}$, respectively, from the rules in
  Section~\ref{sec:prelim}).
\end{proof}

\begin{lem}
  The notion of reduction $\betaV$ is \sv-substitutive.
\end{lem}
\begin{proof}
  Thm.~1 in \cite[p.135]{Plo75} states that $\beqV$ is \mbox{\sv-substitutive}
  in the applied $\lamV$. By an argument similar to the proof of that theorem
  it is straightforward to prove that the $\betaV$-rule is
  \mbox{\sv-substitutive}.
\end{proof}

\begin{lem}
  The relations $\relV$, $\mrelV$, and $\beqV$ are
  \mbox{\sv-substitutive}.
\end{lem}
\begin{proof}
  Trivial by Lemma~\ref{lem:notion-red-v-substitutive} above.
\end{proof}

\begin{lem}
  \label{lem:union-v-substitutive}
  Let $R_1$ and $R_2$ be two notions of reduction that are
  \sv-substitutive. The union $R_1\cup R_2$ is \sv-substitutive.
\end{lem}
\begin{proof}
  Trivial, by considering $R_1$ or $R_2$ individually.
\end{proof}

\begin{lem}[$\bVOmV$ is \sv-substitutive]
  \label{lem:bvomv-substitutive}
  Let $M,N\in\Lambda$ and $V\in\Val$. $M\rel{\bVOmV}N$ implies
  $\cas{V}{x}{M}\rel{\bVOmV}\cas{V}{x}{N})$.
\end{lem}
\begin{proof}
  By Lemma~\ref{lem:union-v-substitutive}, it is enough to prove that $\OmV$
  is \sv-substitutive. Let $M \rel{\OmV} \OMEGA_n$. By
  Lemma~\ref{lem:subst-pres-order-lamV}, the substitution instance
  $\cas{V}{x}{M}$ is \mbox{$\lamV$-unsolvable} of order $n$ for any
  $V\in\Val$. By Def.~\ref{def:omv-red} above, $\cas{V}{x}{M} \rel{\OmV}
  \OMEGA_n$ and $\OMEGA_n\equiv\cas{V}{x}{\OMEGA_n}$ because all the
  $\OMEGA_n$ (including $\Omom$) are closed terms.
\end{proof}
\begin{lem}
  The relations $\mrel{\bVOmV}$, and $=_{\bVOmV}$ are \mbox{\sv-substitutive}.
\end{lem}
\begin{proof}
  Trivial by Lemma~\ref{lem:notion-red-v-substitutive} above.
\end{proof}

\begin{defi}
  Let $M,N\in\Lambda$. $M$ and $N$ are \mbox{$\lamV$-solvably} equivalent,
  $M\sim_{s_\vv} N$, iff for every arbitrary context $\ctx{C}\hole$,
  $\ctx{C}[M]$ is $\lamV$-unsolvable of order $n$ iff $\ctx{C}[N]$ is
  $\lamV$-unsolvable of order $n$.

  Relation $\sim_{s_\vv}$ is reflexive, symmetric, and transitive, and hence
  it is an equivalence relation.
\end{defi}

\begin{lem}
  \label{lem:solv-equiv}
  Let $M,N\in\Lambda$.
  \begin{enumerate}
  \item[\textup{(1)}] $M \beqV N$ implies $M \sim_{s_\vv} N$.
  \item[\textup{(2)}] $M =_\OmV N$ implies $M \sim_{s_\vv} N$.
  \end{enumerate}
\end{lem}
\begin{proof}
  First consider (1). Since $\beqV$ is compatible, for any context
  $\ctx{C}\hole$ then $\ctx{C}[M]\beqV \ctx{C}[N]$, and (1) trivially follows.

  Now consider (2). Since $\sim_{s_\vv}$ is an equivalence relation, it is
  enough to show that $M\sim_{s_\vv}\OMEGA_n$ for $M$ $\lamV$-unsolvable of
  order $n$. Suppose $\ctx{C}[M]$ is $\lamV$-solvable. Then there exists a
  function context $\ctx{F}\hole$ such that $\ctx{F}[\ctx{C}[M]]\beqV N$ for
  some $N \in \VNF$. By the Partial Genericity Lemma
  (Lemma~\ref{lem:partial-genericity}) then $\ctx{F}[\ctx{C}[\OMEGA_n]]\beqV
  N$. Similarly, $\ctx{C}[\OMEGA_n]$ being $\lamV$-solvable implies
  $\ctx{C}[M]$ is $\lamV$-solvable, and (2) follows.
\end{proof}

\begin{rem}
  We write $\OmV(M)$ for the $\OmV$-normal-form (abbrev. $\OmV$-nf) of the
  term $M$.
\end{rem}

\begin{lem}
  \label{lem:omv-normal-form}
  Every term has a unique $\OmV$-nf.
\end{lem}
\begin{proof}
  The maximal $\OmV$-redexes are mutually disjoint. By replacing them by the
  appropriate $\OMEGA_n$s, no new $\OmV$-redexes are created, since $U_n
  \sim_{s_\vv} \OMEGA_n$ for $U_n$ $\lamV$-unsolvable of order $n$. The
  \mbox{$\OmV$-nf} is unique since $\OmV$-reduction is Church-Rosser.
\end{proof}

The complete $\OmV$-development of a term defined below adapts the notion of
complete development of a term \cite[Sec.4.5,p.106]{Ter03} to
$\bVOmV$-reduction.

\begin{defi}
  \label{def:complete-omv-development}
  The complete $\OmV$-development $M\cdv{\OMEGA}$ of a term $M$ consists of
  the complete development of the $\OmV$-nf of $M$, \ie\ $M\cdv{\OMEGA} =
  (\OmV(M))\cdv{}$
\end{defi}

\begin{lem}[Confluence of $\mrel{\OmV}$]
  \label{lem:omv-church-roser}
  The relation $\mrel{\OmV}$ is Church-Rosser.
\end{lem}
\begin{proof}
  It is enough to prove that $(\rel{\OmV}\cup\equiv)$ has the diamond
  property. Consider $M\rel{\OmV}M_1$ by contracting the \mbox{$\OmV$-redex}
  $U_1$, and $M\rel{\OmV}M_2$ by contracting the $\OmV$-redex $U_2$. We
  analyse the cases:
  \begin{enumerate}
  \item $U_1$ and $U_2$ are disjoint. The lemma trivially holds.
  \item $U_1$ and $U_2$ overlap. Let $U_1$, a $\lamV$-unsolvable of order $m$,
    be a superterm of $U_2$, a $\lamV$-unsolvable of order $n$. The diagram
    \begin{center}
    \begin{tikzpicture}[description/.style={fill=white,inner sep=2pt}]
      \matrix (a) [matrix of math nodes, row sep=3em,
      column sep=3em,text height=1.5ex,text depth=0.25ex]
      { M \equiv \ctx{C}_1[U_1] \equiv \ctx{C}_1[\ctx{C}_2[U_2]]
        & M_2 \equiv \ctx{C}_1[\OMEGA_m]\\
        M_1 \equiv \ctx{C}_1[\ctx{C}_2[\OMEGA_n]]
        & \ctx{C}_1[\OMEGA_m] \\};
      \path[->,font=\scriptsize]
      (a-1-1) edge node[above] {$U_1$}               (a-1-2)
              edge node[right]  {$U_2$}               (a-2-1)
      (a-2-1) edge node[above] {$\ctx{C}_2[\OMEGA_n]$} (a-2-2);
      \path[triple]
      (a-1-2) to (a-2-2);
      \path[thirdline]
      (a-1-2) to (a-2-2);
    \end{tikzpicture}
    \end{center}
    commutes because $\ctx{C}_2[\OMEGA_n] \sim_{s_\vv} U_1$ holds by
    Lemma~\ref{lem:solv-equiv} above.\qedhere
  \end{enumerate}
\end{proof}

\begin{lem}[Confluence of $\bVOmV$]
  \label{lem:conf-bo}
  $\bVOmV$-reduction is Church-Rosser.
\end{lem}
\begin{proof}
  It is enough to prove that $\rel{\bVOmV}$ has the Z property \cite{Oos08}:
  \begin{center}
  \begin{tikzpicture}[description/.style={fill=white,inner sep=2pt}]
    \matrix (a) [matrix of math nodes, row sep=4em,
    column sep=4em,text height=1.5ex,text depth=0.25ex]
    { M             & N \\
      M\cdv{\OMEGA} & N\cdv{\OMEGA} \\};
    \path[->,font=\scriptsize]
    (a-1-1) edge node[below,pos=0.7]{$\bVOmV$} (a-1-2);
    \path[->,dashed,font=\scriptsize]
    (a-1-2) edge node[sloped,above,pos=1]{$*$}
                 node[sloped,below,pos=0.7]{$\bVOmV$} (a-2-1);
    \path[->,dashed,font=\scriptsize]
    (a-2-1) edge node[above,pos=1]{$*$}
                 node[below,pos=0.7]{$\bVOmV$} (a-2-2);
  \end{tikzpicture}
  \end{center}
  There are two cases, $M \rel{\OmV} N$ and $M \rel{\betaV} N$:
  \begin{enumerate}
  \item Case $M \rel{\OmV} N$. It follows that \mbox{$\OmV(M) \equiv \OmV(N)$}
    and $M\cdv{\OMEGA} \equiv N\cdv{\OMEGA}$. Therefore \mbox{$N \mrel{\bVOmV}
      M\cdv{\OMEGA}$} and $M\cdv{\OMEGA} \mrel{\bVOmV} N\cdv{\OMEGA}$ and so
    the lemma follows.
  \item Case $M \rel{\betaV} N$. Let $R$ be the $\betaV$-redex contracted in
    $M \rel{\betaV} N$. Let $\mathsf{S}$ be the set of maximal $\OmV$-redexes
    in $M$. If $R$ is disjoint with $\mathsf{S}$ then $M\cdv{\OMEGA} \equiv
    N\cdv{\OMEGA}$ and the lemma follows as in the previous case. If $R$ is
    not disjoint with some $U \in\mathsf{S}$ then we consider the sub-cases:
    \begin{enumerate}
    \item Sub-case $U \equiv \ctx{C}[R]$ is $\lamV$-unsolvable of order
      $n$. Let $R'$ be the contractum of $R$. By Lemma~\ref{lem:solv-equiv}
      above we have $\OmV(\ctx{C}[R]) \equiv \OmV(\ctx{C}[R'])$ and $\OmV(M)
      \equiv \OmV(N)$. Therefore $M\cdv{\OMEGA} \equiv N\cdv{\OMEGA}$.
    \item Sub-case $R \equiv (\lambda x.B)\ctx{C}[U]$ is $\lamV$-solvable
      with $B$ disjoint with $\mathsf{S}$. Let $n$ be the order of $U$. The
      following diagram
      \begin{center}
      \begin{tikzpicture}[description/.style={fill=white,inner sep=2pt}]
        \matrix (a) [matrix of math nodes, row sep=2em,
        column sep=4em,text height=1.5ex,text depth=0.25ex]
        { M             & N \\
          \ctx{C}'[(\lambda x.B)\ctx{C}[U]]
          & \ctx{C}'[\cas{\ctx{C}[U]}{x}{B}\,]\\
          \ctx{C}'[(\lambda x.B)\ctx{C}[\OMEGA_n]]
          & \\
          \ctx{C}'[\cas{\ctx{C}[\OMEGA_n]}{x}{B}]
          & \\
          M\cdv{\OMEGA} & N\cdv{\OMEGA} \\};
        \path[->,font=\scriptsize]
        (a-1-1) edge node[below,pos=0.7]{$\betaV$} (a-1-2);
        \path[triple]
        (a-1-1) to (a-2-1);
        \path[thirdline]
        (a-1-1) to (a-2-1);
        \path[triple]
        (a-1-2) to (a-2-2);
        \path[thirdline]
        (a-1-2) to (a-2-2);
        \path[->,font=\scriptsize]
        (a-2-1) edge node[left,pos=0.7]{$\OmV$} (a-3-1);
        \path[->,font=\scriptsize]
        (a-3-1) edge node[left,pos=0.7]{$\betaV$} (a-4-1);
        \path[->,font=\scriptsize]
        (a-2-2) edge node[sloped,below,pos=0.7]{$\OmV$} (a-4-1);
        \path[->,font=\scriptsize]
        (a-4-1) edge node[right,pos=1]{$*$}
                     node[left,pos=0.7]{$\bVOmV$} (a-5-1);
        \path[->,font=\scriptsize]
        (a-2-2) edge node[right,pos=1]{$*$}
                     node[left,pos=0.94]{$\bVOmV$} (a-5-2);
        \path[triple]
        (a-5-1) to (a-5-2);
        \path[thirdline]
        (a-5-1) to (a-5-2);
      \end{tikzpicture}
      \end{center}
      commutes because $\ctx{C}'[\cas{\ctx{C}[\OMEGA_n]}{x}{B}] \mrel{\bVOmV}
      M\cdv{\OMEGA} \equiv N\cdv{\OMEGA}$, since $B$ and $\mathsf{S}$ are
      disjoint.
    \item Sub-case $R \equiv (\lambda x.\ctx{C}[U])V$ is $\lamV$-solvable with
      $V\in\Val$ not necessarily disjoint with~$\mathsf{S}$.  Let $n$ be the
      order of $U$. The following diagram
      \begin{center}
      \begin{tikzpicture}[description/.style={fill=white,inner sep=2pt}]
        \matrix (a) [matrix of math nodes, row sep=2em,
        column sep=2.5em,text height=1.5ex,text depth=0.25ex]
        { M             & N \\
          \ctx{C}'[(\lambda x.\ctx{C}[U])V]
          &  \ctx{C}'[\cas{V}{x}{(\ctx{C}[U])}]\\
          \ctx{C}'[(\lambda x.\ctx{C}[\OMEGA_n])V]
          & \ctx{C}'[\cas{V}{x}{(\ctx{C}[\OMEGA_n])}]\\
          \ctx{C}'[(\lambda x.\ctx{C}[\OMEGA_n])\OmV(V)]
          & \\
          \ctx{C}'[\cas{\OmV(V)}{x}{(\ctx{C}[\OMEGA_n]})]
          & \\
          M\cdv{\OMEGA} & N\cdv{\OMEGA} \\};
        \path[->,font=\scriptsize]
        (a-1-1) edge node[below,pos=0.7]{$\betaV$} (a-1-2);
        \path[triple]
        (a-1-1) to (a-2-1);
        \path[thirdline]
        (a-1-1) to (a-2-1);
        \path[triple]
        (a-1-2) to (a-2-2);
        \path[thirdline]
        (a-1-2) to (a-2-2);
        \path[->,font=\scriptsize]
        (a-2-1) edge node[left,pos=0.7]{$\OmV$} (a-3-1);
        \path[->,font=\scriptsize]
        (a-2-2) edge node[left,pos=0.7]{$\OmV$} (a-3-2);
        \path[->,font=\scriptsize]
        (a-3-1) edge node [right,pos=1]{$*$}
                     node[left,pos=0.7]{$\OmV$} (a-4-1);
        \path[->,font=\scriptsize]
        (a-3-2) edge node[above,pos=1]{$*$}
                     node[sloped,below,pos=0.7]{$\OmV$} (a-5-1);
        \path[->,font=\scriptsize]
        (a-4-1) edge node[left,pos=0.7]{$\betaV$} (a-5-1);
        \path[->,font=\scriptsize]
        (a-5-1) edge node[right,pos=1]{$*$}
                     node[left,pos=0.7]{$\bVOmV$} (a-6-1);
        \path[->,font=\scriptsize]
        (a-3-2) edge node[right,pos=1]{$*$}
                     node[left,pos=0.94]{$\bVOmV$} (a-6-2);
        \path[triple]
        (a-6-1) to (a-6-2);
        \path[thirdline]
        (a-6-1) to (a-6-2);
      \end{tikzpicture}
      \end{center}
      commutes because
      \begin{displaymath}
        \ctx{C}'[\cas{V}{x}{(\ctx{C}[\OMEGA_n])}] \mrel{\OmV}
        \ctx{C}'[\cas{\OmV(V)}{x}{(\ctx{C}[\OMEGA_n]})]
      \end{displaymath}
      follows because of (i) Prop.~2.1.17(ii) in \cite{Bar84}, (ii)
      $\ctx{C}[\OMEGA_n]$ and $\mathsf{S}\setminus \{U\}$ are disjoint, and
      (iii) $\OmV$ is \sv-substitutive.\qedhere
    \end{enumerate}
  \end{enumerate}
\end{proof}

\begin{defi}
  We say that any theory containing $\HV$ is \mbox{$\omega$-sensible} (and by
  extension, any model satisfying $\HV$ is \mbox{$\omega$-sensible}).
\end{defi}

\begin{defi}[Consistent theory]
  Let $\mathcal{T}$ be a set of equations between terms. $\mathcal{T}$ is
  consistent, $\Con(\mathcal{T})$, iff $\mathcal{T}$ does not prove every
  closed equation, \ie\
  \begin{displaymath}
    \mathcal{T} \not\vdash M=N\ \text{with}\ M,N\in\Lambda^0
  \end{displaymath}
\end{defi}

\begin{defi}[$\lamV$-theory]
  Let $\mathcal{T}$ be a set of closed equations between terms. $\mathcal{T}$
  is a $\lamV$-theory iff $\Con(\mathcal{T})$ and
  \begin{displaymath}
    \mathcal{T}=\{M=N~|~M,N\in\Lambda^0\ \text{and}\ \lamV +
    \mathcal{T} \vdash M=N\}
  \end{displaymath}
\end{defi}

\begin{prop}
  The theory of $\betaV$-convertible closed terms, $\lamV$, is a
  $\lamV$-theory. Observe that $\lamV$ is consistent by confluence of
  \mbox{$\betaV$-reduction}.
\end{prop}

\begin{defi}[Theory $\HV$]
  \label{def:theory-Hv}
  Let $\mathcal{V}_0$ be the following set of equations:
  \begin{displaymath}
    \mathcal{V}_0 = \{M = N~|~M,N\in\Lambda^0\
    \text{$\lamV$-unsolvable of the same order $n$}\}
  \end{displaymath}\smallskip

\noindent The theory $\HV$ is the set of equations:
  \begin{displaymath}
    \HV=\{M=N~|~M,N\in\Lambda^0\ \text{and}\
    \lamV + \mathcal{V}_0 \vdash M=N\}
  \end{displaymath}
\end{defi}

% [TODO: Clarify the role of free variables in $\HV_0$. It seems that if
% $M,N\in\Lambda$ are $\lamV$-unsolvable of order $n$, then, for any closing
% context $\ctx{C}\hole$, $\ctx{C}[M]$ and $\ctx{C}[N]$ are unsolvables of
% whatever (same) order, or they have the same normal form.]

\begin{lem}
  \label{lem:bo-generates-hv}
  $\bVOmV$-reduction generates $\HV$, \ie\
  \begin{displaymath}
    \HV \vdash M = N\ \text{iff}\ M =_\bVOmV N\
    \text{with}\ M,N\in\Lambda^0
  \end{displaymath}
\end{lem}
\begin{proof}
  We first consider the direction ($\Longrightarrow$). If ${\HV_0 \vdash M =
    N}$ then $M \rel{\OmV} \OMEGA_n$ and $M \rel{\OmV} \OMEGA_n$ because both
  $M$ and $N$ are \mbox{$\lamV$-unsolvable} of order $n$. Consequently, for
  all axioms $M_0 = N_0$ in the set $\HV_0$ that generates $\HV$, ${M_0
    =_{\bVOmV} N_0}$ holds, and then $M =_\bVOmV N$ follows by compatibility,
  reflexivity, symmetry and transitivity.

  Now for the direction ($\Longleftarrow$). The theory $\HV$ is generated by
  $\lamV + \HV_0$, and then each $\betaV$- or $\OmV$-reduction step is
  provable in $\HV$.
\end{proof}

\begin{thm}
  \label{thm:Hv-lamV-theory}
  $\HV$ is a $\lamV$-theory.
\end{thm}
\begin{proof}
  By Def.~\ref{def:theory-Hv} and because $\Con(\HV)$ by
  Lemmata~\ref{lem:bo-generates-hv} and \ref{lem:conf-bo}.
\end{proof}

%%%%%%%%%%%%%%%%%%%%%%%%%%%%%%%%%%%%%%%%%%%%%%%%%%%%%%%%%%%%%%%%%%%%%%%%%%%%%%
\section{Related work}
\label{sec:related-work}
We have commented at length from the introduction onwards on the relevant
related work on solvability in $\lamK$ and $\lamV$. We only comment here
briefly on several outstanding points and on other work of related interest.

As mentioned in Section~\ref{sec:value-normal-order}, value normal order is
not the same strategy as the complete reduction strategy of $\lamV$ named
$\rel{\Gamma}^p$ that is obtained as an instantiation of the `principal
reduction machine' of \cite[p.70]{RP04}. The principal reduction machine is
actually a template of small-step reduction strategies that is parametric on a
set of permissible operands and a set of irreducible terms. The complete
reduction strategy $\rel{\Gamma}^p$ is obtained by instantiating the template
with the set of permissible operands fixed to $\Val$ and the set of
irreducible terms fixed to $\VNF$ (in \cite{RP04} $\Val$ is called $\Gamma$
and $\VNF$ is called $\Gamma$-NF). Value normal order differs from
$\rel{\Gamma}^p$ when reducing a term $(\lambda x.B)N$ where $N$ converts to a
neutral. In $\rel{\Gamma}^p$ the operand $N$ is reduced to the neutral $N'$
using call-by-value so that $(\lambda x.B)N'$ is a block. At this point
$\rel{\Gamma}^p$ keeps reducing $N'$ fully to \betaVnf\ before reducing $B$
fully to \betaVnf. In contrast, value normal order proceeds in left-to-right
fashion with the block $(\lambda x.B)N'$, first reducing $B$ fully to
\betaVnf\ and then reducing $N'$ fully to \betaVnf.
% The block $(\lambda x.B)N'$ has to be construed as an instance of a \betaVnf\
% `at the first level' \cite[p.504]{Wad76}. If it had a \betaVnf, it would be of
% the form $(\lambda x.B')N''$ with $B'\beqV B$ and $N''\beqV N'$.
The left-to-right order is the regular one, at least so in all the strategies
cited in this paper. And we have defined value normal order as the $\lamV$
analogue of $\lamK$'s normal order following the results in \cite{BKKS87}.  At
any rate, reducing blocks left-to-right or right-to-left does not affect
completeness. Both $\rel{\Gamma}^p$ and value normal order entail standard
reduction sequences (Def.~\ref{def:v-standard}) and are therefore complete
(this is shown for $\rel{\Gamma}^p$ in \cite[p.11]{RP04}).

The $\lambda_{\betaV}^*$ calculus of \cite[Def.~11]{EHR91,EHR92} is a calculus
with partial terms. There is a unique constant $\Omega$ that represents
`bottom'. The calculus has reduction rules $M\,\Omega \rel{} \Omega$ and
$\Omega\,M \rel{} \Omega$ which capture preservation of unsolvability by
application (Section~\ref{sec:effective-use-lamK}). In \cite[p.508]{Wad76} we
find conversion rules $\Omega\,M = \Omega$ and $\lambda x.\Omega = \Omega$ now
in the context of $\lamK$. In both approaches $\Omega$ is uniquely used as
`bottom'. However, we have considered infinite bottoms with different orders,
and have followed in Section~\ref{sec:consistent-lamV-theory} the syntactic
approach of \cite{Bar84} where $\OMEGA$ is a term (not a constant representing
`bottom') and $M \rel{} \OMEGA$ when $M$ unsolvable. The $\OMEGA_n$ of
Section~\ref{sec:consistent-lamV-theory} are terms.

The computational lambda calculus of \cite{Mog91} adds the equations $\I\,X =
X$ and $(\lambda x.y\,x)X = y\,X$, for all $X\in\Lambda$, as axioms to the
proof-theory. These equations affect sequentiality
(Section~\ref{sec:neutrals-seq}).

The occurrence of a free variable can be seen as the result of implicitly
applying the `opening' operation to a locally-nameless representation of a
program (a closed term) \cite{Cha12}. In the local scope operational
equivalence is refined by considering open and non-closing contexts
(Section~\ref{sec:open-open}) that disclose the differences in
sequentiality. After that, the program can be recovered by the `closing'
operation.

The Genericity Lemma (Section~\ref{sec:effective-use-lamK}) conforms with the
axiomatic framework for meaningless terms of \cite{KOV99}. The axioms for
$\lamK$ state that meaningless terms are closed under reduction and
substitution (Axioms 1 and 3) and that if $M$ is meaningless then $M\,N$ is
meaningless, \ie\ $M$ cannot be used as a function (Axiom 2). For $\lamK$,
Axioms~1, 2, and 3 are enough to prove the Genericity Lemma and the
consistency of the proof-theory extended with equations between all
meaningless terms.

However, in $\lamV$ there is partiality in meaninglessness, \ie\ not all
meaningless terms are bottom. The analogues of the axioms have to be
order-aware. In particular, Lemma~\ref{lem:subst-pres-order-lamV} is the
order-aware analogue of Axiom~3. The analogue of Axiom~1 is trivial, just
consider $\beqV$. As for Axiom~2, if $M$ \mbox{$\lamV$-unsolvable} of order
$n$, then $M\,N$ (with $N\in \Val$) is \mbox{$\lamV$-unsolvable} of order
$n-1$. However, if $N\not\in\Val$, then $M\,N$ is $\lamV$-unsolvable of order
$0$. We leave the proof of the analogy as future work.

%%%%%%%%%%%%%%%%%%%%%%%%%%%%%%%%%%%%%%%%%%%%%%%%%%%%%%%%%%%%%%%%%%%%%%%%%%%%%%
\section{Conclusions and future work}
\label{sec:conclusion-future-work}
The presupposition of $v$-solvability (Section~\ref{sec:v-solv}) is that terms
with \betaVnf\ that are not transformable to a value of choice (such as $\B$
and $\U$) are observationally equivalent to terms without \betaVnf\ that are
also not transformable to a value of choice (such as $\OMEGA$ and $\lambda
x.\OMEGA$), and that all of them are operationally irrelevant and meaningless.
This gives rise to an inconsistent $\lamV$-theory. We have shown that these
terms can be separated operationally and that this conforms to $\lamV$'s
nature. Neutral terms differ at the point of potential divergence, \ie\ at the
blocking variable which has to be given the opportunity to be substituted by
an arbitrary value according to $\lamV$'s principle of `preserving confluence
by preserving potential divergence' (Section~\ref{sec:pure-lamV}). The actual
choice of values for blocking variables lets us separate terms with the same
functional result that nonetheless have different sequentiality, or may have
different sequentiality when using a different complete reduction
strategy. The functional models of $\lamV$ do not have such separating
capabilities, but functional models are not the only possible models. We have
to follow the other line of investigation, namely, to `vary the model to fit
the intended calculus'. Models that capture sequentiality exist, and we
believe there are \mbox{$\omega$-sensible} models that resemble the sequential
algorithms of \cite{BC82} (Section~\ref{sec:lamV-solv}).

As discussed in Section~\ref{sec:complete-standard}, standard reduction
sequences fall short of capturing all complete strategies of $\lamV$. A result
analogous to $\lamK$'s `quasi-needed reduction is normalising'
\cite[p.208]{BKKS87} is missing for $\lamV$. We are currently developing the
analogue for $\lamV$ of quasi-needed reduction, and the proof that it is
normalising.

As discussed in Section~\ref{sec:operational-characterisation}, we believe it
is possible to give an operational characterisation of $\lamV$-solvability,
\ie\ a reduction strategy of $\lamV$ that terminates iff the input term is
\mbox{$\lamV$-solvable}. But we believe it cannot be compositional because it
requires non-local information about the shape of the term to decide which is
the next $\betaV$-redex. We have a preliminary implementation that uses a
mark-test-and-contract algorithm. Terms with positive polarity are tested for
transformability and terms with negative polarity are tested for
valuability. In order to test we keep a sort of stratified environment that
references the operands in the nested accumulators of a \chnf. The environment
grows as reduction proceeds inside the body of nested blocks, where a table of
lexical offsets defines what is visible at each layer. The $\betaV$-redexes
are marked for contraction, but are only contracted after testing the
$\lamV$-solvability of the subterm in which they occur.

Our implementation can be refined using the `linear blocking structure' of the
sequent term calculus \cite{Her95,CH00,San07}.  The blocking structure of
{\chnf}s (\ie\ the structure of nested blocks around the blocking variable)
becomes a linear structure when injecting the {\chnf}s into their sequent-term
representation.  The sequent-term representation seems promising to develop
the analogue of Böhm trees in $\lamV$. Let us illustrate this by adopting the
untyped lambda-Gentzen calculus of \cite{San07} ($\lamGtz$ for short). Assume
the injection $\widehat{\ }:\CHNF \to \LamGtz$ and consider the shape of a
\chnf\ from Section~\ref{sec:value-normal-order}:
\begin{displaymath}
  % \begin{array}{l}
  %   \lambda x_1\ldots x_n.\\
  %   \quad(\lambda y_p.B_p)\\
  %   \quad\quad(~\cdots((\lambda y_1.B_1)((z\,W_0^0)W_1^0\cdots W_{m_0}^0)
  %   W_1^1\cdots W_{m_1}^1)\\
  %   \quad\quad\;\;\cdots~)
  %   W_1^p\cdots W_{m_p}^p
  % \end{array}
  \lambda x_1\ldots x_n.(\lambda y_p.B_p)
  (~\ldots((\lambda y_1.B_1)((z\,W_0^0)W_1^0\cdots W_{m_0}^0)
  W_1^1\cdots W_{m_1}^1)\ldots~) W_1^p\cdots W_{m_p}^p
\end{displaymath}
This shape is injected into the sequent term:
%\begin{displaymath}
  % \begin{array}{l}
  %   \lambda x_1\ldots x_n.\\
  %   \quad(z[\widehat{W^0_0}])[\widehat{W^0_1},\ldots,\widehat{W^0_{m_0}}]@\\
  %   \quad(y_1)(\widehat{B_1}[\widehat{W^1_1},\ldots,\widehat{W^1_{m_1}}]@\\
  %   \quad\quad\quad\quad\ldots\\
  %   \quad(y_p)(\widehat{B_p}[\widehat{W^p_1},\ldots,\widehat{W^p_{m_p}}])\ldots)
  % \end{array}
%  \lambda x_1\ldots x_n.
%  (z[\widehat{W^0_0}])[\widehat{W^0_1},\ldots,\widehat{W^0_{m_0}}]@
%  (y_1)(\widehat{B_1}[\widehat{W^1_1},\ldots,\widehat{W^1_{m_1}}]@(y_2)(\ldots
%  (y_p)(\widehat{B_p}[\widehat{W^p_1},\ldots,\widehat{W^p_{m_p}}])\ldots))
%\end{displaymath}
\[\eqalign{
  \lambda x_1\ldots x_n.
  (z[\widehat{W^0_0}])[\widehat{W^0_1},\ldots,\widehat{W^0_{m_0}}]
  @(y_1)(&\widehat{B_1}[\widehat{W^1_1},\ldots,\widehat{W^1_{m_1}}]\cr
&@(y_2)(\ldots
  (y_p)(\widehat{B_p}[\widehat{W^p_1},\ldots,\widehat{W^p_{m_p}}])\ldots))
  }
\]
The $\lamGtz$ representation reflects the blocking structure of the nested
blocks and accumulators in linear fashion, where the blocking variable $z$
appears in the leftmost position, and each accumulator in the trailing context
`unblocks' the subsequent accumulator.
% It is also possible to use the $\lamMumu$ representation of \cite{CH00}. In
% $\lamMumu$ a \hnf\ reads $\lambda x_1\ldots x_n.\mu\alpha.\langle y~|~
% N_1\cdot\ldots\cdot N_m \cdot \alpha\rangle$, and a chest normal form reads
% \begin{displaymath}
%   \begin{array}{l}
%     \lambda x_1\ldots x_n.\mu\alpha.\\
%     \quad\langle\mu\beta.\langle z~|~W^0_0\cdot\beta\rangle
%     ~|~W^0_1\cdot\ldots\cdot W^0_{m_0} \cdot\\
%     \quad\overline{\mu} y_1.
%     \langle B_1~|~W^1_1\cdot\ldots\cdot W^1_{m_1}\cdot\\
%     \quad\quad\quad\quad\ldots\\
%     \quad\overline{\mu} y_p.
%     \langle B_p~|~W^p_1\cdot\ldots\cdot W^p_{m_p}\cdot
%     \alpha\rangle\ldots\rangle\rangle
%   \end{array}
% \end{displaymath}
% The $\lamGtz$ representation reduces some of the clutter of the $\lamMumu$
% representation ($\mu$ operators, angles, etc.) and therefore we prefer the
% former over the latter.

\section*{Acknowledgement}
A preliminary version of this work was presented at the 11th International
Workshop on Domain Theory and Applications (Domains XI), Paris, 10th September
2014. We are grateful to Thomas Ehrhard and Alberto Carraro for their
insightful comments during the workshop. We also thank Beniamino Accattoli and
Flavien Breuvart for the stimulating discussions in the early stages of this
work. Our gratitude to Luca Aceto at the Icelandic Centre of Excellence in
Theoretical Computer Science, and to Manuel Carro and Manuel Hermenegildo at
the IMDEA Software Institute of Madrid, for providing an excellent supporting
environment. We are also indebted to Luca for his thoughtful review and
encouragement.

%%%%%%%%%%%%%%%%%%%%%%%%%%%%%%%%%%%%%%%%%%%%%%%%%%%%%%%%%%%%%%%%%%%%%%%%%%%%%%
\bibliographystyle{alpha}
\bibliography{lv-solvability}
%%%%%%%%%%%%%%%%%%%%%%%%%%%%%%%%%%%%%%%%%%%%%%%%%%%%%%%%%%%%%%%%%%%%%%%%%%%%%%
\newpage
\appendix

\section{Full glossary of terms and sets of terms}
\label{app:full-gloss}
\begin{center}
\begin{tabular}{llll}
  Set &&& Description\\
  \hline\hline
  $\Lambda$ & $::=$ & $x~|~\lambda x.\Lambda~|\ \Lambda\,\Lambda$ &
  lambda terms \\
  $\Val$ & $::=$ & $x~|~\lambda x.\Lambda$ & values \\
  $\Neu$ & $::=$ & $x\,\Lambda\,\{\Lambda\}^*$ & $\lamK$ neutrals \\
  $\NF$ & $::=$ & $\lambda x.\NF\ |\ x\,\{\NF\}^*$ & {\betaKnf}s \\
  $\HNF$ & $::=$ & $\lambda x.\HNF~|~x\,\{\Lambda\}^*$ & {\hnf}s \\
  $\Neu\V$ & $::=$ & $\Neu~|~\Block\,\{\Lambda\}^*$ & $\lamV$ neutrals \\
  $\Block$ & $::=$ & $(\lambda x.\Lambda)\Neu\V$ & blocks \\
  $\VNF$  & $::=$ & $x\ |\ \lambda x.\V\NF\ |\ \Stuck$ & {\betaVnf}s \\
  $\Stuck$ & $::=$ & $x\,\VNF\,\{\VNF\}^*$ & stucks \\
  & $|$ & $\Block\NF\,\{\VNF\}^*$ & \\
  $\Block\NF$ & $::=$ & $(\lambda x.\VNF)\,\Stuck$ & blocks in \betaVnf \\
  $\CHNF$ & $::=$ & $x~|~ \lambda x.\CHNF~|~ \NeuW$ & {\chnf}s\\
  $\VWNF$ & $::=$ & $\Val~|~ \NeuW$ & {\vwnf}s\\
  $\NeuW$ & $::=$ & $x\,\VWNF\,\{\VWNF\}^*$ & neutral {\vwnf}s \\
  & $|$ & $(\lambda x.\Lambda)\,\NeuW\,\{\VWNF\}^*$ &
\end{tabular}

\begin{tabular}{llll}
  Abbreviation & Term & has \betaKnf & has \betaVnf\\
  \hline\hline
  $\I$ & $\lambda x.x$ & yes & yes \\
  $\K$ & $\lambda x.\lambda y.x$ & yes & yes \\
  $\K^m$ & $\lambda x.\K_1(\cdots(\K_m\,x)\cdots)$ & yes & yes \\
  $\DELTA$ & $\lambda x.xx$ & yes & yes \\
  $\OMEGA$ & $\DELTA\DELTA$ & no & no \\
  $\U$ & $\lambda x.\B$ & no  & yes\\
  $\B$ & $(\lambda y.\DELTA)(x\,\I)\DELTA$ & no  & yes\\
  $\T_1$ & $(\lambda y.\DELTA)(x\,\I)\DELTA(x(\lambda x.\OMEGA)$ & no & no \\
  $\T_2$ & $(\lambda y.\DELTA)(x\,\I)\DELTA(\lambda x.\OMEGA)$ & no & no \\
  $\OMEGA_n$ & $\lambda x_1\ldots x_n.\DELTA\DELTA$ & no & no \\
  $\Y$ & $\lambda f.(\lambda x.f(x\, x))(\lambda x.f(x\,x)))$ & no & no \\
  $\Omom$ & $\Y\,\K$ & no & no
\end{tabular}
\end{center}
%%%%%%%%%%%%%%%%%%%%%%%%%%%%%%%%%%%%%%%%%%%%%%%%%%%%%%%%%%%%%%%%%%%%%%%%%%%%%%
\section{Proof of Thm.~\ref{thm:equiv-solvs-open} on
    page~\pageref{thm:equiv-solvs-open} and example}
\label{app:open-open}
\begin{proof}
  From \textsc{SolH} we prove \textsc{SolF} immediately because function
  contexts subsume head contexts and therefore \textsc{SolF} subsumes
  \textsc{SolH}.

  From \textsc{SolF} with the function context $\ctx{F}\hole$ we prove
  \textsc{SolH} by closing the function context to produce a head context
  $\ctx{H}\hole$.

  $\ctx{F}\hole$ is of the form $(\lambda x_1\ldots x_n.\hole)N_1\cdots N_k$,
  with $n\geq 0$, $k\geq 0$, and $N_i\in\Lambda$. Let $\{y_1,\ldots,y_m\}$
  (with $m \geq 0$) be the free variables in $N$, and
  $\{y_1,\ldots,y_m,y_{m+1},\ldots,y_{m+p}\}$ (with $p \geq 0$) be the free
  variables in $\ctx{F}[M]$. Since the $\{y_{m+1},\ldots,y_{m+p}\}$ do not
  occur in $N$, they are eventually discarded in the conversion to $N$ and can
  therefore be substituted by arbitrary closed terms without violating
  \textsc{SolF}.

  We focus on the $\{y_1,\ldots,y_m\}$. The \betaKnf\ $N$ is of the form
  $\lambda z_1\ldots z_q.h\,M_1 \cdots M_r$ with $q\geq 0$, $r\geq 0$, $h$ the
  head variable, and $M_1 \in \NF$, \ldots, $M_r \in \NF$. Since the $M_1$,
  \ldots, $M_r$ are {\betaKnf}s, all the variables in $N$ are head variables
  of some \betaKnf\ subterm, and so are the free variables
  $\{y_1,\ldots,y_m\}$. Let $\{P_{i1},\ldots, P_{is_i}\}$ (with $i\in
  \{1,\ldots,m\}$ and $s_i \geq 0$) be the maximal subterms that are in
  \betaKnf\ and that have a particular occurrence of the free variable $y_i$
  as the head variable.\footnote{Here `maximal' is used as in Def.~2.3 of
    {\cite{BKKS87}}, \ie\ it refers to the subterm ordering. However, notice
    that different subterms with different particular occurrences of the same
    variable $y_i$ as the head variable may not be disjoint. Consider the term
    $\lambda x_1.\lambda x_2.y_1(y_1\,\I\,\I)\I\,\I$. Both
    $P_{11}\equiv\lambda x_1.\lambda x_2.y_1(y_1\,\I\,\I)\I\,\I$ and
    $P_{12}\equiv y_1\,\I\,\I$ are maximal subterms with \emph{each of the two
      particular occurrences} of $y_1$ as head variable.} The $\{P_{11},\ldots,
  P_{1s_1},\ldots,P_{m1},\ldots,P_{ms_m}\}$ need not be disjoint. For each
  $i\in\{1,\ldots,m\}$ let $o_i$ be the maximum number of operands of $y_i$ in
  any \betaKnf\ subterm $P_{ij}$ having $y_i$ as the head variable:
  \begin{displaymath}
    o_i=\max\{\ell_j~|~P_{ij}\equiv
    \lambda u_1\ldots u_t.y_i\,Q_1\cdots Q_{\ell_j}~
    \text{with}~j\leq s_i\}
  \end{displaymath}
  We let $T_i \equiv \lambda v_1\ldots v_{o_i}\,w.w\,v_1 \cdots v_{o_i}$. The
  $y_1$, \ldots, $y_m$ can be replaced by the respective $T_1$, \ldots, $T_m$
  without violating \textsc{SolH}, since for any $i \leq m$ and $j\leq s_i$ we
  have
  \begin{displaymath}
    % \begin{array}{rcl}
    %   &&\cas{T_i}{y_i}{P_{ij}}\\
    %   &\equiv&\lambda u_1\ldots u_t.(\lambda v_1\ldots v_{o_i}\,w.w\,v_1
    %   \cdots v_{o_i}) Q'_1\cdots Q'_{\ell_j}\\
    %   &\beqK&\lambda u_1\ldots u_t\,v_{\ell_j+1}\ldots v_{o_i}\,w.
    %   w\,Q'_1\cdots Q'_{\ell_j}\,v_{\ell_j+1}\cdots v_{o_i}
    % \end{array}
    \begin{array}{rcl}
      \cas{T_i}{y_i}{P_{ij}}
      &\equiv&\lambda u_1\ldots u_t.(\lambda v_1\ldots v_{o_i}\,w.w\,v_1
      \cdots v_{o_i}) Q'_1\cdots Q'_{\ell_j}\\
      &\beqK&\lambda u_1\ldots u_t\,v_{\ell_j+1}\ldots v_{o_i}\,w.
      w\,Q'_1\cdots Q'_{\ell_j}\,v_{\ell_j+1}\cdots v_{o_i}
    \end{array}
  \end{displaymath}
  where $Q'_c\equiv\cas{T_i}{y_i}{Q_c}$ (with $c\leq \ell_j$). The term
  obtained is a \hnf\ in which the free variable $y_i$ no longer occurs, the
  closed $w$ is now the head variable, and there are additional binding
  occurrences $\lambda v_{\ell_j}\ldots v_{o_i}\,w$ and trailing closed
  operands $v_{\ell_j}\cdots v_{o_i}$. The term obtained can be proved to be a
  \betaKnf\ by a straightforward induction, since
  $Q'_c\equiv\cas{T_i}{y_i}{Q_c}$ (with $c\leq \ell_j$), and the $Q_c$ are
  subterms of $P_{ij}$. Consequently, the term
  $\cas{T_1}{y_1}{\ldots\cas{T_m}{y_m}{N}}$ is a closed \betaKnf. (Notice that
  although the $P_{ij}$ may not be disjoint, the substitutions
  $\cas{T_1}{y_1}{\ldots\cas{T_m}{y_m}{}}$ commute by the Substitution Lemma
  \cite[Lemma~2.1.16]{Bar84} because the $\{T_1,\ldots,T_m\}$ are closed
  terms, \ie\ the $\{y_1,\ldots,y_m\}$ do not occur free in them.)

  The head context $\ctx{H}\hole$ we are looking for is
  \begin{displaymath}
    \begin{array}{rcl}
      \ctx{H}\hole &\equiv& (\lambda y_1\ldots y_m\,y_{m+1}\ldots y_{m+p}\,
      x_1\ldots x_n.\hole)\\
      &&\quad T_1\cdots T_m\,\I_1 \cdots \I_p\\
      &&\quad\cas{T_1}{y_1}{\ldots\cas{T_m}{y_m}{\cas{\I}{y_{m+1}}
          {\ldots\cas{\I}{y_{m+p}}{N_1}}}}\\
      &&\quad\ldots\\
      &&\quad\cas{T_1}{y_1}{\ldots\cas{T_m}{y_m}{\cas{\I}{y_{m+1}}
          {\ldots\cas{\I}{y_{m+p}}{N_k}}}}
    \end{array}
  \end{displaymath}
  where the $y_1$, \ldots, $y_m$ are substituted respectively by $T_i$,
  \ldots, $T_m$ and the operationally irrelevant $\{y_{m+1},\ldots,y_{m+p}\}$
  are substituted by a closed term (we pick $\I$ but any other closed term
  would do). By the Substitution Lemma \cite[Lemma~2.1.16]{Bar84} the equation
  $\ctx{H}[M] \beqK \cas{T_1}{y_1}{\ldots\cas{T_m}{y_m}{N}}$ holds.
\end{proof}
The next example illustrates the proof of Thm.~\ref{thm:equiv-solvs-open}
by constructing a solving head context from the solving function context
$(\lambda x.\hole)\K$ that solves the term $M \equiv
x(y\,z(y\,\I))(t\,\OMEGA)$.
\begin{exa}
  The term $x(y\,z\,(y\,\I))(\OMEGA\,t)$ is solved by the function context
  $(\lambda x.\hole)\K$, \ie\ $(\lambda x.[x(y\,z\,(y\,\I))(\OMEGA\,t)])\K
  \beqK y\,z(y\,\I)$ where $y\,z(y\,\I)$ is an open term in \betaKnf. The free
  variables of the RHS of the equation are $\{y,z\}$, and the free variables
  of the LHS are $\{y,z,t\}$.

  The maximal subterms in \betaKnf\ having $y$ as the head variable are
  $y\,z(y\,\I)$ and $y\,\I$. The maximum number of operands to which the $y$
  is applied is $o_y=2$ (\ie\ the $z$ and the $y\,\I$ in $y\,z(y\,\I)$). The
  maximal subterm in \betaKnf\ having $z$ as the head variable is $z$, with
  $o_z=0$. Let $T_y\equiv\lambda v_1v_2w.w\,v_1\,v_2$ and $T_z\equiv\lambda
  w.w$. The solving head context is
  \begin{displaymath}
    % \begin{array}{rcl}
    %   &&\begin{array}{l}
    %     (\lambda yztx.\hole)(\lambda v_1v_2w.w\,v_1\,v_2)(\lambda
    %     w.w)\I\\
    %     \quad\quad
    %     (\cas{\lambda v_1v_2w.w\,v_1\,v_2}{y}
    %     {\cas{\lambda w.w}{z}{\cas{\I}{t}{\K}}})
    %   \end{array}\\
    %   &\equiv&
    %   \begin{array}{l}
    %   (\lambda yztx.\hole)(\lambda v_1v_2w.w\,v_1\,v_2)(\lambda
    %   w.w)\I\,\K
    %   \end{array}
    % \end{array}
    \begin{array}{rcl}
      &&(\lambda yztx.\hole)(\lambda v_1v_2w.w\,v_1\,v_2)(\lambda w.w)\I
      (\cas{\lambda v_1v_2w.w\,v_1\,v_2}{y}
      {\cas{\lambda w.w}{z}{\cas{\I}{t}{\K}}})\\
      &\equiv&
      (\lambda yztx.\hole)(\lambda v_1v_2w.w\,v_1\,v_2)(\lambda
      w.w)\I\,\K
    \end{array}
  \end{displaymath}
  Let us show that it solves the term:
  \begin{displaymath}
    % \begin{array}{rcl}
    %   &&\begin{array}{l}
    %     (\lambda yztx.[x(y\,z(y\,\I))(\OMEGA\,t)])\\
    %     \quad(\lambda v_1v_2w.w\,v_1\,v_2)(\lambda w.w)\I\,\K
    %   \end{array}\\
    %   &\beqK&\textit{\small \{substitute $y$\}}\\
    %   &&\begin{array}{l}
    %     (\lambda ztx.\\
    %     \quad[x((\lambda v_1v_2w.w\,v_1\,v_2)z((\lambda
    %     v_1v_2w.w\,v_1\,v_2)\I))\\
    %     \quad\quad(\OMEGA\,t)])(\lambda w.w)\I\,\K\\
    %   \end{array}\\
    %   &\beqK&\textit{\small \{substitute rightmost $v_1$\}}\\
    %   &&\begin{array}{l}
    %     (\lambda ztx.\\
    %     \quad[x((\lambda v_1v_2w.w\,v_1\,v_2)z(\lambda
    %     v_2w.w\,\I\,v_2))\\
    %     \quad\quad(\OMEGA\,t)])(\lambda w.w)\I\,\K\\
    %   \end{array}\\
    %   &\beqK&\textit{\small \{substitute $z$\}}\\
    %   &&\begin{array}{l}
    %     (\lambda tx.\\
    %     \quad[x((\lambda v_1v_2w.w\,v_1\,v_2)(\lambda w.w)(\lambda
    %     v_2w.w\,\I\,v_2))\\
    %     \quad\quad(\OMEGA\,t)])\I\,\K\\
    %   \end{array}\\
    %   &\beqK&\textit{\small \{substitute $v_1$ and leftmost $v_2$\}}\\
    %   &&\begin{array}{l}
    %     (\lambda tx.\\
    %     \quad[x(\lambda w.w(\lambda w.w)(\lambda v_2w.w\,\I\,v_2))
    %     (\OMEGA\,t)])\I\K\\
    %   \end{array}\\
    %   &\beqK&\textit{\small \{substitute $t$\}}\\
    %   &&\begin{array}{l}
    %     (\lambda x.\\
    %     \quad[x(\lambda w.w(\lambda w.w)(\lambda v_2w.w\,\I\,v_2))
    %     (\OMEGA\,\I)])\K\\
    %   \end{array}\\
    %   &\beqK&\textit{\small \{substitute $x$ and convert constant operator\}}\\
    %   &&\lambda w.w(\lambda w.w)(\lambda v_2w.w\,\I\,v_2) \in \NF^0
    % \end{array}
    \begin{array}{rcl}
      &&(\lambda yztx.[x(y\,z(y\,\I))(\OMEGA\,t)])
      (\lambda v_1v_2w.w\,v_1\,v_2)(\lambda w.w)\I\,\K\\
      &\beqK&\textit{\small \{substitute $y$\}}\\
      &&(\lambda ztx.[x((\lambda v_1v_2w.w\,v_1\,v_2)z((\lambda
      v_1v_2w.w\,v_1\,v_2)\I))
      (\OMEGA\,t)])(\lambda w.w)\I\,\K\\
      &\beqK&\textit{\small \{substitute rightmost $v_1$\}}\\
      &&(\lambda ztx.[x((\lambda v_1v_2w.w\,v_1\,v_2)z(\lambda
      v_2w.w\,\I\,v_2))(\OMEGA\,t)])(\lambda w.w)\I\,\K\\
      &\beqK&\textit{\small \{substitute $z$\}}\\
      &&(\lambda tx.[x((\lambda v_1v_2w.w\,v_1\,v_2)(\lambda w.w)(\lambda
      v_2w.w\,\I\,v_2))(\OMEGA\,t)])\I\,\K\\
      &\beqK&\textit{\small \{substitute $v_1$ and leftmost $v_2$\}}\\
      &&(\lambda tx.[x(\lambda w.w(\lambda w.w)(\lambda v_2w.w\,\I\,v_2))
      (\OMEGA\,t)])\I\K\\
      &\beqK&\textit{\small \{substitute $t$\}}\\
      &&(\lambda x.[x(\lambda w.w(\lambda w.w)(\lambda v_2w.w\,\I\,v_2))
      (\OMEGA\,\I)])\K\\
      &\beqK&\textit{\small \{substitute $x$ and convert constant operator\}}\\
      &&\lambda w.w(\lambda w.w)(\lambda v_2w.w\,\I\,v_2) \in \NF^0
    \end{array}
  \end{displaymath}
\end{exa}

%%%%%%%%%%%%%%%%%%%%%%%%%%%%%%%%%%%%%%%%%%%%%%%%%%%%%%%%%%%%%%%%%%%%%%%%%%%%%%
\section{Genericity Lemma in \texorpdfstring{\cite{Wad76}}{Wad76} and \texorpdfstring{\cite{Bar84}}{Bar84}}
\label{app:effective-use-lamK}
Our statement of the Genericity Lemma (Lemma~\ref{lem:lamK-genericity-lemma}
on page~\pageref{lem:lamK-genericity-lemma}) is a combination of the following
versions. We state them using the term identifiers $M$ and $X$ of
Lemma~\ref{lem:lamK-genericity-lemma} for uniformity.

Corollary~5.5 on page 510 of \cite{Wad76}: {\itshape Suppose $M$ is unsolvable
  and $\ctx{C}\hole$ is any context. Then $\ctx{C}[M]$ has a normal form (a
  head normal form) iff $\ctx{C}[X]$ has the same normal form (a similar head
  normal form) for all terms $X$.}

Proposition~14.3.24 on page 374 of \cite{Bar84}: {\itshape Let $M,N\in\Lambda$
  with $M$ unsolvable and $N$ having a nf. Then for all}
$\ctx{C}\hole\in\Lambda$, $\ctx{C}[M] \beqK N \Rightarrow \forall
X\in\Lambda~~\ctx{C}[X] \beqK N$.

%%%%%%%%%%%%%%%%%%%%%%%%%%%%%%%%%%%%%%%%%%%%%%%%%%%%%%%%%%%%%%%%%%%%%%%%%%%%%%
\section{Values are required for substitutivity and confluence}
\label{app:pure-lamV}
Permitting applications in \betaVnf\ as members of $\Val$ breaks
confluence. Such applications would be permissible operands in the conversion
rule ($\betaV$).  The counter-example used in \cite[p.135-136]{Plo75} is
$(\lambda x.(\lambda y.z)(x\,\DELTA))\DELTA$. If the application in \betaVnf\
$(x\,\DELTA)$ is in $\Val$ then the conversion $(\lambda x.(\lambda
y.z)(x\,\DELTA))\DELTA \beqV (\lambda x.z)\DELTA$ would be allowed (the
innermost redex is converted). From that conversion $(\lambda
x.z)\DELTA \beqV z$ follows. However, $(\lambda x.(\lambda y.z)(x\,\DELTA))\DELTA
\beqV (\lambda y.z)(\DELTA\DELTA)$ is a valid conversion (the outermost redex
is converted), but $(\lambda y.z)(\DELTA\DELTA) \beqV z$ does not follow
because $\OMEGA\equiv\DELTA\DELTA$ is not an application in \betaVnf\ and it
cannot be converted to one.

Permitting arbitrary applications as subjects of substitutions breaks
substitutivity. The counter-example used in \cite[p.135-136]{Plo75} is to
consider $(\lambda x.\I)x \beqV \I$ and to show that
$\cas{\OMEGA}{x}{((\lambda x.\I)x)} \beqV \cas{\OMEGA}{x}{\I}$, that is,
$(\lambda x.\I)\OMEGA \beqV \I$, does not hold. The LHS has no \betaVnf\
because the diverging term $\OMEGA$ is converted before substitution whereas
the RHS is a \betaVnf.

An subtle point unstated in \cite[p.135-136]{Plo75} is that permitting
applications in \betaVnf\ as subjects of substitutions also breaks
substitutivity even if permissible operands were values. The counter-example
is to consider $(\lambda x.\I)x \beqV \I$ and to show that
$\cas{(x\,\DELTA)}{x}{((\lambda x.\I)x)} \beqV \cas{(x\,\DELTA)}{x}{\I}$, that
is, $(\lambda x.\I)(x\,\DELTA) \beqV \I$, does not hold. The LHS cannot
convert to the RHS because the operand $(x\,\DELTA)$ is not permissible.

As a consequence of the last two paragraphs, the substitutivity property in
$\lamV$ has the proviso $L\in\Val$ in its statement \cite[p.135]{Plo75}.

% \begin{lem}
%   \label{lem:red-lamV-unsolv}
%   Let $M\in\Lambda$ be a $\lamV$-unsolvable of order $n$ and $M'$ such that
%   $M\mrel{\betaV} M'$. $M'$ is $\lamV$-unsolvable of order $n$.
% \end{lem}
% \begin{proof}
%   Trivial, by the confluence of $\mrel{\betaV}$.
% \end{proof}

%%%%%%%%%%%%%%%%%%%%%%%%%%%%%%%%%%%%%%%%%%%%%%%%%%%%%%%%%%%%%%%%%%%%%%%%%%%%%%
\section{Head and head spine of a term}
\label{app:head-and-head-spine}
For ease of reference we collect here the results of \cite{BKKS87} relative to
the complete normal order strategy of $\lamK$ on which we base the analogue
results for $\lamV$ in Section~\ref{sec:value-normal-order}.

A redex of $M\in\Lambda$ is \emph{needed} \cite[p.212]{BKKS87} if the redex or
its residual is contracted in every reduction sequence starting in $M$ and
arriving at a \betaKnf. The contraction of a needed redex always decrements
the length of a normalising reduction sequence. Neededness is an undecidable
property, but there exist decidable approximations of the set of needed
redexes that can be computed efficiently. The so-called \emph{spine
  strategies} reduce redexes in several of these decidable approximations of
the needed set.

The \emph{head} and \emph{head spine} of a term \cite[Def.~4.2]{BKKS87}
provide progressively better approximations to the set of \emph{needed
  redexes} in the term \cite[p.212]{BKKS87}. The head is the segment of the
abstract syntax tree of the term that starts at the root node and descends
through lambda nodes and to the left through operators in applications. The
head spine is the segment of the abstract syntax tree that starts at the root
node and descends either through lambda nodes or to the left through operators
in applications. The head spine of a term includes the head of the term and,
recursively, the head of the innermost node reached so far.
Fig.~\ref{fig:head-head-spine} illustrates with an example that is further
developed after the following formal definition of head and head spine.

In Def.~\ref{def:head-head-spine} below we define the functions $\bn$, $\he$,
and $\hs$. The head spine of a term is underlined by function $\hs$ whose
definition we have taken from \cite[Def.~4.2]{BKKS87}. The head of a term is
underlined by function $\he$ that relies on the auxiliary function $\bn$ which
is related to call-by-name as explained below. The definition of $\he$ is
based on the definition of the head reduction strategy in \cite{Bar84} that
reduces up to \hnf. We define head reduction and call-by-name using a
reduction semantics in Def.~\ref{def:head-reduction} and
Def.~\ref{def:cbn-contexts}.
\begin{defi}[Head and head spine]
  \label{def:head-head-spine}
  Functions $\he$ and $\hs$ underline the head and the head spine of a term
  respectively.
  \begin{displaymath}
    \begin{array}{rcl}
      \bn(x)           &=& \underline{x}\\
      \bn(\lambda x.B) &=& \underline{\lambda x}.B\\
      \bn(M\,N)        &=& \bn(M)N\\[4pt]
      \he(x)           &=& \underline{x}\\
      \he(\lambda x.B) &=& \underline{\lambda x}.\he(B)\\
      \he(M\,N)        &=& \bn(M)N\\[4pt]
      \hs(x)           &=& \underline{x}\\
      \hs(\lambda x.B) &=& \underline{\lambda x}.\hs(B)\\
      \hs(M\,N)        &=& \hs(M)N
    \end{array}
  \end{displaymath}
  A $\beta$-redex is head (resp. head spine) if its outermost lambda is
  underlined by function $\he$ (resp. $\hs$).
\end{defi}
Function $\bn$ underlines the outermost lambda of the $\betaK$-redexes that
are reduced by the call-by-name strategy of pure $\lamK$
(Def.~\ref{def:cbn-contexts}). This strategy differs from its homonym in
\cite{Plo75} which is for an applied version of the calculus. See \cite{Ses02}
for details on the difference.

As an example, consider the term whose abstract syntax tree is depicted in
Fig~\ref{fig:head-head-spine}. The head (thick edges in the figure) is
underlined in $\underline{\lambda x.(\lambda y}.(\lambda z.x)M_1)x((\lambda
t.M_2)x)$.
% \begin{displaymath}
%   \underline{\lambda x.(\lambda y}.(\lambda z.x)M_1)x((\lambda t.M_2)x)
% \end{displaymath}
The head spine (thick edges and dotted edges) is underlined in
$\underline{\lambda x.(\lambda y.(\lambda z.x})M_1)x((\lambda t.M_2)x)$.
% \begin{displaymath}
%   \underline{\lambda x.(\lambda y.(\lambda z.x})M_1)x((\lambda t.M_2)x)
% \end{displaymath}
The subterm $(\lambda y.(\lambda z.x)M_1)x$ is both a head and a head spine
$\beta$-redex. The subterm $(\lambda z.x)M_1$ is a head spine $\beta$-redex.
The subterm $(\lambda t.M_2)x$ is neither a head nor a head spine $\beta$-redex.

\begin{figure}
  \begin{center}
    \begin{tikzpicture} [
      level distance=1cm,
      level 2/.style={sibling distance=2.5cm},
      level 3/.style={sibling distance=1.5cm},
      level 4/.style={sibling distance=1.5cm},
      level 5/.style={sibling distance=1.5cm},
      chest/.style={very thick},
      ribcage/.style={dotted},
      norm/.style={thin,solid}
      ]
      \begin{scope}
        \node (lx) {$\lambda x$}
        child[chest]{ node (a1) {$@$}
          child{ node (a2) {$@$}
            child{ node (ly) {$\lambda y$}
              child[ribcage]{ node (a3) {$@$}
                child{ node (lz) {$\lambda z$}
                  child{ node (x3) {$x$}}}
                child[norm]{ node (m1) {$M_1$}}}}
            child[norm]{ node (x1) {$x$} }}
          child[norm]{ node (a3) {$@$}
            child{ node (lt) {$\lambda t$}
              child{ node (m2) {$M_2$} }}
            child{ node (x2) {$x$}}}};
      \end{scope}
    \end{tikzpicture}
  \end{center}
  \caption{Head (thick edges) and head spine (thick edges and dotted edges) of
    the term $\lambda x.(\lambda y.(\lambda z.x)M_1)x((\lambda t.M_2)x)$.}
  \label{fig:head-head-spine}
\end{figure}

We now define call-by-name and head reduction using a reduction semantics.
Call-by-name is the leftmost strategy that never contracts under lambda
abstraction. Head reduction is the leftmost strategy that stops at a \hnf.
Observe that the reduction contexts of head reduction contain the reduction
contexts of call-by-name.

\begin{defi}[Call-by-name strategy]
\label{def:cbn-contexts}
The call-by-name strategy $\rel{\stgy{bn}}$ is defined by the following
reduction semantics:
  \begin{displaymath}
    \begin{array}{l}
      \ctx{BN}\hole\ ::=\ \hole~|~\ctx{BN}\hole\,\Lambda \\\\
      \ctx{BN}[(\lambda x.B)N]\rel{\stgy{bn}}\ctx{BN}[\cas{N}{x}{B}]
    \end{array}
  \end{displaymath}
\end{defi}
\begin{defi}[Head reduction strategy]
  \label{def:head-reduction}
  The head reduction strategy $\rel{\stgy{he}}$ is defined by the following
  reduction semantics:
  \begin{displaymath}
    \begin{array}{l}
      \ctx{HR}\hole\ ::=\
      \hole~|~\ctx{BN}\hole\,\Lambda~|\ \lambda x.\ctx{HR}\hole \\\\
      \ctx{HR}[(\lambda x.B)N]\rel{\stgy{he}}\ctx{HR}[\cas{N}{x}{B}]
      \end{array}
    \end{displaymath}
\end{defi}

%%%%%%%%%%%%%%%%%%%%%%%%%%%%%%%%%%%%%%%%%%%%%%%%%%%%%%%%%%%%%%%%%%%%%%%%%%%%%%
\end{document}